\documentclass[12pt,a4paper,notitlepage, reqno]{amsart}
\usepackage[utf8]{inputenc}
\usepackage[english]{babel}

\usepackage{bm,bbm}
\usepackage{amsbsy,amsmath,amsthm,amssymb,amsfonts}%
\usepackage{graphicx}
\usepackage{subfigure}
\usepackage{verbatim}
\usepackage{accents}
\usepackage{stmaryrd}

\newcommand{\ddiamond}{\spadesuit}

\usepackage[	pdfauthor={},%
			pdftitle={},%
			pdftex]{hyperref}
\hypersetup{	colorlinks,%
			citecolor=black,%
			filecolor=black,%
			linkcolor=black,%
			urlcolor=black%
			}

\usepackage[foot]{amsaddr}

\newcommand{\Z}{\mathbb{Z}}
\newcommand{\N}{\mathbb{N}}
\newcommand{\C}{\mathbb{C}}
\newcommand{\R}{\mathbb{R}}

\newcommand{\half}{\mathbb{H}}
\newcommand{\disc}{\mathbb{D}}

\renewcommand{\P}{\mathbb{P}}

\newcommand{\E}{\mathbb{E}}

\newcommand{\ind}{\mathbbm{1}}

\newcommand{\F}{\mathcal{F}}

\newcommand{\hausdorff}{\mathcal{H}}
\newcommand{\hdim}{\dim_\hausdorff}

\newcommand{\OO}{\mathcal{O}}
\newcommand{\oo}{o}
\newcommand{\de}{\mathrm{d}}

\newcommand{\sgn}{\,\mathrm{sgn}}
\DeclareMathOperator{\dist}{dist}
\DeclareMathOperator{\diam}{diam}

\newcommand{\eps}{\varepsilon}
\DeclareMathOperator{\imag}{Im}
\DeclareMathOperator{\real}{Re}

\newcommand{\dd}{\mathrel{\mathop:}}

\theoremstyle{plain}
\newtheorem{theorem}{Theorem}[section]
\newtheorem{lemma}[theorem]{Lemma}
\newtheorem{proposition}[theorem]{Proposition}

\newtheorem*{mainthmA}{Theorem~\ref{thm: main theorem} (a)}
\newtheorem*{mainthmB}{Theorem~\ref{thm: main theorem} (b)}

\theoremstyle{definition}
\newtheorem{definition}[theorem]{Definition}
\newtheorem{condition}{Condition}

\theoremstyle{remark}
\newtheorem{remark}[theorem]{Remark}

\usepackage[hmargin=3cm,vmargin={2.5cm,2cm},includefoot]{geometry}

\newcommand{\enustyo}{%
\renewcommand{\theenumi}{(\arabic{enumi})}
\renewcommand{\labelenumi}{\theenumi}
}

\newcommand{\enustyii}{%
\renewcommand{\theenumi}{(\roman{enumi})}
\renewcommand{\labelenumi}{\theenumi}
}

\newcommand{\enustyiii}{%
\renewcommand{\theenumi}{(\alph{enumi})}
\renewcommand{\labelenumi}{\theenumi}
}

\DeclareMathOperator{\Proj}{Proj}

\newcommand{\vroot}{v_\text{root}}
\newcommand{\tree}{\mathcal{T}}

\newcommand{\rt}{{\text{root}}}

\newcommand{\lattice}{\mathbb{L}}
\newcommand{\sqlattice}{\lattice^\bullet}
\newcommand{\dsqlattice}{\lattice^\circ}
\newcommand{\dmdlattice}{\lattice^\diamond}
\newcommand{\ddmdlattice}{\lattice^\ddiamond}

\newcommand{\orient}{\rightarrow}

\newcommand{\odmdlattice}{\dmdlattice_\orient}
\newcommand{\oddmdlattice}{\ddmdlattice_\orient}

\newcommand{\LoopEnseble}{\mathcal{L}}
\newcommand{\Loop}{L}
\newcommand{\RandomLoopEnseble}{\Theta}
\newcommand{\RandomLoop}{\theta}

\newcommand{\lpe}{\LoopEnseble}
\newcommand{\lp}{\Loop}
\newcommand{\rndlpe}{\RandomLoopEnseble}
\newcommand{\rndlp}{\RandomLoop}

\newcommand{\Tree}{\mathcal{T}}
\newcommand{\Branch}{T}
\newcommand{\RandomTree}{\Tree}
\newcommand{\RandomBranch}{\Branch}

\newcommand{\ttr}{\Tree}
\newcommand{\bran}{\Branch}
\newcommand{\rndttr}{\RandomTree}
\newcommand{\rndbran}{\RandomBranch}

\newcommand{\ttrtarget}{V_{\text{target}}}

\newcommand{\therndttr}{\RandomTree}

\newcommand{\therndlpe}{\Theta}

\newcommand{\delp}{\de_{\text{loop}}}

\newcommand{\itop}{\text{T}}
\newcommand{\ibot}{\text{B}}
\newcommand{\ileft}{\text{L}}
\newcommand{\iright}{\text{R}}

\newcommand{\pleft}{x_\ileft}
\newcommand{\pright}{x_\iright}

\newcommand{\domain}{\Omega}
\newcommand{\vardomain}{\domain}

\newcommand{\extarc}[2]{#1 \smile #2}

\newcommand{\intarcp}[4]{(#1 \frown #2, #3 \frown #4)}
\newcommand{\extarcp}[4]{(#1 \smile #2, #3 \smile #4)}

\newcommand{\todisc}[1]{#1_\disc}

\newcommand{\varconfmap}{\phi}

\newcommand{\ffrac}[2]{\left(\frac{#1}{#2}\right)}

\newcommand{\sle}[1]{SLE$(#1)${}}
\newcommand{\slevar}[1]{SLE$[#1]${}}
\newcommand{\varsle}[1]{\slevar{#1}}

\newcommand{\kr}{\sle{{\kappa,\rho}}}
\newcommand{\slek}{\sle{\kappa}}

\newcommand{\intpoint}{w_0}

\begin{document}

\title{Conformal invariance of boundary touching loops of FK Ising model}

\author{Antti Kemppainen$^1$}
\address{$^1$Department of Mathematics and Statistics\\
         P.O. Box 68\\
         FIN-00014 University of Helsinki\\
         Finland}
\author{Stanislav Smirnov$^2$}
\address{$^2$Section de math\'ematiques,
         Universit\'e de Gen\`eve,
         2-4, rue du Li\`evre, c.p. 64,
         1211 Gen\`eve 4, Switzerland, \textnormal{and} 
Chebyshev Laboratory, St.~Petersburg State University, Russia, \textnormal{and}
Skolkovo Institute of Science and Technology, Moscow, Russia}
\email{Antti.H.Kemppainen@helsinki.fi, Stanislav.Smirnov@unige.ch}

\begin{abstract}
In this article we show the convergence of a loop ensemble of interfaces in the FK Ising model at criticality,
as the lattice mesh tends to zero, to a unique conformally invariant scaling limit. 
The discrete loop ensemble is described by a canonical tree glued from the interfaces, which then is shown to converge to a tree of
branching SLEs.
The loop ensemble contains unboundedly many loops and hence our result describes the joint
law of infinitely many loops in terms of SLE type processes, and the result gives the full scaling limit
of the FK Ising model in the sense of random geometry of the interfaces.

Some other results in this article are convergence of the exploration process 
of the loop ensemble (or the branch of the exploration tree) to SLE$(\kappa,\kappa-6)$, $\kappa=16/3$,
and convergence of a generalization of this process for $4$ marked points
to SLE$[\kappa,Z]$, $\kappa=16/3$, where $Z$ refers to a partition function. The latter SLE process
is a process that can't be written as a SLE$(\kappa,\rho_1,\rho_2,\ldots)$ process, which are the most commonly considered
generalizations of SLEs.
\end{abstract}

\maketitle

\tableofcontents


\section{Introduction}

\subsection{The setup informally}

The \emph{Ising model} is one of the most studied lattice models in statistical physics. The Ising model (and Potts models generalizing it) have
percolation-type representations called \emph{Fortuin--Kasteleyn random cluster models} (FK model). 
Given a graph the Ising model assigns probabilities to configurations of $\pm 1$ spins on the vertices of the graph
and the random cluster model assigns probabilities to configurations of open/closed edges.
FK model is obtained from the percolation model (giving weight $p$ to each open edge and $1-p$ to each closed edge) 
by weighting by a factor $q$ per each connected cluster of open edges.
We will consider the FK Ising model, i.e., FK model whose parameter value $q$ corresponds to the Ising model, 
on a square lattice. The lattice we are considering is the usual square lattice $\Z^2$, for convenience rotated by $45^\circ$.
In Figure~\ref{sfig: graphs a} it is the lattice formed by the centers of the black squares.

\begin{figure}[tbh]
\centering
\subfigure[Lattices $\sqlattice, \dsqlattice, \dmdlattice$.] 
{
	\label{sfig: graphs a}
	\includegraphics[scale=.8]
{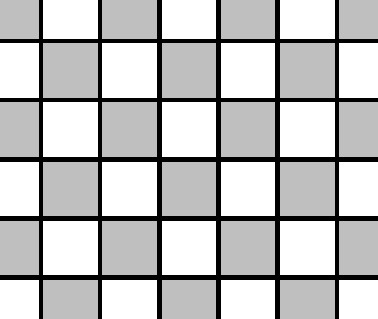}
} 
\hspace{0.5cm}
\subfigure[Modification and $\ddmdlattice$.]
{
	\label{sfig: graphs b}
	\includegraphics[scale=.8]
{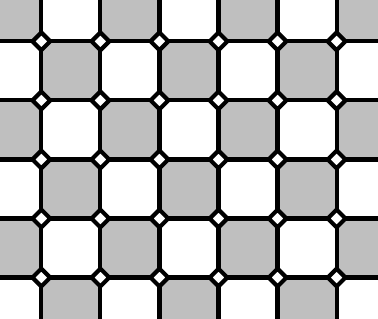}
}
\caption{%
The square lattices we are considering are $\sqlattice$ 
formed by the centers of the black squares 
(the vertices are connected by an edge 
if the corresponding squares touch by corners),
$\dsqlattice$ formed by the centers of the white squares 
(the edges similarly as for $\sqlattice$)
and $\dmdlattice$ formed by the vertices and edges of the black (or equivalently white) squares.
We will also consider the square--octagon lattice $\ddmdlattice$ 
(infinite graph formed by the vertices and edges of the squares and octagons
in the picture)
which we see as a modification of $\dmdlattice$.}
\label{fig: graphs}
\end{figure}

We will consider a bounded simply connected subgraph of the square lattice. We will make this definition clearer later.
An FK configuration in the graph is illustrated in \ref{sfig: loop representation a} and its dual configuration
in \ref{sfig: loop representation b}. 
The dual configuration is defined on the dual lattice, formed by the centers of the white squares in
Figure~\ref{sfig: graphs a}, with the rule that exactly one of any (primal) edge and its dual edge 
(the dual edge is the unique edge on the dual lattice crossing the primal edge)
is present in the configuration or in the dual configuration. 
The so called \emph{loop representation} of the FK model is defined
on so called \emph{medial lattice}, shown in Figure~\ref{sfig: graphs a} and
formed by the corners of the squares, 
by taking the inner and outer boundaries of all the connected components in a configuration of edges. The loop representation
can be seen as a dense collection of \emph{simple} loops when we resolve the possible double vertices by the modification
illustrated in Figure~\ref{sfig: graphs b} or Figure~\ref{fig: loop representation}.

\begin{figure}[tbh]
\centering
\subfigure[A configuration on $\sqlattice$ satisfying wired boundary conditions. 
Under the wired
boundary conditions, the edges next to the boundary of the domain are conditioned to belong
to the configuration.
] 
{
	\label{sfig: loop representation a}
	\includegraphics[scale=.375]
{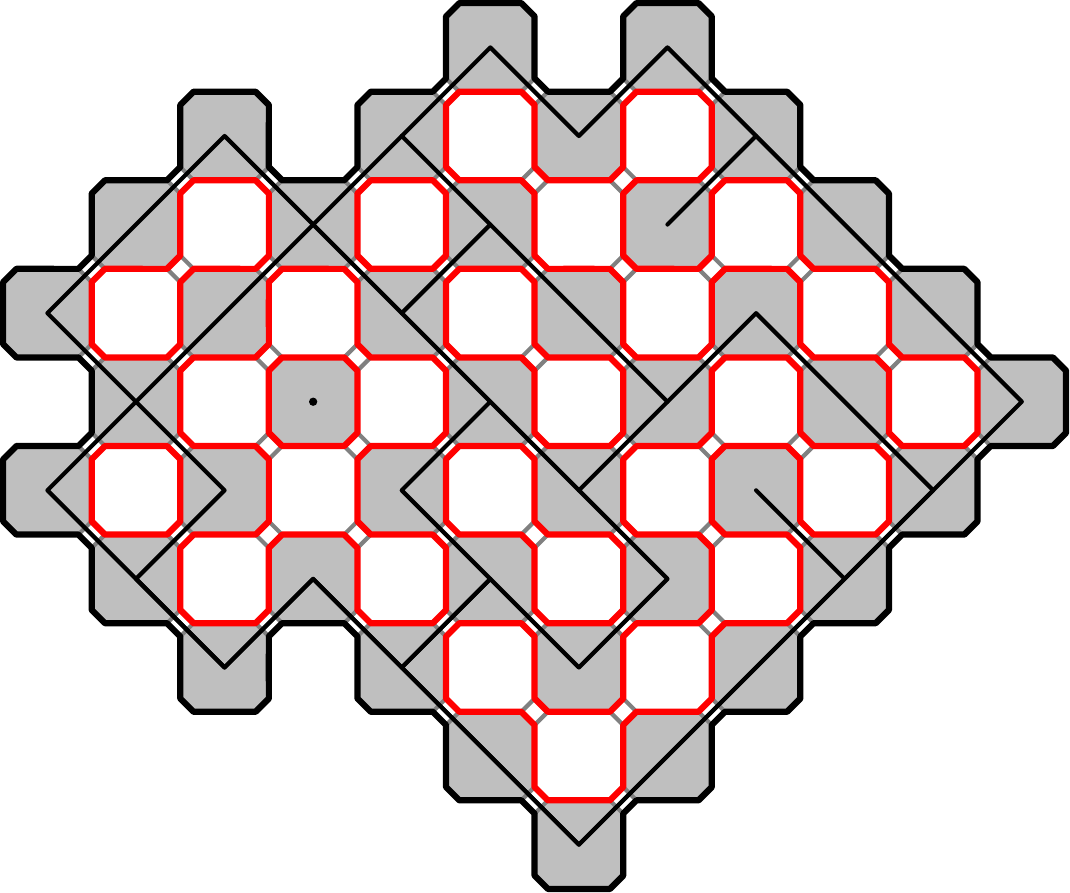}
} 
\hspace{0.5cm}
\subfigure[A configuration on $\dsqlattice$ satisfying free boundary conditions.
As the name suggest, the free
boundary conditions imply that the FK model is considered in its standard
form, without any added weight factors.
]
{
	\label{sfig: loop representation b}
	\includegraphics[scale=.375]
{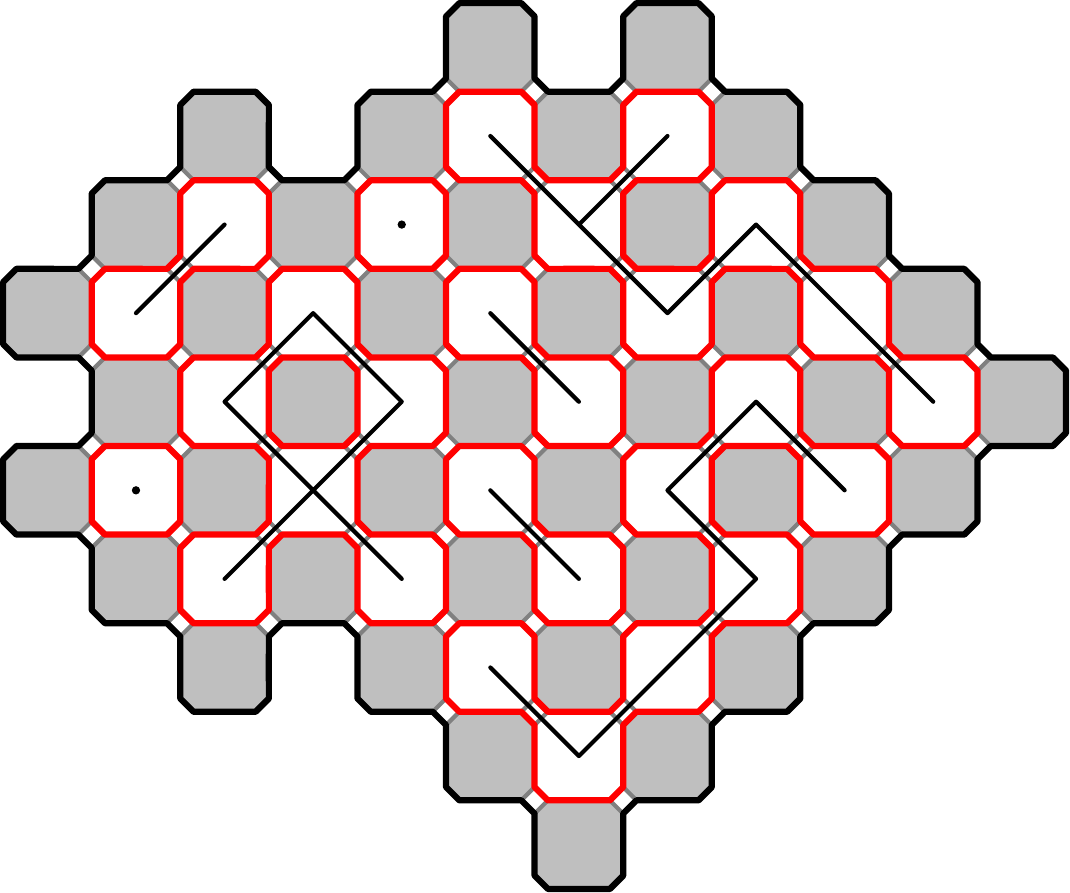}
}
\caption{FK configurations 
with wired and free boundary conditions
on graphs dual to each other.
Notice that these two configurations are dual to each other and thus
have the same loop representations.}
\label{fig: loop representation}
\end{figure}

The duality of the FK configurations described above gives a mapping between 
the set of configurations with wired boundary conditions
and the set of configurations with free boundary conditions, 
defined on the graph and on its dual, respectively. 
One can check that the FK model parameter values
$p$ and $q$ get mapped under this involution of FK configurations to values $p^*$ and $q^*=q$ 
where $p^*$ is given by
\begin{equation}
\frac{p \, p^*}{(1- p)(1- p^*)} = q .
\end{equation}
In this article we will consider critical point of the model which happens to be the self-dual value of $p$, that is,
$p=p_c$ when $p=p^*$. For the FK Ising model, the parameter $q=2$ and the \emph{critical parameter} $p_c= \frac{\sqrt{2}}{1+\sqrt{2}}$.
By the fact that it is the self-dual point we see that in Figures~\ref{sfig: loop representation a} and \ref{sfig: loop representation b}
the FK models have the same parameter values. The only difference is in boundary conditions. The resulting
loop configurations have the same law for the critical parameter on both setups, either wired boundary conditions on the primal graph
or free boundary conditions on the dual graph.

The correspondence of boundary conditions is slightly more complicated for other boundary conditions than the wired and free ones.

\subsubsection{The role of the critical parameter}

The parameter is chosen to be critical for several reasons. First of all it is expected that only for this value
of the parameter the scaling limit will be non-trivial. For the other values, we get either zero or one macroscopic loop
and all the other loops will be microscopic and vanishing in the scaling limit. 
The only macroscopic loop will be rather uninteresting
as it will follow closely the boundary, so that the loop will fluctuate from the boundary only to a distance which vanishes in the scaling limit. In contrast, for the critical parameter the scaling limit will consist of
(countably) infinite number of loops as 
we will be showing.

The second reason for selecting this
value of the parameter is more technical. For that value, the observable we are defining in Section~\ref{ssec: setup}
is going to satisfy a relation which we can interpret as a discrete version of the Cauchy--Riemann equations.
This makes it possible to pass to the limit and recover in the limit a holomorphic function solving a boundary value problem.

The third reason for the choice of the critical parameter is that at criticality we expect that the loop collection will
have a conformal symmetry. This is already suggested by the existence of the holomorphic observable.

In fact, the discrete holomorphicity and related techniques allow
us to control the scaling limit
well and to identify the scaling limit
and to show its conformal invariance.
And thus can be seen as the most important reason to consider a system at criticality.

\subsection{The exploration tree and the main results}

\subsubsection{The exploration tree}\label{sssec: exploration tree}

\begin{figure}[tbh]
\centering
	\includegraphics[scale=.6]
{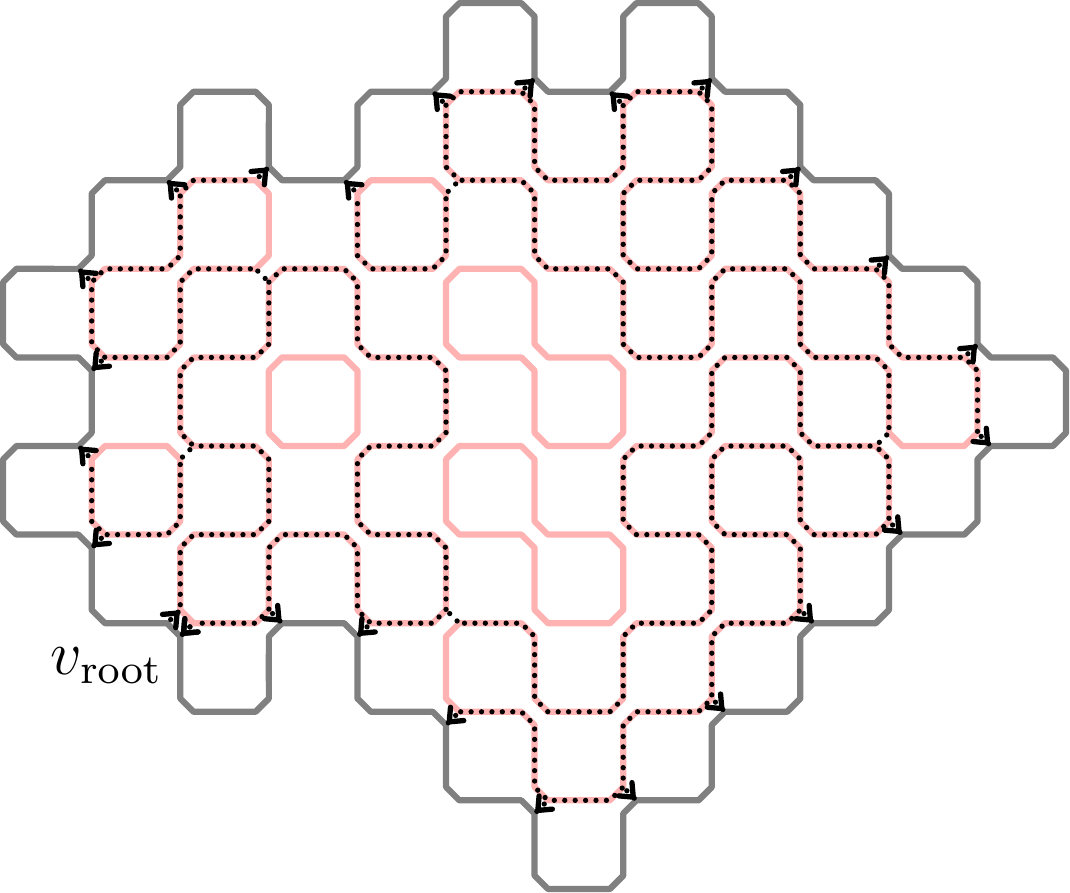}
\caption{The loop collection and the exploration tree: 
solid pink 
(including
solid pink on the background of dotted lines) indicate
the loops. The dotted black lines form the tree. The dotted black lines 
with no pink on background
are the jumps from one loop to another.}
\label{fig: exploration tree}
\end{figure}

Figure~\ref{fig: exploration tree} illustrates the construction of the \emph{exploration tree} of a loop configuration.
The (chordal) exploration tree connects the fixed root vertex to any other boundary point. The construction of 
the branch form the root to a fixed target vertex is the following:
\begin{enumerate}\enustyo
\item \label{enui: exploration tree intro i}
To initialize the process cut open the loop next to the root vertex and \emph{start the process} 
at that location on that loop. 
\item \label{enui: exploration tree intro ii}
\emph{Explore the current loop in clockwise direction} until 
the target point is reached or any point at which
it is clear that it is impossible to reach the target point along the current loop 
(meaning that the current location of the exploration is disconnected  from the target
in the discrete slit domain where the explored part has been removed
from the original domain).
\item 
\emph{If the target was reached, then
stop and return} (as the result of the algorithm) \emph{the path}  
concatenated from the subpaths of loops
in the order that they were explored.
\item 
\emph{If the target was not reached yet, 
cut open the loop next to the current location} and jump to that loop and 
\emph{go to Step~\ref{enui: exploration tree intro ii}}.
\end{enumerate}

The construction depends on the direction of the exploration which was chosen 
in (2)
to be clockwise. 
Instead of a deterministic choice, independent coin flips could be used to decide whether
to follow each loop in clockwise or counterclockwise direction.
This would lead to a different
process.
In this article we will use the above construction which suits well our purposes.
After all, the main goal is to show convergence of the loop collection, and the above construction agrees well with our observable.

\subsubsection{Main result}\label{sssec: intro main results}

The main theorem 
of this article
is the following result. 
For its formulation, call
a discrete domain \emph{admissible} if it
is simply connected 
and bounded
and its boundary consists of a chain of black octagons and small squares 
as in Figure~\ref{sfig: loop representation a}.
More generally, introduce lattice mesh by scaling the lattices
($\sqlattice, \dsqlattice, \dmdlattice, \ddmdlattice$ etc.) by a factor
$\delta>0$ and 
consider discrete, admissible domains with lattice mesh $\delta>0$,
and
use notation $\Omega^{(\delta)}$ to explicitly refer to such a domain.
A sequence of simply connected domains $\Omega_n$
(say $\Omega_n=\Omega^{(\delta_n)}$) converges to a domain $\Omega$ 
in \emph{Carath\'eodory sense} with respect to a point $\intpoint \in \Omega$,
if $\intpoint \in \Omega_n$ for all $n$ and the conformal and onto maps
$\psi_n: \disc \to \Omega_n$ with $\psi_n(0)=w_0$ and $\psi_n'(0)>0$
converge uniformly on compact subsets of $\disc$ to a conformal and onto map
$\psi: \disc \to \Omega$ with $\psi(0)=w_0$ and $\psi'(0)>w_0$.
Notice that then the inverse maps $\phi_n = \psi^{-1}_n$ 
converge uniformly 
on any compact subset $K$ of $\Omega$
($K$ belongs to the domain of $\phi_n$ for large $n$).

\begin{theorem}\label{thm: main theorem}
Let $\Omega_{\delta_n}$ be a sequence of admissible domains 
that converges to a domain $\Omega$ in Carath\'eodory sense with respect to some fixed $\intpoint \in \Omega$
 and
let $\rndlpe_{\partial,\delta_n}$ be the random collection of loops which 
is the collection of all loops 
of the loop representation of the FK Ising model on $\Omega_{\delta_n}$ that intersect the boundary
(the boundary touching loops).
Let $\phi_n$ be a conformal map that sends $\Omega_n$ onto the unit disc and $\intpoint$ to $0$.
Then as $n\to \infty$, the sequence of random loop collections 
$\phi_n(\rndlpe_{\partial,\delta_n})$ converges weakly to a limit $\therndlpe$  whose
law is independent of the choices of $\Omega, \delta_n, \Omega_{\delta_n}, \intpoint$ and $\phi_n$, 
and hence the law is invariant under all conformal isomorphisms of $\disc$.

Moreover the law of $\therndlpe$ 
is the boundary touching loops of CLE$(\kappa)$ with $\kappa=16/3$ and
is given by the image of the SLE$(\kappa,\kappa-6)$ exploration tree with $\kappa=16/3$
under a tree-to-loops mapping which inverts the construction in Section~\ref{sssec: exploration tree} and
is explained in more details in Section~\ref{ssec: exploration tree definition}
(including definitions needed for understanding this theorem 
and making the statements more precise).
\end{theorem}

\begin{remark}
In this article we will prove Theorem~\ref{thm: main theorem}
only in the case 
that $\Omega$ is smooth and that each of its discrete 
approximations have boundary which close to the
boundary of $\Omega$ in the sense that their distance
is bounded by a uniform constant times the lattice step of the 
approximation. 
By these assumptions,
we will exclude the cases where the boundary forms long fjords to have
better estimates for the harmonic measure. 
The general case follows from the sequel~\cite{Kemppainen:2016th} of this article
on the \emph{radial} exploration tree of the FK Ising model.
These restrictions are technical and they are not needed, for instance, for the  
convergence in the 
so-called
4-point case 
(Section~\ref{ssec: convergence observable 4-pt} and other arguments leading to 
the convergence of the interface to \varsle{\kappa,Z} in Section~\ref{ssec: martingale char for 4-pt}). 
\end{remark}

The present article aims to provide clear 
details for the basic proof techniques which include 
the regularity properties of
trees, derivation of the martingale observables and the corresponding martingale
characterization in the boundary-touching-loop setting.
In principle, one should be able to deduce the complete picture by repeatedly iterating 
this construction inside the resulting holes appearing after removing the boundary touching loops, 
the main difficult point being the fractal boundary. Instead, in the sequel~\cite{Kemppainen:2016th} 
we build the complete tree towards interior points, thus not having to deal with fractal boundaries.
This requires working with a more complicated observable, and the proof in the current article better 
explains what follows in the sequel.

Our result in the 4-point setting is interesting in its own right. In a follow-up paper~\cite{Kemppainen:CEhYg0we},
we use it to show that the interface  conditioned on an internal
arc pattern converges towards so-called hypergeometric SLE.

A sample of FK Ising branch is illustrated in Figure~\ref{fig: fk ising sample and exploration tree}.

\begin{figure}[tbh]
\subfigure[A sample of FK Ising model with free boundary conditions.] 
{
	\includegraphics[scale=.9]
{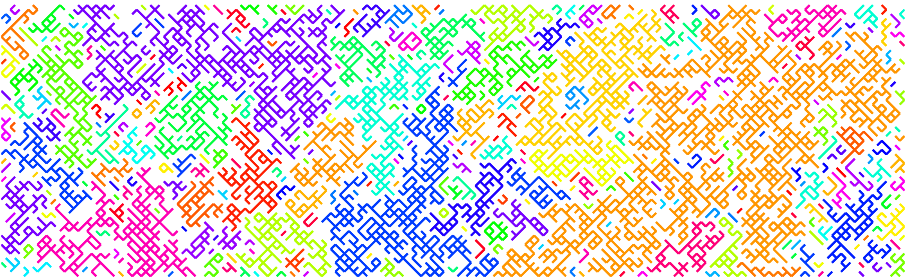}
}
\subfigure[The boundary touching loops of the loop representation which intersect the bottom of the rectangle.] 
{
	\includegraphics[scale=.9]
{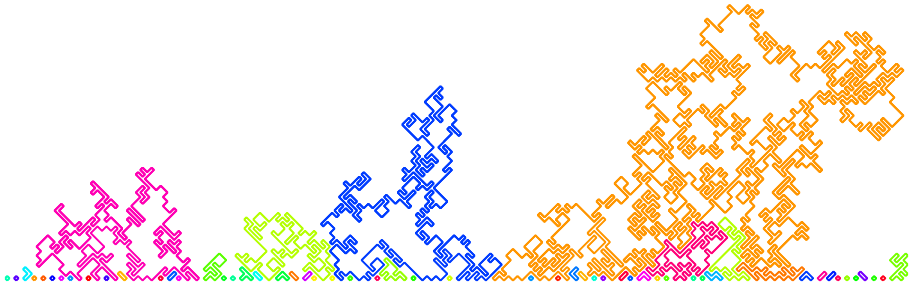}
}
\subfigure[The exploration process between the lower left and right corners. Notice that the process uses only some
loops touching the bottom arc and it uses only the loop arcs which are at top as seen from the lower right corner.] 
{
	\includegraphics[scale=.9]
{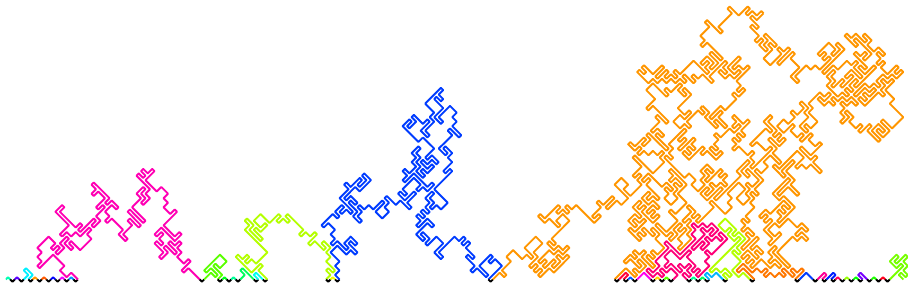}
}
\caption{A sample configuration of the critical FK Ising model on a rectangular domain 
with free boundary conditions. Components with only one vertex are not shown. 
The colors distinguish different components. 
Notice that in this particular sample, the large orange loop happens to come fairly close 
to the boundary without disconnecting the small loops to its right.
Thus the exploration path turns and explores those boundary touching loops
inside the ``fjord.''}
\label{fig: fk ising sample and exploration tree}
\end{figure}

\subsection{Previous results on conformally invariant 
scaling limits of random curves and loops}

So far, convergence of a single discrete interface  to \slek's has been established for but a few models:
 $\kappa=2$ and $\kappa=8$ \cite{Lawler:2004un},
$\kappa=3$ and $\kappa=\frac{16}{3}$ \cite{Chelkak:2014gs},
$\kappa=4$ \cite{Schramm:2005wh,Schramm:2006uc} and
$\kappa=6$ \cite{Smirnov:2001hw,Smirnov:2009vj}.
However, the framework for the full scaling limit, including all interfaces, is less developed: 
in addition to the present article
only $\kappa=3$ \cite{Benoist:2016uy}
and $\kappa=6$ \cite{Camia:2006wv} and our subsequent work 
\cite{Kemppainen:2016th}.
A similar result on a collection of random curves is
the convergece of the free arc ensemble of the Ising model in 
\cite{benoist-2016a}.

\subsection{Organization of the article}

We will give further definitions in Section~\ref{sec: setup}. In Section~\ref{sec: regularity of trees and loops}
we explore the regularity and tightness properties of the loop configurations and the exploration trees
based on crossing estimates. This gives a priori knowledge needed in the main argument.
In Section~\ref{sec: observable} we define the holomorphic observable and show its convergence.
In Section~\ref{sec: characterization} we combine these tools and extract information from the observable
so that we can characterize the scaling limit and prove the main theorem in Section~\ref{sec: proof main theorem}.


\section{The setup and more details of the main result}\label{sec: setup}

\subsection{Graph theoretical notations and setup}

In this article the lattice $\sqlattice$ is the square lattice $\Z^2$ rotated by $\pi/4$, 
$\dsqlattice$ is its dual lattice, which itself is also a square lattice and $\dmdlattice$ is their (common) medial lattice. 
More specifically, we define three lattices $G=(V(G),E(G))$, where $G= \sqlattice,\dsqlattice,\dmdlattice$, as
\begin{gather}
V( \sqlattice ) = \left\{ (i,j) \in \Z^2 \,:\, i+j \text{ even} \right\}, \quad
   E( \sqlattice ) = \left\{ \{v,w\} \subset V( \sqlattice ) \,:\, |v-w| = \sqrt{2} \right\},  \\
V( \dsqlattice ) = \left\{ (i,j) \in \Z^2 \,:\, i+j \text{ odd} \right\}, \quad
   E( \dsqlattice ) = \left\{ \{v,w\} \subset V( \dsqlattice ) \,:\, |v-w| = \sqrt{2} \right\},  \\
V( \dmdlattice ) = (1/2 + \Z)^2 , \quad
   E( \dmdlattice ) = \left\{ \{v,w\} \subset V( \dmdlattice ) \,:\, |v-w| = 1\right\} 
\end{gather}
Notice that sites of $\dmdlattice$ are the midpoints of the edges of $\sqlattice$ and $\dsqlattice$.

We call the vertices and edges of $V( \sqlattice )$ \emph{black} and
the vertices and edges of $V( \dsqlattice )$ \emph{white}.
Correspondingly the faces of $\dmdlattice$ are colored black and white depending whether the
center of that face belongs to $V( \sqlattice )$ or $V( \dsqlattice )$.

The directed version $\odmdlattice$ is defined by setting $V(\odmdlattice)=V(\dmdlattice)$
and orienting the edges around any black face in the counter-clockwise direction.

The modified medial lattice $\ddmdlattice$, which is a square--octagon lattice, 
is obtained from $\dmdlattice$ by replacing each site by a small square.
See 
Figure~\ref{sfig: graphs b}.
The oriented lattice $\oddmdlattice$
is obtained from $\ddmdlattice$ by orienting the edges around black and white octagonal faces
in counter-clockwise and clockwise directions, respectively.

\begin{definition}
A simply connected, non-empty, bounded domain $\domain$ is said to be a \emph{wired $\oddmdlattice$-domain} (or \emph{admissible domain})
if $\partial \domain$ oriented in counter-clockwise 
direction is a path in $\oddmdlattice$.
\end{definition}

See Figure~\ref{fig: random cluster sample and branch of tree} for an example of such a domain.
The wired $\oddmdlattice$-domains are in one to one correspondence with non-empty finite subgraphs of $\sqlattice$
which are simply connected, i.e., they are graphs who have an unique unbounded face 
and the rest of the faces are 
unit-size
squares.

\subsection{FK Ising model: notations and setup for the full scaling limit}

Let $G$ be a simply connected subgraph of the square lattice $\sqlattice$. Consider the random cluster measure
$\mu=\mu_{p,q}^1$ of $G$ with all \emph{wired boundary conditions} in the special case of the critical FK Ising model,
that is, when $q=2$ and $p = \sqrt{2}/(1 + \sqrt{2})$.
Its dual model is again a critical FK Ising model, now with free boundary conditions on the dual graph $G^\circ$ of $G$
which is a (simply connected) subgraph of $\dsqlattice$. 
The \emph{loop representation} is obtained as 
loops which form boundaries between open cluster and the dual open
clusters and is defined as a collection of loops  on the corresponding
subgraph $G^\ddiamond$ of the modified medial
lattice $\ddmdlattice$. The loop collection satisfies the properties of the following definition.
See also Figure~\ref{fig: loop representation} for illustration of
the  common loop representation shared by the random cluster model and its dual model.

Let's call a 
(unordered)
collection of loops $\mathcal{L}=(L_j)_{j =1 ,\ldots  N}$ 
on $G^\ddiamond$ a \emph{dense collection of non-intersecting loops} (DCNIL)
if 
\begin{itemize}
\item each $L_j \subset G^\ddiamond$ is a simple loop
\item $L_j$ and $L_k$ are vertex-disjoint when $j \neq k$
\item for every edge $e \in E^\diamond$ 
there is a loop $L_j$ that visits $e$. Here we use that $E^\diamond$ is naturally a subset of $E^\ddiamond$.
\end{itemize}
Let the collection of all the loops in the loop representation be $\rndlpe=(\rndlp_j)_{j =1 ,\ldots  N}$. 
Then DCNIL is exactly the support of $\rndlpe$ and
for any DCNIL collection $C$ of loops 
\begin{equation}
\mu(\rndlpe = \lpe) = \frac{1}{Z} (\sqrt{2})^{\text{\# of loops in }\lpe}
\end{equation}
where $Z$ is 
the partition function that normalizes the probability measure.

\begin{figure}[tbh]
\centering
\subfigure[A sample of random cluster model and the corresponding loop configuration.] 
{
	\includegraphics[scale=.37]
{fki-p18}
} 
\hspace{0.2cm}
\subfigure[Two branches of exploration tree from $v_\text{root} \in V_\partial$ to
$w,w' \in V_\partial$. Note the target independence of the process and that the branching takes place
on the vertices of $V_{\partial,1}$.]
{
	\includegraphics[scale=.37]
{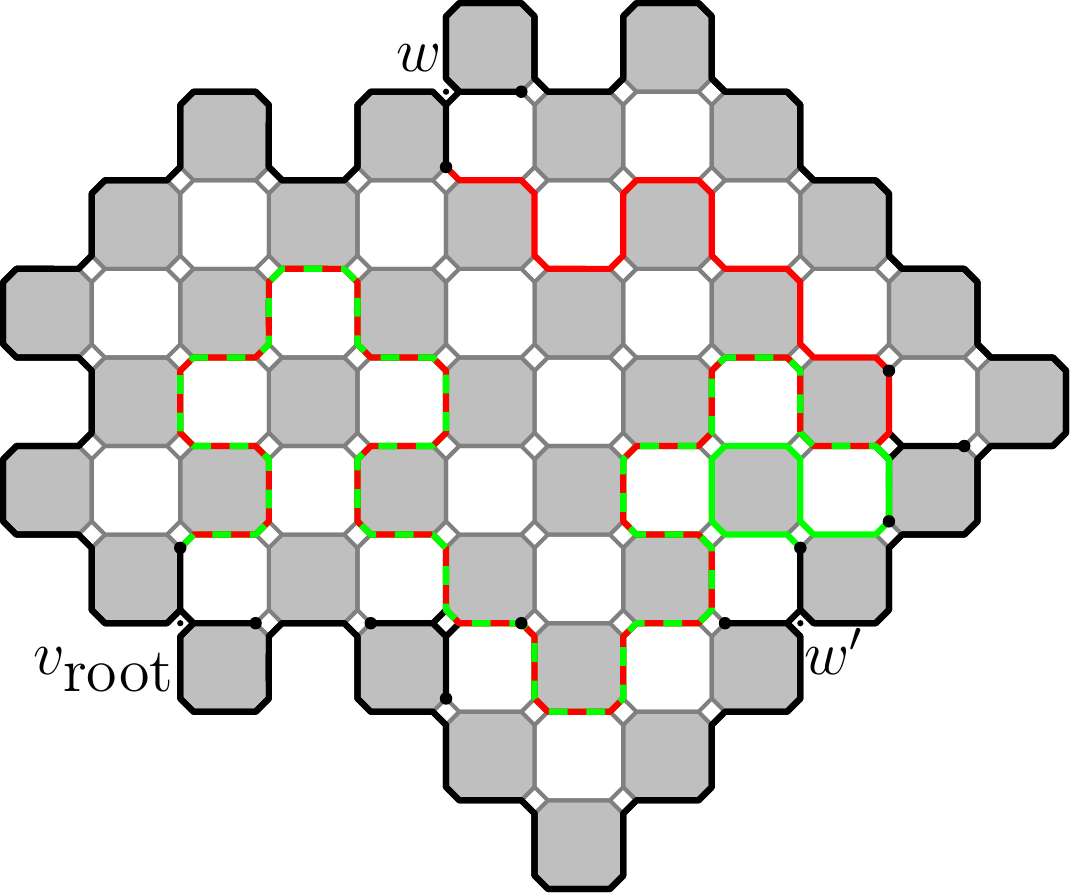}
}
\caption{The construction of the exploration tree of the random cluster model.}
\label{fig: random cluster sample and branch of tree}
\end{figure}

We consider two boundaries of the domain, one which is the boundary of the domain
and one which shifted by one lattice step from the first one towards
the interior of the domain. They are both simple loops on the lattice
which satisfy the same parity condition as the loops of the random cluster 
loop representation
(the octagons on both sides have uniform color).
More specifically
\begin{itemize}
\item $\partial G^\ddiamond$ is the boundary of the domain,
in the usual topological sense. Call it the \emph{external boundary}.
\item $\partial_1 G^\ddiamond$ is the outermost (simple) loop can be drawn in 
$G^\ddiamond$. In other words, it is the outermost loop of 
the empty random cluster configuration with wired boundary conditions.
Call it the \emph{internal boundary}. We say that $\partial_1 G^\ddiamond$ touches
the boundary everywhere and that a \emph{loop touches the boundary} if 
if it intersects $\partial_1 G^\ddiamond$. Notice that if
a loop and $\partial_1 G^\ddiamond$ intersect then they
share an edge (which is an edge shared by two octagons).
\end{itemize}%

Define the collection of \emph{boundary touching loops}, $\rndlpe_\partial \subset \rndlpe$, 
to be simply the set of loops which intersect the 
internal
boundary $\partial_1 G^\ddiamond$.

Recall the Carath\'eodory convergence from Section~\ref{sssec: intro main results}.
Take a bounded simply connected domain $\domain$ in the plane. 
And take a sequence $\delta_n \searrow 0$ as $n \to \infty$
and a sequence of simply connected graphs
$G_{\delta_n}^\bullet \subset \delta_n \sqlattice$ which approximate $\domain$ in the sense that, if we denote
by $\domain_{\delta_n}$ the bounded component of $\C \setminus \partial G_{\delta_n}^\ddiamond$,
then $\domain_{\delta_n}$ converges in Carath\'eodory convergence 
to $\domain$ (with respect to any interior point of $\domain$).
Fix any $\intpoint \in \domain$ and let $\phi_{\delta_n}: \domain_{\delta_n} \to \disc$ 
be conformal transformations normalized in the usual way using $\intpoint$, that is,
\begin{equation*}
\phi_{\delta_n}(\intpoint)=0, \quad \phi_{\delta_n}'(\intpoint)> 0 .  
\end{equation*}
Let $\rndlpe_{\partial,\delta_n}$ be 
the collection of boundary touching loops in $\Omega_{\delta_n}$ and set
$\tilde{\rndlpe}_{\partial,\delta_n} = \phi_{\delta_n} (\rndlpe_{\partial,\delta_n}) $ .

Let us 
rephrase
here the first half of Theorem~\ref{thm: main theorem}.

\begin{mainthmA}[Conformal invariance of $\rndlpe_\partial$]
As $n \to \infty$, $\tilde{\rndlpe}_{\partial,\delta_n}$ converges weakly to a random collection $\therndlpe$ of loops in $\disc$.
The law of $\therndlpe$ is independent of $\domain, \delta_n, \domain_{\delta_n}$ and $\intpoint$.
\end{mainthmA}

\begin{remark}
Notice that the independence of the law of $\therndlpe$ from $\domain, \delta_n, \domain_{\delta_n}$ and $\intpoint$
implies that  $\therndlpe$ is invariant under all 
conformal automorphisms 
(M\"obius transformations) of $\disc$.
The rotational invariance requires a separate argument using
the correspondence between exploration tree and the loop collection
and the fact that we are free to choose the root for the exploration.
See Section~\ref{sec: proof main theorem}.
\end{remark}

\subsection{The exploration tree of FK Ising model}\label{ssec: exploration tree definition}

Let's simplify the notation so that we use $\partial \Omega$ and $\partial_1 \Omega$ to denote
$\partial G^\ddiamond$ and $\partial_1 G^\ddiamond$.
Remember that $\partial \Omega$ and $\partial_1 \Omega$ are simple loops on $\oddmdlattice$
and that $\rndlpe_\partial$ was the set of loops in $\rndlpe$ that intersected $\partial_1 \Omega$.
Next we will explain the construction of the \emph{exploration tree} of $\rndlpe_\partial$.
The branches of the tree will be simple paths from a root edge to a directed edge of $\partial_1 \Omega$.
More specifically let $\ttrtarget$ be the vertex set of $\partial_1 \Omega$
and for each $v \in \ttrtarget$,  let $f_v$ be the edge of $\partial_1 \Omega$ arriving to $v$.

We assume that the root vertex $\vroot \in V_{\partial}$ of the exploration is fixed. For any $w \in \ttrtarget$, 
we are going to construct a simple path which starts from the inwards pointing edge of $\vroot$
and ends on the edge $f_w$ of $w$, denoting this path by $\bran_w=\bran_{\vroot,w}$.
We will call the mapping from the loop collections to the trees the ``\emph{loops-to-tree map}.''

\begin{figure}[tbh]
\centering
	\includegraphics[scale=.675]
{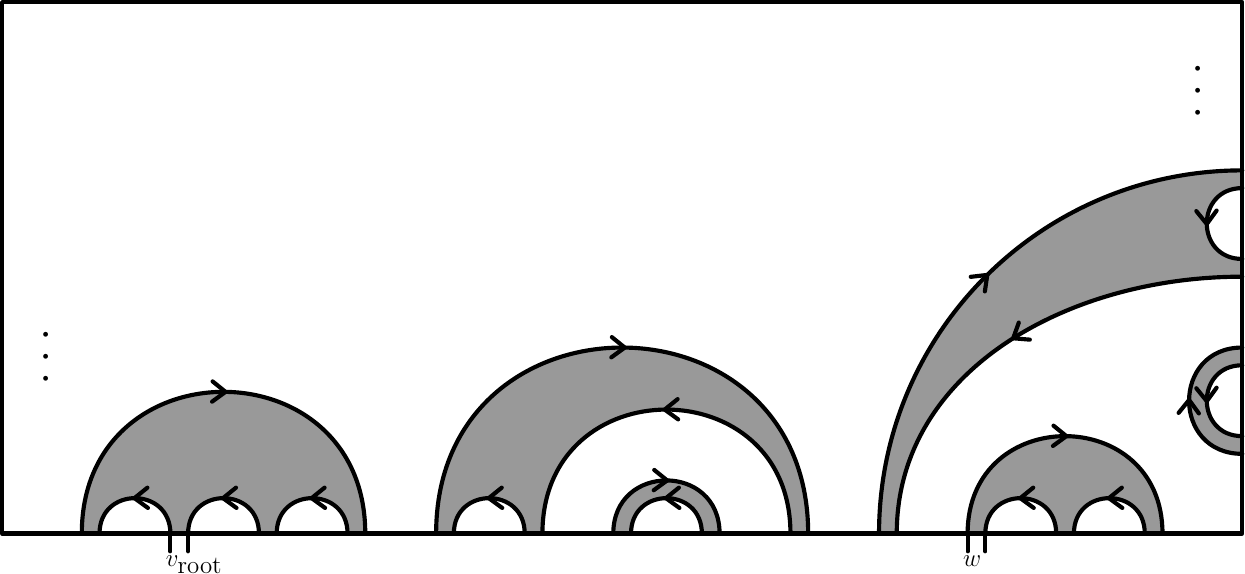} 
\caption{A schematic picture of the boundary touching loops. The interiors of the loops are shaded to make it easier to
distinguish the loops and the arrows indicate the clockwise orientation of the loops.}
\label{fig: schematic loops}
\end{figure}

We describe next the algorithm of Section~\ref{sssec: exploration tree} when the target point is on the boundary.
Consider a loop collection $\lpe=(\lp_j)_{j \in J}$ where $J$ is some finite index set,
and assume that $\lpe$ satisfies the properties of DCNIL.
The index $j \in J$ shouldn't be confused with the concrete sequence $\lp_1,\ldots,\lp_n$ 
chosen below.
Here we consider $\lp_j$ as a path in $\dmdlattice$
and lift it to $\ddmdlattice$ when needed.
\begin{enumerate}\enustyo
\item 
Set $e_0$ to be the inward pointing edge of $\partial_1 \Omega$ at $\vroot$ 
and $f_\text{end}=f_w$, that is the edge of $\partial_1 \Omega$ arriving to $w$.
Denote the set of edges in $\partial_1 \Omega$ that lie between $f_\text{end}$ and $e_0$, including $f_\text{end}$,
by $F_w$.
\item Set $\lp_1$ to be the loop going through $e_0$.
Find the first edge of $\lp_1$ after $e_0$ in the orientation (remember that all the loops are oriented in the clockwise direction)
of $\lp_1$ 
that
lies in $F_w$. Call it $f_1$ and the part of $\lp_1$ between $e_0$ and $f_1$, not including
$f_1$, $\lp_1^\itop$.

Notice
that $\lp_1$ goes through $f_\text{end}$ if and only if $f_1$ is equal to $f_\text{end}$.
Notice also that if $f_1$ is not $f_\text{end}$, then it is the first edge that takes the loop to a component of the domain that is no longer
``visible'' to $f_\text{end}$.
\item\label{enui: tree algorithm iterative step}%
Suppose that $f_n$ and $\lp_1^\itop, \lp_2^\itop, \ldots, \lp_n^\itop$ are known. If $f_n$ is equal to $f_\text{end}$, stop and return
the concatenation of $\lp_1^\itop, \lp_2^\itop, \ldots, \lp_n^\itop$ and $f_n$ as the result $\bran_w$. 
Otherwise take 
the 
inward pointing $e_n$
edge next to $f_n$ 
(starting at the endpoint of $f_n$)
and the loop $\lp_{n+1}$ passing through $e_n$. 
Find the first edge of $\lp_{n+1}$ after $e_n$ in the orientation 
of $\lp_{n+1}$ that lies in $F_w$. Call it $f_{n+1}$ and the part of $\lp_{n+1}$ between $e_n$ and $f_{n+1}$, 
not including $f_{n+1}$,
$\lp_{n+1}^\itop$. 
Repeat \ref{enui: tree algorithm iterative step}.
\end{enumerate}
The result of the algorithm $\bran_w$ is called the \emph{exploration process} from $\vroot$ to $w$.
The collection $\ttr = (\bran_w)_{w \in \ttrtarget}$ where $\ttrtarget=V_\partial$ is called
the \emph{exploration tree} of the loop collection $\lpe$ rooted at $\vroot$.

The following result is immediate from the definition of the exploration tree. 
Two branches coincide until the first time that the branch disconnects the target points 
by that result,
and after that the branches explore disjoint regions, which will later imply independence
of this processes for the FK Ising exploration tree.

\begin{proposition}[Target independence of exploration tree]
Suppose that $\vroot, w, w'$ are vertices in $V_\partial$ in counterclockwise order, that can be the same. 
Let $F_{w,w'}$ the edges of $\partial_1 \Omega$ that 
lie between the outward pointing edges of $w$ and $w'$
in counterclockwise direction, including the edge at $w'$ and but excluding the edge at $w$.
Then $\bran_w$ and $\bran_{w'}$ are equal up until the first edge lying in $F_{w,w'}$.
\end{proposition}

Next we will construct the inverse of the loops-to-tree map which we will call a ``\emph{tree-to-loops map}.''
The structure of the tree $\ttr$ is the following: the branches of $\ttr$ are simple and follow the rule of leaving
white squares on their right and black on their left. The branching occurs in a subset of vertex set of $\partial_1 \Omega$.
There is a one-to-one correspondence between branching points of $\ttr$ and the boundary touching loops
of $\lpe$: the point on the loop, which lies on the boundary and is the closest one to the root if we move clockwise along the boundary,
is a branching point and every branching point has this property for some loop. 
See also Figure~\ref{fig: schematic loops}.
Suppose that at a branching point $w$ the incoming edges are $e_1$ and $e_2$ and the outgoing edges
are $f_1$ and $f_2$ and they are in the order $e_1,f_2,e_2,f_1$ 
clockwise
and that the exploration process enters $w$ through $e_1$. 
Then necessarily each of the pairs
$e_1, f_1$ and $e_2, f_2$
lie on the same loop of $\lpe$ and these two loops are different. 
Also it follows that $f_1$ and $e_2$ are on $\partial_1 \Omega$ while $f_2$ and $e_1$ are not. 
It follows that the last edge of $T_w$ is $e_2$ and that the part of $T_w$ between $f_2$ and $e_2$ is exactly the loop
of $\lpe$ that touches the boundary at $f_w=e_2$.
Doing the same thing for every branching point 
defines a mapping from a suitable set of trees onto the set of loop collections of boundary touching loops of DCNIL.
This mapping inverts the construction of the exploration tree and we summarize it in the following lemma.

\begin{lemma}\label{lem: tree to loop ensemble}
The mapping from the collection of boundary touching loops $\lpe_\partial$ to the chordal exploration tree $\ttr$ is a bijection.
\end{lemma}

Similar constructions work in the continuous setting. See \cite{Sheffield:2009} for the construction of the SLE$(\kappa,\kappa-6)$
exploration tree and the construction for recovering the loops.

Let us repeat here the second half of Theorem~\ref{thm: main theorem}.

\begin{mainthmB}
The law of $\therndlpe$ is given by the image  of the SLE$(\kappa,\kappa-6)$ exploration tree with $\kappa=16/3$
under the above tree-to-loops mapping.
\end{mainthmB}

The proof of Theorem~\ref{thm: main theorem}
is given in Section~\ref{sec: proof main theorem}.


\section{Tightness of trees and loop collections}\label{sec: regularity of trees and loops}

In this section, we establish a priori bounds for trees and loop ensembles. The setting
is relatively general, although we only apply it here to the FK Ising exploration tree
of the boundary touching loops.

\subsection{A probability bound on multiple crossings by the tree}\label{ssec: multip crossings tree}

An approach to establish
compactness properties of sequences of probability measures based on
probability bounds of multiple crossings of annuli by random curves was set up in 
\cite{Kemppainen:2012vm}
extending the results of \cite{Aizenman:1999dt}.
Below we use that type of result for the FK Ising exploration tree.
We start from the essential definitions.

For any fixed measurable space $(\mathcal{S},\F)$, we
call a random variable $X$ tight over 
a collection $\Sigma_0$
of probability measures $\P \in \Sigma_0$
on the space $(\mathcal{S},\F)$,
if for each $\eps>0$
there exists a constant $M>0$ such that $\P(|X|<M)>1-\eps$
for all $\P$.

A crossing of an annulus $A(z_0,r,R)=\{z \in \C \,:\, r<|z-z_0|<R\}$
is a closed segment of a curve that intersects both connected components of $\C \setminus A(z_0,r,R)$
and a minimal crossing doesn't contain any genuine subcrossings.

Recall the general setup of \cite{Kemppainen:2012vm}: we are given a collection $(\varconfmap,\P) \in \Sigma$ 
where the conformal map $\varconfmap$ 
contains also the information about its domain of definition 
$(\vardomain,\vroot,\intpoint)=(\vardomain(\varconfmap),\vroot(\varconfmap),\intpoint(\varconfmap))$ 
through the requirements
\begin{equation}\label{eq: def conf map setting}
\varconfmap^{-1}(\disc)=\vardomain, \qquad \varconfmap(\vroot)=-1 \qquad \text{and} \qquad \varconfmap(\intpoint)=0
\end{equation}
and $\P$ is the probability law of
FK Ising model on the discrete domain $\vardomain$ and in particular gives the distribution of
the FK Ising exploration tree.
Given the collection $\Sigma$ of pairs $(\varconfmap,\P)$ 
we define the collection $\Sigma_\disc = \{\varconfmap\P \,:\, (\varconfmap,\P)\in\Sigma \}$
where $\varconfmap\P$ is the pushforward measure defined by $(\varconfmap\P) (E) = \P(\varconfmap^{-1}(E))$. 

\begin{theorem}\label{thm: aizenman--burchard bounds}
The following claim holds for the collection of the probability laws of FK Ising exploration trees
\begin{itemize}
\item 
for any $\Delta>0$, there exists $n \in \N$ and $K>0$ such that 
\begin{equation}\label{ie: annulus multiple crossing}
\P( \text{at least $n$ disjoint segments of $\rndttr$ cross } A(z_0,r,R))
  \leq K \left( \frac{r}{R} \right)^\Delta
\end{equation}
for all $\P \in \Sigma_\disc$ and for all $z_0 \in \C$ and $R>r>0$.
\end{itemize}
and there exist positive numbers $\alpha, \alpha'>0$ such that the following claims hold
\begin{itemize}
\item 
if for each $r>0$, $M_r$ is the minimum of all $m$ such that each 
$\rndbran \in \rndttr$ can be split into $m$
segments of diameter less or equal to $r$, then
there exists a random variable $K(\rndttr)$ such that
$K$ is a tight random variable for the family $\Sigma_\disc$
and
\begin{equation*}
M_r \leq K(\rndttr) \, r^{-\frac{1}{\alpha}}
\end{equation*}
for all $r>0$.
\item All branches of $\rndttr$ can be jointly parametrized so that they are
all $\alpha'$-H\"older continuous and the H\"older norm can be bounded
by a random variable $K'(\rndttr)$ such that $K'$  
is a tight random variable for the family $\Sigma_\disc$.
\end{itemize}
\end{theorem}

Each of the claims have their own applications
below although they are closely related, see \cite{Aizenman:1999dt}.

\begin{proof}
We need to verify the first claim
and the two other claims follow from it, by results of \cite{Aizenman:1999dt};
more specifically the second claim follows from the reformulated statement
presented in the beginning of the proof of Theorem~1.1 of \cite{Aizenman:1999dt}
and the third claim follows from Theorem~1.1 of \cite{Aizenman:1999dt}.
Notice that we need to verify 
the inequality~\eqref{ie: annulus multiple crossing}
for $\Delta = \Delta_{n_i}$ and $n=n_i$ where $n_i$ is an increasing sequence of natural numbers
and $\Delta_{n_i}$ is a sequence of positive real numbers tending to infinity, since 
the left-hand side of the inequality in the first claim is non-increasing in $n$.

Let $A=A(z_0,r,R)$ be annulus, $z_0 \in \disc$. Since
either $B(z_0,\sqrt{rR}) \subset \disc$ or $B(z_0,\sqrt{rR}) \cap \partial \disc \neq \emptyset$,
it holds that we can choose $r_1= \sqrt{rR}$, $R_1=R$ or $r_1=r$, $R_1= \sqrt{rR}$ 
such that for $A_1 = A(z_0,r_1,R_1)$, either $A_1 \subset \disc$ or 
$B(z_0,r_1) \cap \partial \disc \neq \emptyset$.
For that $A_1$ and for $C>1$ big enough, apply the estimate of conformal distortion given by either
Lemma~\ref{lem: dist ann bulk} or Lemma~\ref{lem: dist ann boundary}, depending on the case,
to show that for any $m=1,2,\ldots,\lfloor \log \ffrac{R}{r} \rfloor$,
there exist $\tilde A_m = A(z_m,r_m,2 r_m)$ such that
the conformal image
of any crossing of $A(z_0,C^{m-1} r, C^m r)$ under $\varconfmap^{-1}$ is a crossing of $A_m$.

By Lemmas~\ref{lem: crossings and arms bulk} and \ref{lem: crossings and arms boundary}
in Appendix~\ref{sec: appendix a priori}
and the
results of \cite{Chelkak:2016jw} (in particular, Lemma~5.7)
applied to crossings of $\tilde A_m$,
it follows that for each $\eps>0$ there is $n$ such that 
$\P(\text{at least $n$ disjoint segments of $\therndttr$ cross $\tilde A_m$} ) < \eps$.
Thus 
\eqref{ie: annulus multiple crossing} holds for $n$ and constants
$K = \eps^{-1}$ and $\Delta = \log \frac{1}{\eps}$,
Here the constant $\Delta$ tends to $\infty$ as $n\to \infty$ (i.e. as $\eps$ tends to zero),
and the estimates are uniform over all $\P \in \Sigma_\disc$ and
annuli $A(z_0,r,R)$ with $R>r$.
\end{proof}

\subsection{The crossing property of trees}\label{ssec: tree crossing}

Let $\gamma_k$, $k=0,\ldots,N-1$, is the  
collection $\rndttr = ( \rndbran_x)$ in a (random) order
and suppose that the random curves $\gamma_k$ are each parametrized by $[0,1]$. 
The chosen permutation specifies
\emph{the order of exploration} of the curves. More specifically, set 
$\underline{\gamma}(t) = \gamma_k(t-k)$ when $t \in [k,k+1)$.
We call $\underline{\gamma}$ an \emph{explored} collection of branches.

For a given domain $\vardomain$ and for a given simple (random) curve $\gamma$ on $\vardomain$, we set
$\vardomain_\tau = \vardomain \setminus \gamma[0,\tau]$ for each (random) time $\tau$. 
Similarly, for a given domain $\vardomain$ and for the given finite 
collection of curves $\gamma_k$ on $\vardomain$, we set
$\underline{\vardomain}_\tau = \vardomain \setminus \underline{\gamma}[0,\tau]$ for each (random) time $\tau$. 

We call $\vardomain_\tau$ or $\underline{\vardomain}_\tau$  the domain at time
$\tau$.

The following definition generalizes Definition~2.3 from \cite{Kemppainen:2012vm}.
This definition is needed in order to recognize those
crossing events which have low probability.

\begin{definition}
For a given domain $(\vardomain,\vroot)$ and for a given order of exploration 
(which defines $\underline{\gamma}$) 
of curves $\gamma_x$,
$x \in V_\text{target}$, where each curve 
$\gamma_x$ 
is contained in $\overline{\vardomain}$, starting from $\vroot$ and ending at a point $x$ in the set $V_\text{target}$,  
define for any annulus $A = A(z_0,r,R)$, for every (random) time $\tau \in [0,N]$ and $x \in V_\text{target}$, 
$A^{\text{u},x}_\tau = \emptyset$ if $\partial B(z_0,r) \cap \partial \underline{\vardomain}_\tau = \emptyset$ and
\begin{equation}
A^{\text{u},x}_\tau = \left\{ z \in \underline{\vardomain}_\tau \cap A \,:\, 
   \begin{gathered}
   \text{the connected component of $z$ in $\underline{\vardomain}_\tau \cap A$} \\
   \text{doesn't disconnect 
   $\underline{\gamma}(\tau)$
   from $x$ in $\underline{\vardomain}_\tau$}
   \end{gathered}
   \right\}
\end{equation}
otherwise. Define also
\begin{equation}
A^{\text{f},x}_\tau = \left\{ z \in \vardomain_\tau \cap A \,:\, 
   \begin{gathered}
   \text{the connected component of $z$ in $\underline{\vardomain}_\tau \cap A$} \\
   \text{is crossed by any path connecting 
   $\underline{\gamma}(\tau)$
   to $x$ in $\underline{\vardomain}_\tau$}
   \end{gathered}
   \right\}
\end{equation}
and set $A^{\text{u}}_\tau = \bigcap_{x \in V_\text{target}}A^{\text{u},x}_\tau$
and $A^{\text{f}}_\tau = \bigcup_{x \in V_\text{target}} A^{\text{f},x}_\tau $.
We say that $A^{\text{u},x}_\tau$ is \emph{avoidable} for $\gamma_x$ and $A^{\text{u}}_\tau$
is \emph{avoidable for all} (branches).
We say that $A^{\text{f},x}_\tau$ is \emph{unavoidable} for $\gamma_x$ and $A^{\text{f}}_\tau$
is \emph{unavoidable for at least one} (branch).
Here and in what follows we only consider allowed lattice paths when we talk about
connectedness.
\end{definition}

\begin{remark}
Note that $A^{\text{u}}_\tau \cap A^{\text{f}}_\tau = \emptyset$.
This follows from $A^{\text{u},x}_\tau \cap A^{\text{f},x}_\tau = \emptyset$
which holds by definition.
\end{remark}

Recall that $\partial \Omega$ is the boundary of $\Omega$ and $\partial_1 \Omega$ 
is the internal boundary of $\Omega$ which is the outermost of all (simple) lattice paths
are contained in $\Omega$.
A point on $\partial_1 \Omega$ is a \emph{branching point} of the tree $\rndttr$
if it is the last common point of two branches. In that case the edge on the primal lattice
passing through the point has to be open in the random cluster configuration.
See also Figure~\ref{fig: branching and open paths}.

\begin{figure}[tbh]
\centering
	\includegraphics[scale=.35]
{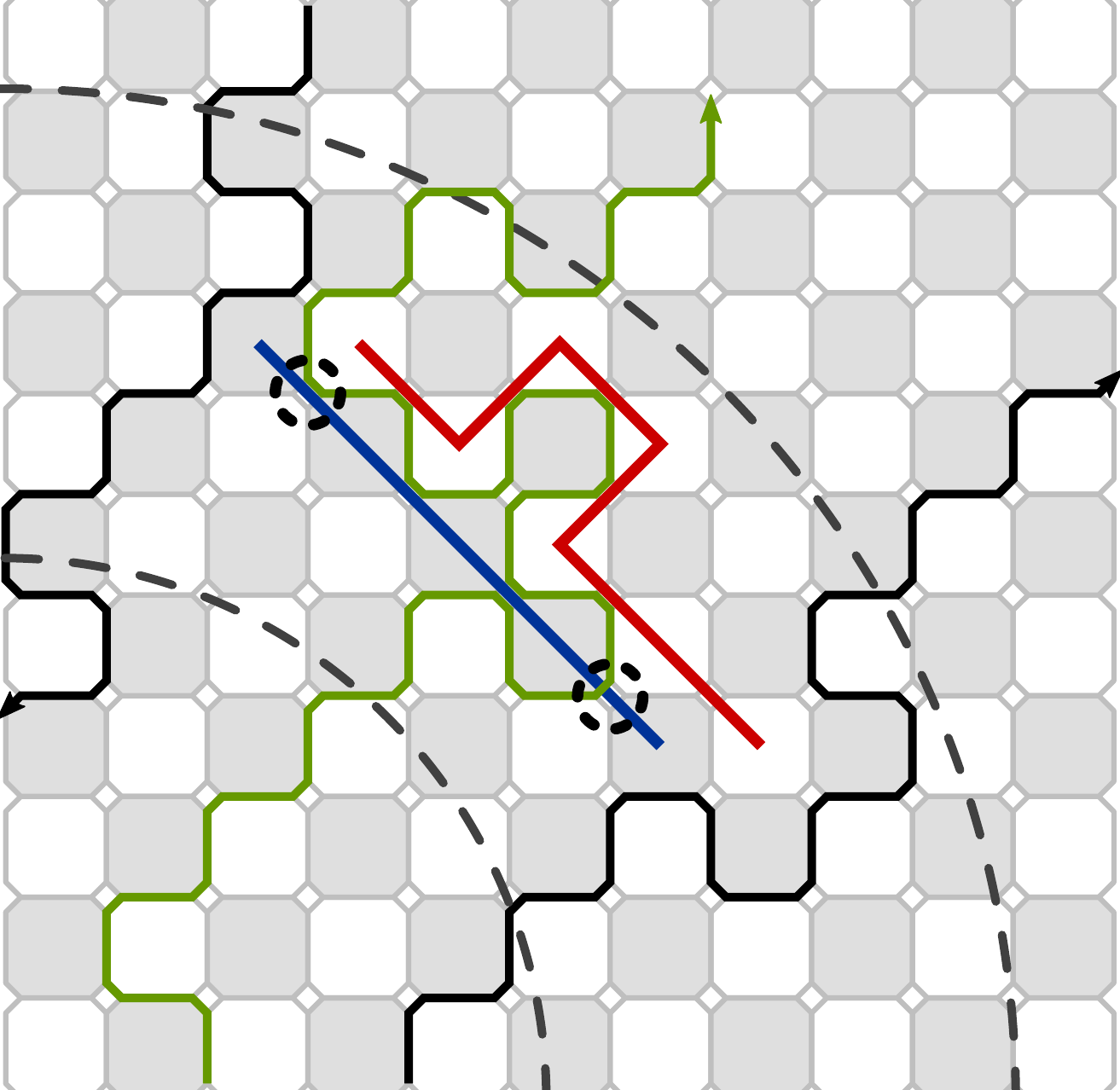} 
\caption{The region between the large dashed circular arcs is a quarter of an annulus.
A crossing of an annulus 
by the exploration path
is drawn in green and the boundaries of the domain $\underline{\vardomain}_\tau$ in black and
they are oriented according to the orientation of the medial lattice.
As indicated by the figure if there are transversal open and dual-open paths in the annulus,
then the crossing has to go near the left- and right-hand boundaries of the annular sector.
Notice that at the branching point on the right-hand side the drawn branch jumps
from a loop to another one,
whereas on the left-hand side it keeps following the loop
explored at that time.
}
\label{fig: branching and open paths}
\end{figure}

Next we will write down 
an estimate
in the form of a hypothesis analogous the ones presented in Section~2
of \cite{Kemppainen:2012vm}. 
The estimate is sufficient for the desired compactness properties
of the exploration tree.
In fact, we will present two equivalent conditions here. As we will later
see that conformal invariance will hold for this type of conditions, again analogously to  \cite{Kemppainen:2012vm}.
These conditions will be verified for the FK Ising exploration tree below.

\begin{condition}\label{cond: annulus}
Let $\Sigma$ be as above.
If there exists $C >1$ such that
for any $(\phi,\P) \in \Sigma$,
for any stopping time $0 \leq \tau \leq N$ and for any annulus $A=A(z_0,r,R)$ where $0 < C \, r \leq R$, 
it holds that
\begin{equation}\label{eq: cond annulus}
\P \left( \left. \;
      \begin{gathered}
      \underline{\gamma}[\tau,N] \textrm{ makes a crossing of } A \\
      \text{ which is contained in } \overline{A^u_\tau} \\
      \text{\emph{or}} \\ 
      \text{which is contained in } \overline{A^f_\tau} 
      \text{ and the first minimal crossing}\\
      \text{doesn't have branching points on both sides} 
      \end{gathered}      
   \;\,\right|\, \F_\tau \right)  < \frac{1}{2}  . 
\end{equation}
then the family $\Sigma$
is said to satisfy a \emph{geometric joint unforced--forced crossing bound}
Call the event above $E^{\text{u,f}}$.
\end{condition}

See Figure~\ref{sfig: f crossing} for more information about different types of  branching points.

\begin{condition}\label{cond: annulus exp}
The family $\Sigma$ 
is said to satisfy a \emph{geometric joint unforced--forced crossing power-law bound}
if there exist $K >0$ and $\Delta>0$ such that 
for any $(\phi,\P) \in \Sigma$,
for any stopping time $0 \leq \tau \leq N$ and for any annulus $A=A(z_0,r,R)$ where $0 < r \leq R$,
\begin{equation}\label{eq: cond annulus exp}
\text{LHS} \leq K \left( \frac{r}{R} \right)^\Delta  . 
\end{equation}
Here LHS is the left-hand side of \eqref{eq: cond annulus}.
\end{condition}

\begin{figure}[tb]
\newcommand*{\myscale}{0.65}
\centering
\subfigure[Crossing with branching points on both sides] 
{
	\label{sfig: f crossing a}
	\includegraphics[scale=\myscale]
{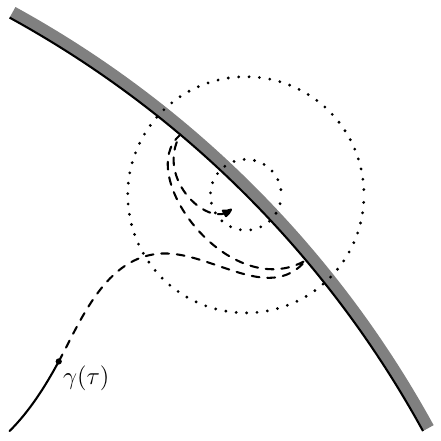}
} 
\hspace{0.5cm}
\subfigure[Crossing without any branching points]
{
	\label{sfig: f crossing b}
	\includegraphics[scale=\myscale]
{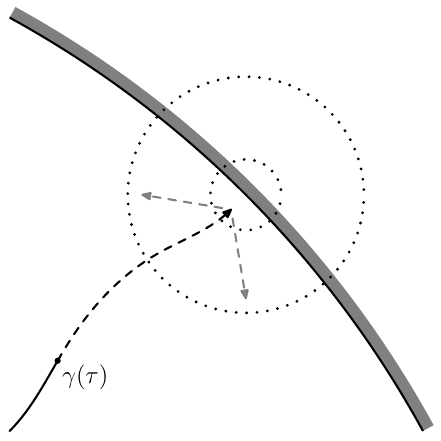}
}
\subfigure[Crossing with branching points only on its left]
{
	\label{sfig: f crossing c}
	\includegraphics[scale=\myscale]
{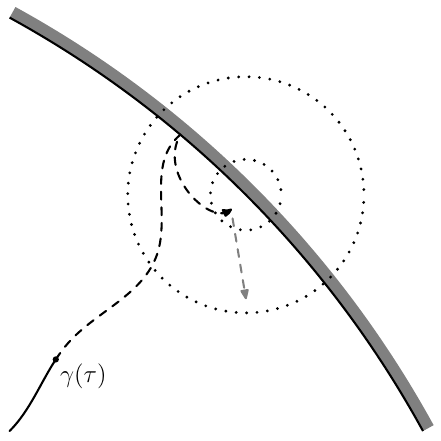}
} 
\hspace{0.5cm}
\subfigure[Crossing with branching points only on its right]
{
	\label{sfig: f crossing d}
	\includegraphics[scale=\myscale]
{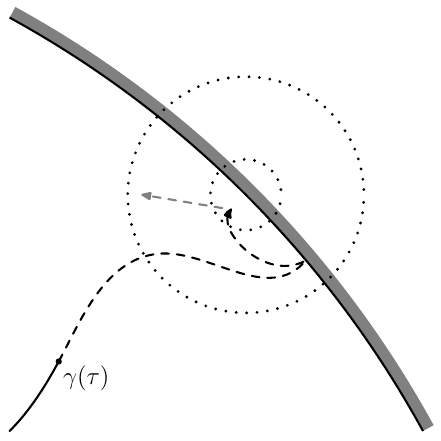}
}
\caption{Condition~\ref{cond: annulus} or \ref{cond: annulus exp} imply that the crossing events of any of the types
illustrated in the figures \subref{sfig: f crossing b}, \subref{sfig: f crossing c} and \subref{sfig: f crossing d}
has small probability. The longer black arrow is the crossing event considered in \eqref{eq: cond annulus} and
the shorter gray arrows are the crossings of the annulus that are still possible afterwards.}
\label{sfig: f crossing}
\end{figure}

We want to use Condition~\ref{cond: annulus} or equivalently \ref{cond: annulus exp} as a hypothesis for theorems.
We start by verifying them for the  critical FK Ising model exploration tree.

\begin{theorem}\label{thm: FK Ising and condition}
If $\Sigma$ is the collection of pairs $(\phi,\P)$ where $\phi$ 
satisfies the properties given in 
\eqref{eq: def conf map setting}
and also that 
its domain of definition $\Omega(\phi)$
is a discrete domain with some lattice mesh,
and $\P$ is the law of the  critical FK Ising model exploration tree $\rndttr$ on $U(\phi)$,
then $\Sigma$ satisfies Conditions~\ref{cond: annulus} and \ref{cond: annulus exp}.
\end{theorem}

\begin{proof}
The theorem can be proved in the same way as the result that a single interface in FK Ising model 
satisfies a similar condition, which was presented in Section~4.1 of \cite{Kemppainen:2012vm}. 
However, stronger crossing estimates are needed for the exploration tree
compared to a single interface.
Luckily such estimates were established in
Theorem~1.1 of \cite{Chelkak:2016jw}. 
The full argument goes as follows
\begin{itemize}
\item Similarly as in \cite{Kemppainen:2012vm}, 
we try to bound uniformly from above the probability of crossings by the interface in an annular sector.
This bound can achieved by given a uniform lower bound for open or dual open paths of edges
in the random cluster model in the transversal direction.
By symmetry we can suppose that we consider open crossings. 
\item By a similar arguments as in Section~4.1 of \cite{Kemppainen:2012vm},
we can use FKG inequality to reduce it to a question of open crossing of 
a topological quadrilateral. See also Figures~12 and 13 in  \cite{Kemppainen:2012vm}.
We can move the wired boundary where the crossing starts and
introduce free boundary along the two sides which are parallel to the possible open crossing. 
However we
cannot move the free boundary or replace it by wired boundary if the boundary condition
at the endpoint of the possible open crossing is indeed free. Thus we need to
consider a general topological quadrilateral and not just regular one (with 
an archetypical shape), which was enough in \cite{Kemppainen:2012vm}. 
\item As indicated by considerations in Figure~\ref{sfig: crossing annuli boundary conditions b},
we end up to  wired-free-free-free 
or wired-free-wired-free
boundary conditions after
the FKG transformation. (That is, we introduce new boundary along the boundaries
of the annulus of the opposite ``color'' as the dashed crossing of the quadrilateral,
in the sense that if we consider open crossing the new boundary is free and
if the crossing crossing is dual open the new boundary is wired. Then we use the 
duality as we stated above.)
\item Then we use the crossing estimates of \cite{Chelkak:2016jw} to reduce
an uniform lower bound. By Theorem~1.1 of \cite{Chelkak:2016jw},
the lower bound (as well as the upper bound) are uniform and depend only on
the extremal length of the topological quadrilateral.
\end{itemize}
The upper bound for existence
of a curve can improved to the form $C \left( \frac{r}{R} \right)^\alpha$
using a similar argument as Proposition~2.6 of \cite{Kemppainen:2012vm}
by considering a disjoint set of concentric annuli where an upper
bound of the form $1-\eps$ holds.
\end{proof}

\begin{figure}[tb]
\centering
\subfigure[The different types of locations of annuli needed to be considered in Theorem~\ref{thm: FK Ising and condition}.] 
{
	\label{sfig: crossing annuli boundary conditions a}
	\includegraphics[scale=1.3]
{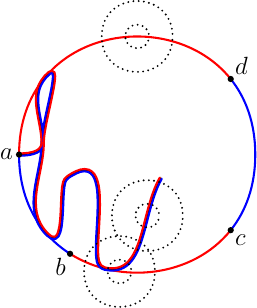}
} 
\hspace{0.5cm}
\subfigure[For each annulus type,
we need to establish a probability lower bound on the crossing events of connected random cluster paths
indicated by the dashed lines to get
the bound of Condition~\ref{cond: annulus exp}.]
{
	\label{sfig: crossing annuli boundary conditions b}
\parbox[b]{.5\textwidth}{\centering%
	\includegraphics[scale=.8]
{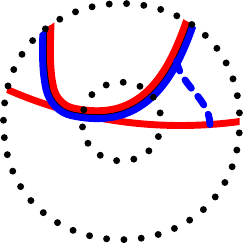}\hspace{0.5cm}
	\includegraphics[scale=.8]
{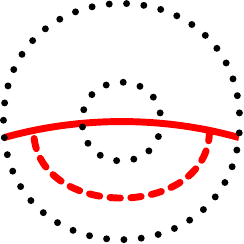}\\
	\includegraphics[scale=.8]
{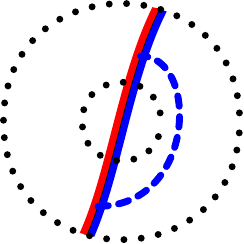}
} }
\caption{Crossing estimates needed for Theorem~\ref{thm: FK Ising and condition}. The setup with four marked points is the
same as in Figure~\ref{fig: domain in general case} in Section~\ref{sec: observable}.}
\label{fig: crossing annuli boundary conditions}
\end{figure}

As shown in 
\cite[Proposition 2.6]{Kemppainen:2012vm},
this type of bounds behave well under conformal maps. We have uniform
control on how the constants in Conditions~\ref{cond: annulus} and \ref{cond: annulus exp} change
if we transform the random objects conformally from one domain to another.

Given a collection $\Sigma$ of pairs $(\phi,\P)$ 
we define the collection $\Sigma_\disc = \{\phi\P \,:\, (\phi,\P)\in\Sigma \}$
where $\phi\P$ is the pushforward measure defined by $(\phi\P) (E) = \P(\phi^{-1}(E))$. 

\begin{theorem}\label{thm: FK Ising and condition conf}
If $\Sigma$ is as in Theorem~\ref{thm: FK Ising and condition},
then $\Sigma_\disc$ 
satisfies Conditions~\ref{cond: annulus} and \ref{cond: annulus exp}.
\end{theorem}

\subsection{Uniform approximation by finite subtrees}\label{ssec: approximating by finite trees}

\subsubsection{Topology on trees}\label{sssec: topo trees}

Let $\de_\text{curve}$ to be the curve distance 
(defined to be the infimum over all reparametrizations of the supremum norm of the difference).
We define the space of 
trees as the space of closed subsets of the space of curves
and we endow this space with the Hausdorff metric.
More explicitly, if $(\bran_x)_{x \in  S}$ and $(\hat{\bran}_{\hat{x}})_{\hat{x} \in  \hat{S}}$
are trees then the distance between them is given by
\begin{equation}
\de_\text{tree}\left( (\bran_x)_{x \in  S}\,,\,(\hat{\bran}_{\hat{x}})_{\hat{x} \in  \hat{S}}\right)
  = \max \left\{ \sup_x \inf_{\hat{x}} \de_\text{curve}\left(\bran_x,\hat{\bran}_{\hat{x}}\right) \,,\,
      \sup_{\hat{x}} \inf_x \de_\text{curve}\left(\bran_x,\hat{\bran}_{\hat{x}}\right) \right\} .
\end{equation}

\subsubsection{Uniform approximation by finite subtrees}\label{sssec: approximating by finite trees}

Let us divide the boundary of the unit disc into a finite number of connected arcs $I_k$,
$k=1,\ldots,N$,
which are disjoint, except possibly at their end points. Let $x^\pm_k$ be the
end points of $\overline{I_k}$. Suppose that $-1 = \phi(\vroot)$ is an end point.
Naturally it is then an end point of two arcs.

Let $\mathcal{I} = \{ I_k \,:\, k=1,\ldots,N \}$ and denote
the maximum of the diameters of $I_k$ by $m(\mathcal{I})$ and $\{ x^\pm_k \,:\, k=1,\ldots,N\}$ by $\hat{\mathcal{S}}^\disc$. 
Consider now the random tree $\tree_\delta = (\rndbran_{\vroot,x})_x$ whose law is given by $\P_\delta$. 
The finite subtree $\tree_\delta(\mathcal{I})$ is defined 
by the following steps:
\begin{itemize}
\item first take 
a
discrete approximation
$\hat{\mathcal{S}}_\delta$ of the set of points $\phi^{-1}(\hat{\mathcal{S}}^\disc)$ 
on the vertex set of $\partial_1 \Omega_\delta$.
\item then consider the finite subtree $\hat{\tree}_\delta = (T_{\vroot,x})_{x \in \hat{\mathcal{S}}_\delta}$
\item finally set $\mathcal{S}_\delta$ to be the union of $\hat{\mathcal{S}}_\delta$ and all the branching points
(in the sense of the definition in Section~\ref{ssec: tree crossing})
of $\hat{\tree}_\delta$. Denote $(T_{\vroot,x})_{x \in \mathcal{S}_\delta}$ by $\tree_\delta$.
Here $T_{\vroot,x}$ for a branching point is defined as the subpath of $(T_{\vroot,x})_{x \in \hat{\mathcal{S}}_\delta}$
that starts from $\vroot$ and terminates at $x$.
\end{itemize}
In other words, we take the subtree corresponding to $\hat{\mathcal{S}}$ and then we augment it 
by adding all its branching points and the branches ending at those branching points to the tree.
We will call below $\tree_\delta=\tree_\delta(\mathcal{I})$ the 
\emph{finite subtree} corresponding to $\mathcal{I}$. It is finite in a uniform way over
the family of probability laws and domains of definition. Hence the name.

Denote the image of $\tree_\delta(\mathcal{I})$ under $\phi$ by $\tree^\disc_\delta(\mathcal{I})$.

\begin{theorem}\label{thm: uniform tree approximation}
For each $R>0$
\begin{equation}
\sup \P_\delta \left[ \,\de_\text{tree}\left(\tree^\disc_\delta(\mathcal{I})
  \,,\, \tree^\disc_\delta\right)\, > R \right] = \oo(1)
\end{equation}
as $m(\mathcal{I}) \to 0$ 
where the supremum is taken over $\delta>0$.
\end{theorem}

\begin{remark}
In fact, the supremum in the claim can be taken to be over all shapes $\Omega$,
since probability bounds that the proof is based on hold for any shape.
\end{remark}

\begin{figure}[tb]
\centering
\subfigure[The returning path doesn't branch on $x_+ \vroot$ before branching in $x x_+$
and then the path towards $x_+$ reaches distance at least $3R/4$ from $x_-x_+$ before reaching $x_+$.] 
{
	\includegraphics[scale=.35]
{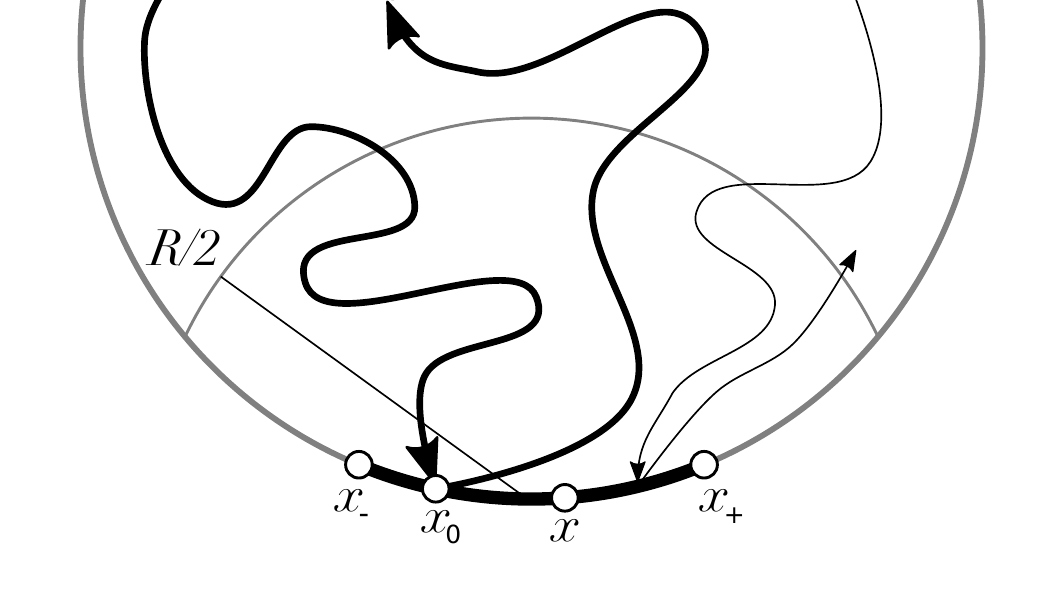}
} 
\hspace{0.5cm}
\subfigure[The returning path branches in $x x_+$ and
then the path towards $x$ reaches distance $3R/4$ from $x_-x_+$ before reaching $x$.]
{
	\includegraphics[scale=.35]
{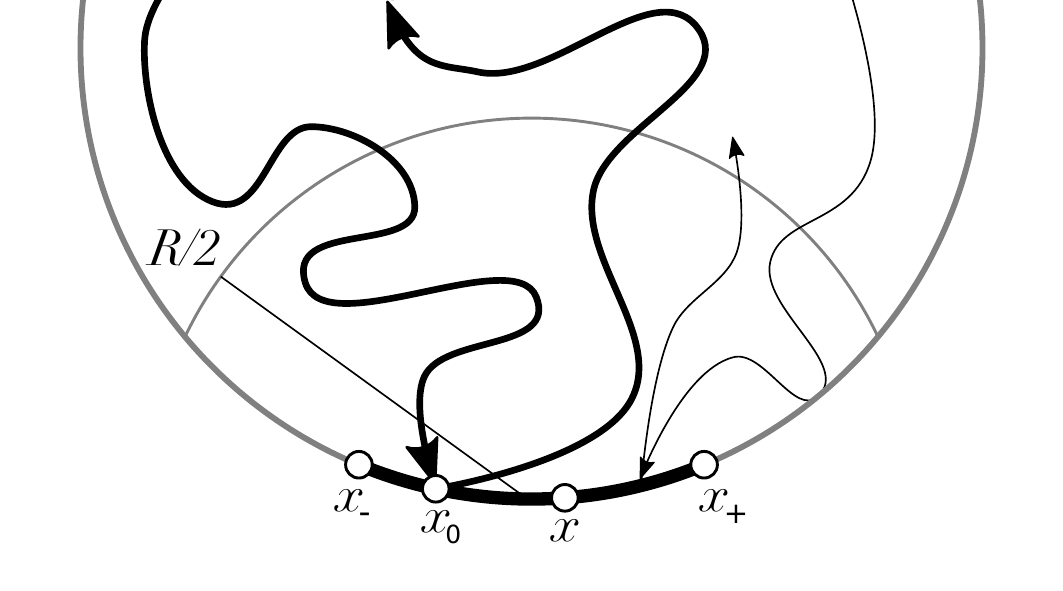}
}
\caption{
The events which are shown to have low probability in the proof of
Theorem~\ref{thm: uniform tree approximation}. Two thick arrows are the initial
segment of the exploration which visits the boundary first time in $x_- x_+$ at $x_0$
and then exits to distance $R/2$. The thin arrows indicate the events whose probability is estimated.
}
\label{fig: crossing annuli 1}
\end{figure}

\begin{proof}
Let $\eps>0$ and suppose that $K$ and $\alpha$ are positive numbers such that
$M_r \leq K r^{-\frac{1}{\alpha}}$ for all $r>0$ with probability greater than $1-\eps$
for any $\P_\delta$. Such constants exist by Theorem~\ref{thm: aizenman--burchard bounds}.

Let $x \in \ttrtarget$. Take $k$ such that $x \in I_k$. Set
$x_\pm = x^\pm_k$.
We will show that with high probability $T_{\vroot,x}$ is close to either
$T_{\vroot,x_-}$, $T_{\vroot,x_+}$ or some $T_{\vroot,\tilde{x}}$ where $\tilde{x}$ is a branch point
(in the sense of the definition in Section~\ref{ssec: tree crossing})
on the arc $x_- x_+$. We will show that with a high probability
this happens for all $x$ simultaneously.

Fix $R>0$. We will study under which circumstances
\begin{equation}\label{eq: middist def}
\min\left\{ \de (T_{\vroot,x}, T_{\vroot,\tilde{x}}) \,:\, \tilde{x} \in x_- x_+\right\} 
\end{equation}
is larger than $R$ or smaller or equal to $R$.
We can assume that $m(\mathcal{I})$ is less than $R/10$, say.
Then the diameter of $x_-x_+$ is at most $R/10$.

First we notice, that if $x_0$ the first branching point on $I_k$, then the branches to $x,x_-$ and $x_+$
are all the same until they reach $x_0$.

We will consider two complementary cases: either $x_0 \in x_-x$ or $x_0  \in x x_+$.

If $x_0 \in x_-x$, the branches of $x$ and $x_+$ continue to be the same after $x_0$ 
at least for some time.
Suppose that 
$\min\{\de (T_{\vroot,x}, T_{\vroot,x_0}),\de (T_{\vroot,x}, T_{\vroot,x_+})\}>R$. 
Then in particular the branch that we follow to continue towards $x_-$ and $x_0$
has to reach distance $3R/4$ from $x_-x_+$ without making a branch point in $x \vroot$,
i.e. a visit to the boundary.
It has to contain therefore a subpath which has diameter at least $R/2$, 
starts at the boundary and  is
otherwise disjoint from the boundary. Call the subpaths with this property
$\mathcal{J}$. 
Then by the bound on $M_r$ we get $\#\mathcal{J} \leq K (R/2)^{-\frac{1}{\alpha}}$
(see Lemma~2.2 in \cite{Aizenman:1999dt} for the relationship of maximal number of disjoint segments
and the minimal number of segments needed to cover a path).

If we happen to reach distance $3R/4$ from $x_-x_+$ without making a branch point,
then in order $\de (T_{\vroot,x}, T_{\vroot,x_+})>R$ to hold the path needs to come close to $x_- x_+$ again
and branch on 
$x x_+$
which it will do surely.  
But while doing so one of the following has to hold (Figure~\ref{fig: crossing annuli 1}):
\begin{itemize}
\item 
the returning path doesn't branch on $x_+ \vroot$ before branching in $x x_+$
and then the path towards $x_+$ reaches distance at least $3R/4$ from $x_-x_+$ before reaching $x_+$
\item 
or the returning path branches in $x x_+$ and
then the path towards $x$ reaches distance $3R/4$ from $x_-x_+$ before reaching $x$.
\end{itemize}
By the conditions presented in the previous section the probabilities of both of them can be made less
than $\eps$ by choosing length of $x_-x_+$ small.

Now there are $\#\mathcal{J}$ subpaths where this can happen. 
Hence
the total probability is at most
\begin{equation}\label{ie: approximation by finite trees estimate}
\eps + K (R/2)^{-\frac{1}{\alpha}} \eps = cst.(R) \eps
\end{equation}
which can be made arbitrarily small. 
This argument
can be made more formal by introducing stopping times $\tau_n$ such that 
$\tau_0=0$ and  $\tau_n$, $n \geq 1$ is the smallest $t$ such that
$t> \tau_{n-1}$ and it holds that there exists $s \in [\tau_{n-1},t)$
such that the exploration at time $s$ in on the boundary and
the exploration on $(s,t)$ is disjoint from the boundary and has diameter at least $R/2$.
Then if $E_n$ is the event above (described in the two bullets), then 
the inequality~\eqref{ie: approximation by finite trees estimate} follows
simply by using a union bound for $\bigcup_n E_n$.%

The other case $x_0  \in x x_+$ is similar.
\end{proof}

\subsection{Precompactness of the loops and recovering them from the tree}

\subsubsection{Topology for loops}

Let $\mathbb{T}$ be the unit circle.
We consider a loop to be a continuous function on $\mathbb{T}$, considered modulo orientation
preserving reparametrizations of $\mathbb{T}$. 
The metric on loops is given in a similar fashion as for curves. We define
\begin{equation}
\delp(\gamma,\hat{\gamma}) = \inf \left\{ \left\Vert f-\hat{f} \right\Vert_\infty 
   \,:\, f\in \gamma, \; \hat{f}\in \hat{\gamma} \right\} 
\end{equation}
where we use the notation $f\in \gamma$ to denote that $f$ is a parametrization
of the (unparametrized) curve $\gamma$.
The space of loops endowed with $\delp$ is a complete and separable metric space.

A loop collection $\rndlpe$ is a closed subset of the space of loops.
On the space of loop collections we use the Hausdorff distance.
Denote that metric by $\de_{\text{LE}}$, where LE stands for \emph{loop ensemble}.

The random loop configuration we are considering is the collection of FK Ising loops which
we orient in clockwise direction.

\subsubsection{Precompactness of the loops}

\begin{theorem}\label{thm: precompactness loop collections}
The family of probability laws of $\rndlpe$ is tight in the metric space of loop collections.
\end{theorem}

\begin{proof}
The claim follows directly from Lemma~\ref{lem: tree to loop ensemble} and 
the tightness of the trees.
Namely, observe that there exists a constant $\alpha>0$ and
a tight random variable $K>0$
such that the minimum number of segments needed to cover any branch in the tree
with segments of diameter at most $R$ is bounded by $K \,R^{-\alpha}$.
By the bijection between loop collections and trees, each loop is
a subpath of a branch of the tree and thus we find that
the minimum number of segments needed to cover any loop in the loop collection
with segments of diameter at most $R$ is bounded by $K \,R^{-\alpha}$.
\end{proof}

\subsection{Uniform approximation of loops by the finite subtrees}\label{ssec: subtree approximation}

By Lemma~\ref{lem: tree to loop ensemble} it is possible to reconstruct the loops from the full tree.
And by Theorem~\ref{thm: uniform tree approximation} we can approximate the full tree
by finite subtrees.
How do we recover loops approximately from the finite subtrees? Can we do it in a uniform manner?

Consider a loop $\rndlp \in \rndlpe$.  
We divide it into \emph{arcs} which are excursions from the boundary to the boundary, otherwise disjoint
from the boundary.
Since the loop is oriented in the clockwise direction, exactly one of the arcs goes away from $\vroot$
in the counterclockwise orientation of the boundary and rest of them go towards $\vroot$. The first
case is the \emph{top arc} $\gamma^\itop$ of the loop $\gamma$ and the rest of them are the \emph{bottom arcs}.
See Figure~\ref{fig: schematic loops}. The concatenation of the bottom arcs is denoted by $\gamma^\ibot$.

Suppose that the loop doesn't intersect a fixed neigborhood of $\vroot^\disc$.
The diameter of the entire loop is bounded by a uniform constant (depending on the chosen neighborhood)
times the diameter of the top arc and vice versa. This means that 
if the loop has diameter greater than a fixed number, then
the top arc is traced by finite subtree for
fine enough mesh according to Theorem~\ref{thm: uniform tree approximation}.

Also the bottom arcs from right to left are traced until the last branch point to the points
defining the finite subtree. Again by Theorem~\ref{thm: uniform tree approximation}, the part which
remains to be discovered of the loops, has small diameter. Hence if we define approximate loop
to have just a linear segment in that place, we see that the loop and the approximation are close
in the given metric.

So we define the finite-subtree approximation $\rndlpe_\delta$ of the random loop collection $\rndlpe$ to be
the collection of all those discovered top arcs concatenated with their discovered bottom arcs
and the linear segment needed to close the loop. By Theorem~\ref{thm: uniform tree approximation},
we have the following result.

\begin{theorem}\label{thm: uniform loop approximation}
For each $R>0$
\begin{equation}
\sup_{ \delta} \P_\delta \left[ \,\de_{\textnormal{LE}}\left(\rndlpe^\disc_\delta(\mathcal{I})
  \,,\, \rndlpe^\disc\right)\, > R \right] = \oo(1)
\end{equation}
as $m(\mathcal{I}) \to 0$.
\end{theorem}

\subsection{Some a priori properties of the loop ensembles}

The following theorem gathers some technical estimates needed below.

\begin{theorem}\label{thm: a priori properties of le}
The family of critical FK Ising loop ensemble measures $(\P_\delta)_{\delta>0}$ satisfy
the following properties
\begin{enumerate}\enustyiii
\item \label{enui: a priori properties of le i}
(finite number of big loops) For each $R>0$,
\begin{equation}\label{eq: a priori no of big loops}
\sup_{\delta>0} \P_\delta \left( 
  \# \left\{ \gamma \,:\, \diam(\gamma) \geq R \right\} \; \geq N
\right) = \oo(1)
\end{equation}
as $N \to \infty$.
\item \label{enui: a priori properties of le ii}
(small loops and branching points are dense on the boundary)
There exists $\Delta>0$ such that for each $x \in \partial \disc$ and $R>0$
\begin{equation}\label{eq: a priori small loops are dense}
\sup_{\delta>0} \P_\delta \left( 
  \sup \left\{ \diam(\gamma) \,:\, \gamma \cap B(x,r) \neq \emptyset \right\} \; \geq R
\right) = \OO\left( \left( \frac{r}{R} \right)^\Delta \right)
\end{equation}
as $r \to 0$.
%
Consequently, for all $c >1$
\begin{equation}\label{eq: a priori small loops are dense ii}
\P_\delta \left( 
  \inf \left\{ \diam(\gamma) \,:\, \gamma \cap B(x,c \delta) \neq \emptyset \right\} \; \geq R
\right) = \OO\left( \left( \frac{\delta}{R} \right)^\Delta \right)
\end{equation}
for all $\delta>0$ and $R> c \delta$, and
thus for any $\beta \in (0,1)$ and $r>0$ and any partitioning of the boundary
of the domain to connected sets $I_j$ of diameter at least $r$ it holds that
\begin{equation}\label{eq: a priori small loops are dense iii}
\P_\delta \left( \forall j, \;
  I_j \text{ is touched by a loop with diameter at most } \delta^\beta
\right) = 1 - \OO\left(  \delta^{\Delta(1-\beta)} \right)
\end{equation}%
as $\delta \to 0$.
%
\item (big loops have positive support on the boundary and they touch the boundary infinitely often around the extremal points)
For any $R>0$
\begin{equation}\label{eq: a priori boundary support}
\sup_{\delta>0} \P_\delta \left( 
  \exists \gamma \text{ s.t. } 
  \diam(\gamma) \geq R  \text{ and } 
  |\pleft(\gamma) - \pright(\gamma)| < r 
\right) = \oo(1)
\end{equation}
as $r \to 0$. 
Here $\pleft(\gamma)$ and $\pright(\gamma)$ are the left-most and the right-most points of $\gamma$
along the boundary of the domain
(as seen from the root).
We call the boundary arc $\pleft(\gamma)\pright(\gamma)$ 
the \emph{support of $\gamma$ on the boundary}.
Furthermore, for any constant $0<\eta<1$ 
there exists a sequence $\delta_0(m)>0$ and constant $\lambda>0$  such that 
if 
$x(\gamma)=\pleft(\gamma)$ or $x(\gamma)=\pright(\gamma)$
\begin{equation}\label{eq: a priori boundary visits}
\sup_{\delta \in (0, \delta_0(m))} \P_\delta \left( 
  \begin{gathered}
  \exists \gamma \text{ s.t. } 
  \diam(\gamma) \geq R  , 
  |\pleft(\gamma) - \pright(\gamma)| \geq r \text{ and } \\
  \# \{ n=1,2,\ldots,m \,:\, \gamma \text{ touches boundary} \\
          \text{in } A(x(\gamma),\eta^n \, r, \eta^{n-1} \, r)   \} \leq \lambda m 
  \end{gathered}
\right) = \oo(1)
\end{equation}
as $m \to \infty$.
\item (big loops are not pinched)
\begin{equation}\label{eq: a priori big loops are fat top}
\sup_{\delta>0} \P_\delta \left( 
  \begin{gathered}
  \exists \gamma \text{ s.t. } \gamma^\itop = \gamma_1^\itop \sqcup \gamma_2^\itop, \, \gamma_1^\itop \cap \gamma_2^\itop = \{x\}\\
  \diam(\gamma_k^\itop) \geq R \text{ for }  k=1,2 \text{ and } \dist(x, \partial \disc) < r  \\
  \end{gathered}
\right) = \oo(1)
\end{equation}
and
\begin{equation}\label{eq: a priori big loops are fat bot}
\sup_{\delta>0} \P_\delta \left( 
  \begin{gathered}
  \exists \gamma \text{ and } \gamma' \subset \gamma^\ibot \text{ s.t. }
    \gamma' = \gamma_1' \sqcup \gamma_2', \, \gamma_1' \cap \gamma_2' = \{x\}\\
  \gamma' \cap \partial_1 \Omega =\emptyset, \; 
  \diam(\gamma_k') \geq R \text{ for }  k=1,2  \text{ and } \dist(x, \partial \disc) < r  \\
  \end{gathered}
\right) = \oo(1)
\end{equation}
as $r \to 0$. 
\item (big loops are 
distinct)
\begin{equation}\label{eq: a priori distinguishable}
\sup_{\delta>0} \P_\delta \left( 
  \exists \gamma, \tilde{\gamma} \text{ s.t. } 
  \diam(\gamma),\diam(\tilde{\gamma}) \geq R  \text{ and } 
  \de_{\text{loop}}(\gamma,\tilde{\gamma}) < r 
\right) = \oo(1)
\end{equation}
as $r \to 0$.
\end{enumerate}
\end{theorem}

\begin{proof}
The bound \eqref{eq: a priori no of big loops} follows from Theorem~\ref{thm: uniform tree approximation}.
The bound \eqref{eq: a priori small loops are dense} 
follows directly from the crossing bound of annuli
at $x$, and the bound~\eqref{eq: a priori small loops are dense ii}
is merely rephrasing the previous bound and specializing to $r= c\delta$.
For the bound~\eqref{eq: a priori small loops are dense iii}, take
$x_j$ to be any point on $I_j$, say, not too close to the endpoints
of $I_j$, and 
use the bound~\eqref{eq: a priori small loops are dense ii}
for $x=x_j$, $R=\delta^\beta$ and sum over $j$ to get the an upper bound for an
union of events of type in \eqref{eq: a priori small loops are dense ii};
the required bound is obtained as the complement.

The bound \eqref{eq: a priori boundary visits} is shown to hold by considering the domain 
$U \setminus \underline{\gamma}(\tau)$
where $\tau$ is a stopping time such that the top arc of a big loop $\gamma$ is reaching its endpoint $x(\gamma)$ at time $\tau$, 
and using the crossing bound of annuli at $x(\gamma)$.

The proofs of \eqref{eq: a priori boundary support}, \eqref{eq: a priori big loops are fat top}
and \eqref{eq: a priori big loops are fat bot} are similar. 
The exploration process discovers, as seen from the root, the top
arc of any loop before the lower arcs.
As noted before, the diameter of the top arc is comparable to the diameter
of the whole loop. Hence it enough to work with loops that have top-arc diameter more than $R$.

For \eqref{eq: a priori boundary support}, stop the process when the current arc exits the ball of radius
$R/4$ around the starting point $z_0$ of the arc.
The rest of the exploration process makes with high probability a branching point
in the annulus $A(z_0,r,R/4)$
before finishing the loop at $z_0$
by the crossing estimate (namely, the estimate~\eqref{eq: cond annulus exp} 
used in $A(z_0,r,R/4)$ implies that the path will touch both sides
of the of the slit domain boundary inside the annulus 
before reaching distance $r$ from $z_0$).
This implies that the event \eqref{eq: a priori boundary support}
does not occur for that loop. These segments of diameter $R/4$ in the exploration process
are necessarily disjoint (they start on the boundary and remain after that disjoint from the boundary)
and hence their number is a tight random variable by Theorem~\ref{thm: aizenman--burchard bounds}. 
The bound \eqref{eq: a priori boundary support}
follows by a simple union bound.

For \eqref{eq: a priori big loops are fat top}, stop the process when the diameter
of the current arc is at least $R$ and the tip lies within distance $r$ from the boundary.
Let the point closest on the boundary be $z_0$. 
These arc of diameter $R$ are necessarily disjoint: they start from the boundary and remain after that
disjoint from the boundary.
The rest of the exploration makes with high probability a branching point before it exits the ball of radius $R$
centered at $z_0$. The bound \eqref{eq: a priori big loops are fat top}
follows by a simple union bound. 

The proof of \eqref{eq: a priori big loops are fat bot} is very similar and we omit the details.

For the bound \eqref{eq: a priori distinguishable}, let $\gamma_1$ and $\gamma_2$ be two loops of
diameter larger than $R$. Let $S$ be the set of points of $\gamma_2$ and let $\gamma_1^\itop$
be the top arc of $\gamma_1$ and let $\gamma_1^\ibot$ be the rest of $\gamma_1$.
Without loss of generality, we can assume that $\gamma_1^\itop$ separates $S$ from $\gamma_1^\ibot$ in $\disc$.
Otherwise we can exchange the roles of $\gamma_1$ and $\gamma_2$. Suppose that $\delp(\gamma_1,\gamma_2) < r$.
Then we claim that the Hausdorff distance of $\gamma_1^\itop$ and $\gamma_1^\ibot$ is less than $2 r$. 
The distance from any point of $\gamma_1^\ibot$ to $\gamma_1^\itop$ is less than $r$.
This follows when we notice that any line segment connecting $\gamma_1^-$ to $S$ intersects $\gamma_1^+$.
Due to topological reasons, since $\gamma_1$ and $\gamma_2$ are oriented and $\delp$ uses this orientation,
it also holds that the distance from  any point of $\gamma_1^\itop$ to $\gamma_1^\ibot$ is less than $2r$.
Namely, any point in $\gamma_1^\itop$ has a point of $\gamma_2$ in its $r$ neighborhood, which has distance less than $r$
to $\gamma_1^\ibot$. We can now use the exploration process to explore $\gamma_1^\itop$. Take a point $z_0 \in \gamma_1^\itop$
that has distance at least $\sqrt{rR}$ to the boundary. 
Such a point exists by \eqref{eq: a priori boundary support} 
and \eqref{eq: a priori big loops are fat bot}.
Now by the crossing property 
used in $A(z_0,r,\sqrt{rR})$ at the stopping time when the exploration
has explored fully $\gamma_1^\itop$,
with high probability the exploration of $\gamma_1^\ibot$ doesn't come close to $z_0$.
\end{proof}

\subsubsection{Some consequences}\label{sssec: a priori consequences}

In the discrete setting we are given a tree--loop ensemble pair. The tree and the loop ensemble
are in one to one correspondence as explained earlier. Recall that, given the loop ensemble,
the tree is recovered by the exploration process which follows the loops in counterclockwise
direction and jumps to the next loop at boundary points where the loop being followed turns away from
the target point. Recall also that the loops are recovered from the tree by noticing
that the leftmost point in the loop corresponds to a branching point of the tree and the rest of the 
loop is the continuation of the branch to the point just right of that branching point.

Consider any random tree--loop ensemble pair which is a subsequential weak
limit of the tree--loop ensemble pairs of the FK Ising model.
Since the discrete collections of loops are have finite number of big loops 
with the uniform bound \eqref{eq: a priori no of big loops},
the loop collection is 
almost surely 
at most countable also in the limit (use the Portmanteau theorem for the closed
event that there are at most $n$ loops of diameter strictly greater than $R$).
By the properties of the loop ensemble given in Theorem~\ref{thm: a priori properties of le},
the limiting pair and the process of taking the limit have the following properties
\begin{itemize}
\item the loops are distinguishable in the sense that
there is no sequence of pairs of 
distinct
loops that would converge to the same loop.
\item Each loop consists of a single top arc which is disjoint from the boundary except
at the endpoints and a non-empty collection of bottom arcs. In particular,
the endpoints of the top arc (the leftmost and rightmost points of the loop) are different.
\end{itemize}
From these properties we can prove the following result. The second assertion basically means that
there is a way to \emph{reconstruct the loops from the tree} also in the limit.

\begin{theorem}\label{thm:  tree to loops in the limit}
Let $(\rndlpe^\disc,\rndttr^\disc)$ be the almost sure limit of $(\rndlpe_{\delta_n}^\disc,\rndttr_{\delta_n}^\disc)$
as $n \to \infty$.
Write $\rndlpe^\disc =(\rndlp_j)_{j \in J}$ and $\rndlpe_{\delta_n}^\disc = (\rndlp_{n,j})_{j \in J}$
(with possible repetitions) such that almost surely for all $j \in J$, $\rndlp_{n,j}$ converges
to $\rndlp_j$ as $n \to \infty$, and then set $x_{n,j}$ to be the target point of the branch of 
$\rndttr_{\delta_n}^\disc$ that corresponds to $\rndlp_{n,j}$ in the above bijection 
(described in the beginning of the subsection~\ref{sssec: a priori consequences}).
Then 
\begin{itemize}
\item Almost surely all $x_{n,j}$ converge to some points $x_{j}$ as $n\to \infty$ and
all the branches $\rndbran_{x_{n,j}^+}$ converge to some branches denoted by $\rndbran_{x_{j}^+}$  as $n\to \infty$.
Furthermore, $x_{j}$ are distinct and 
they form a dense subset of $\partial \disc$ and $\rndttr$ is the closure of $(\rndbran_{x_{j}^+})_{j \in J}$.
\item On the other hand, $(\rndbran_{x_{j}^+})_{j \in J}$ is characterized as being the subset of $\rndttr$
that contains all the branches of $\rndttr$ that have a doublepoint on the boundary.
Furthermore, that doublepoint is unique and it is the target point (that is, endpoint) of that branch.
\item
Any loop $\rndlp_j$ can be reconstructed from $(\rndbran_{x_{j}^+})_{j \in J}$ in the following way:
the leftmost point of $\rndlp_j \cap \partial \disc$ (the point closest to the root in the counterclockwise direction)
is one of the doublepoints $x^+_j$
and the loop $\rndlp_j$ is the part between the first and last visit to $x^+_j$ by $\rndbran_{x_{j}^+}$.
\end{itemize}
\end{theorem}

\subsection{Precompactness of branches and description as Loewner evolutions}\label{ssec: a priori single branch}

\subsubsection{Precompactness of a single branch} 

We know now that the sequence of exploration trees is tight in the space of curve collections,
which has the topology of Hausdorff distance on the compact sets of the space of curves.
This enables us to choose for any subsequence a convergent subsequence. However, 
it turns out that we need stronger tools to be able to characterize the limit.
We will review the results of \cite{Kemppainen:2012vm} that we will use.

The hypothesis  of \cite{Kemppainen:2012vm} is similar to the conditions in Section~\ref{ssec: tree crossing}.
Once that hypothesis holds for a sequence of random curves, it is shown in that paper that the sequence
is tight in the topology of the space of curves. Furthermore, it is established that such a sequence
is also tight in the topology of uniform convergence of driving terms of Loewner evolutions
in such a way that mapping between curves and Loewner evolutions is uniform enough so that
if a sequence converges in both of the above mentioned topologies the limits have to be the same.

The hypothesis  of \cite{Kemppainen:2012vm} for the FK Ising branch has been already established
since we can see it as a special case of Theorem~\ref{thm: FK Ising and condition}.

\begin{theorem}[Kemppainen--Smirnov, \cite{Kemppainen:2012vm}]\label{thm: ks}
Under certain hypothesis, a sequence probability laws $\P_n$ of random curves of $\half$ 
has the following properties:
for each $\eps>0$, there exists an event $K$ such that $\inf_n \P_n(K) \geq 1 - \eps$ and
the capacity-parametrized curves in $K \cap \{\gamma \text{ simple}\}$   
form an equicontinuous family, their driving processes form an equicontinuous family
and finally $|\gamma(t)| \to \infty$ as $t \to \infty$ uniformly.
Moreover the driving processes on the event $K \cap \{\gamma \text{ simple}\}$
are $\beta$-H\"older continuous with a bounded H\"older constant for any $\beta \in (0, \frac{1}{2})$.
\end{theorem}

In addition to \cite{Kemppainen:2012vm}, see also Section~6.3 in \cite{Kemppainen:tb} 
for this type of argument in the case of site percolation.

\subsubsection{Precompactness of finite subtrees}\label{ssec: a priori several branches}

For fixed finite number of curves it is straightforward to generalize Theorem~\ref{thm: ks}.
In fact, the conclusions of Theorem~\ref{thm: ks} hold for any finite subtree that we considered in 
Section~\ref{ssec: subtree approximation}.

In the rest of the paper we use these tools available for us and aim to characterize
the scaling limits of finite subtrees of the exploration tree. If we manage to establish
the uniqueness of the subsequent scaling limit of those objects, then Theorem~\ref{thm: main theorem}
follows from tightness and from the approximation result, Theorem~\ref{thm: uniform loop approximation}.


\section{Preholomorphic martingale observable}\label{sec: observable}

\subsection{The setup for the observable}\label{ssec: setup}

It is natural to generalize the setup of the previous section to domains of type illustrated in 
Figure~\ref{fig: domain in general case}.

\begin{figure}[tbh]
\centering
	\includegraphics[scale=.52]
{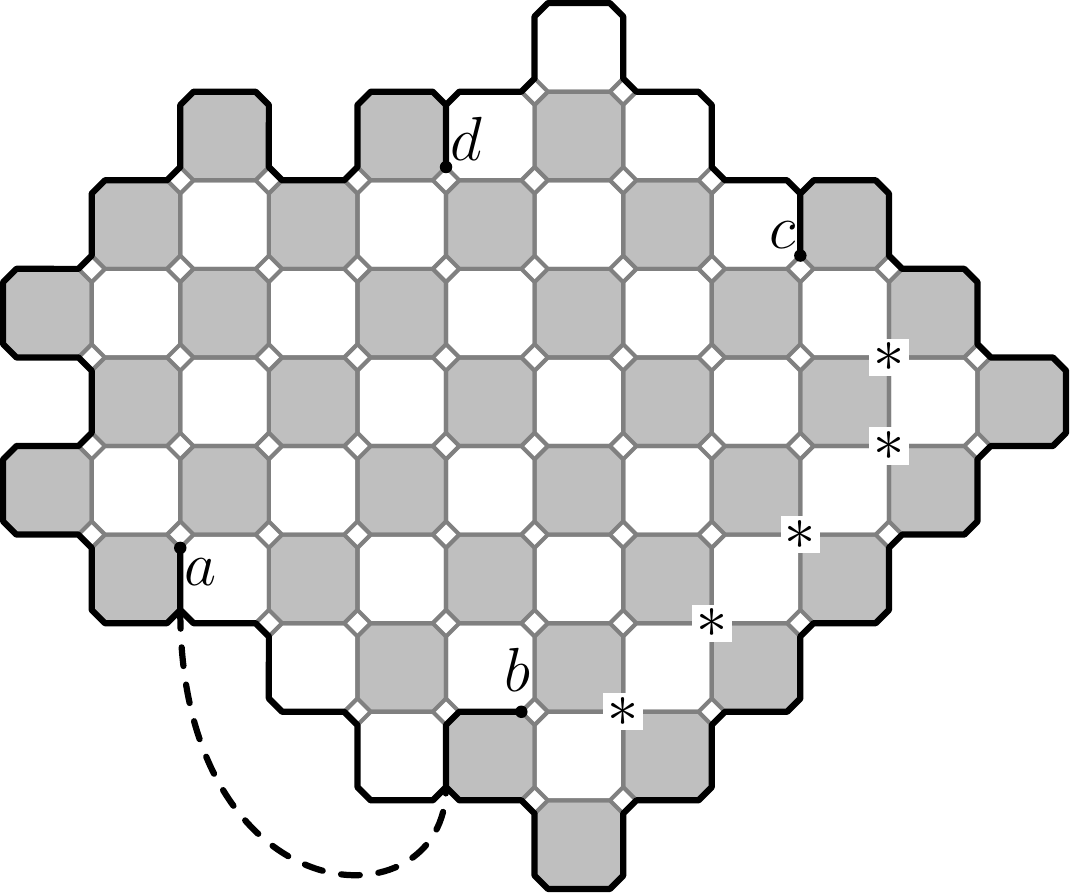}
\caption{Generalized setting for the chordal exploration tree. 
We add an external arc from $b$ to $a$ with the following interpretation.
The observable contains an indicator factor for the curve starting from $c$ and a complex factor depending
on the winding of that curve. If that curve goes to $b$ we continue to follow it through the external arc to
$a$ and from there to $d$. On the other hand, if the curve starting from $c$ ends directly to $d$, then
the curve connecting $a$ to $b$ is counted as a loop giving an additional weight $\sqrt{2}$ to the configuration.
The arc $bc \subset V_{\partial,1}$ is marked with $*$'s.
Notice that the general domain can have common parts for different boundary arcs longer than
one lattice step such as near $b$. In fact, we need to consider domains in this generality, if we wish to explore the interfaces starting
at the marked points and the corresponding observables
conditioned on the information of the progressing exploration.%
}
\label{fig: domain in general case}
\end{figure}

Henceforth we consider $G^\ddiamond$ with four special boundary vertices
$a,b,c,d$, where the boundary edges (when consider as edges of the directed graph $G^\ddiamond_\orient \subset \oddmdlattice$) 
at $a$ and $c$ point inwards and the boundary edges at $b$ and $d$ outwards. 
Now $a$ and $d$ play the roles of $v_\rt$ and $w$, respectively, in the construction of the exploration
tree of the previous section. 
The arcs $ab$ and $cd$ have white boundary ($\dsqlattice$ wired) 
and $bc$ and $da$ have black boundary ($\sqlattice$ wired).

To get a random cluster measure on $G^\ddiamond$ consistent with the all wired boundary conditions and the exploration process 
in Section~\ref{sec: setup}, we have to count the wired arcs $bc$ and $da$ to be in the same cluster. This corresponds
to the external arc configuration where $a$ and $b$ are connected by an arc and $c$ and $d$ are connected by an arc.
Denote such \emph{external arc pattern} $\extarcp{a}{b}{c}{d}$. Later we will denote \emph{internal arc patterns} by $\intarcp{a}{b}{c}{d}$ etc.

In fact, we will choose not to draw the external arc $\extarc{c}{d}$. The reason for this is that it is not used in the definition
of the observable and the weights for loop configurations that we get either with or without $\extarc{c}{d}$ are all proportional
by the same $\sqrt{2}$ factor. Thus it doesn't change the probability distribution.

The configuration on $G^\ddiamond$ is $(\gamma_1,\gamma_2,l_i \,:\, i=1,2,\ldots N_\text{loops})$,
where $\gamma_1$ and $\gamma_2$ are the paths starting from $a$ and $c$, respectively.
Define 
\begin{equation}\label{eq: def hat gamma}
\hat{\gamma} = \begin{cases}
  \gamma_2   & \text{if $\gamma_2$ exits through $d$} \\
  \gamma_2 \sqcup \alpha \sqcup \gamma_1   & \text{if $\gamma_2$ exits through $b$}
\end{cases}
\end{equation}
where $\alpha$ is the planar curve that realizes the exterior arch $\extarc{a}{b}$ as 
in Figure~\ref{fig: domain in general case}
and $\sqcup$ denotes the concatenation of paths.
The first case in \eqref{eq: def hat gamma} is the internal arc pattern $\intarcp{a}{b}{c}{d}$ and
the second is $\intarcp{a}{d}{c}{b}$.

For a sequence of domains $G_\delta^\ddiamond \subset \delta \ddmdlattice$ with
$a_\delta,b_\delta,c_\delta,d_\delta$,
define the observable as
\begin{equation}
f_\delta ( e ) = \theta_\delta \E_\delta (  \ind_{e \in \hat{\gamma}} e^{- i \frac{1}{2} W(d_\delta,e)} ) 
\end{equation}
for any $e \in E(G^\diamond_\delta) \subset E(G^\ddiamond_\delta)$.
Here $W(d_\delta,e)$ is the winding from boundary edge of $d_\delta$ to $e$ along the reversal of $\hat{\gamma}$ and
$\theta_\delta$ (satisfying $|\theta_\delta|=1$) is a constant, whose value we specify later.
Notice that $f_\delta$ doesn't depend on the choice of $\alpha$ since the winding is well-defined modulo $4 \pi$.

\subsection{Preholomorphicity of the observable}\label{ssec: preholo}

\begin{figure}[tbh]
\centering
{
	\includegraphics[scale=.8]
{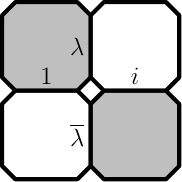}
} 
\hspace{1cm}
{
	\includegraphics[scale=.8]
{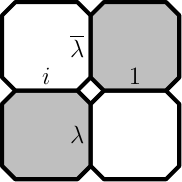}
}
\caption{The associated lines $\R, i \R, \lambda \R, \overline{\lambda} \R$ around the two different types of vertices
of $\dmdlattice$. Note that for a directed edge $e \in \odmdlattice$ the line in the complex plane is $\sqrt{\overline{e}} \, \R$ where
$e$ is interpreted as complex unit vector in its direction.}
\label{fig: complex factor on the lattice}
\end{figure}

Let $\lambda=e^{-i\frac{\pi}{4}}$.
Associate to each edge $e$ of the modified medial lattice one of the following four lines through the origin
$\R, i \R, \lambda \R, \overline{\lambda} \R$ as in Figure~\ref{fig: complex factor on the lattice}.
Denote this line by $l(e)$.

\textbf{Spin preholomorphicity:}
Choose the constant $\theta$ in the definition of $f_\delta$ so that 
the value of $f_\delta$ at the edge $e$ belongs to the line $l(e)$. 
Then $\theta \in \{\pm 1, \pm i, \pm \lambda, \pm \overline{\lambda} \}$.
The $\pm$ sign mostly doesn't play any role, but in some situations it should be chosen consistently.
The observable $f_\delta$ satisfies the relation
\begin{equation}\label{eq: edge s-hol}
f_\delta(e_W) + f_\delta(e_E)=f_\delta(e_S) + f_\delta(e_N) 
\end{equation}
for every vertex $v$ of the medial graph, whose four neighboring edges in counterclockwise order are called $e_N,e_W,e_S,e_E$.
The relation is verified using the same involution among the loop configurations as in \cite{Smirnov:2010ie}.

Therefore, using \eqref{eq: edge s-hol},
we can define $f_\delta(v) \dd= f_\delta(e_W) + f_\delta(e_E)=f_\delta(e_S) + f_\delta(e_N)$ and it satisfies
for any neighboring vertices $v,w$ the identity
\begin{equation}\label{eq: proj s-hol}
\Proj_e (f_\delta(v)) = \Proj_e (f_\delta(w)) 
\end{equation}
where $e$ is the edge between $v$ and $w$ and $\Proj_e$ is the orthogonal projection to the line $l(e)$.
That is, if $l(e)=\eta_e \R$ where $\eta_e$ is a complex number with unit modulus, then
\begin{equation}
\Proj_e (z) = \real[z \,\eta_e^*] \,\eta_e = \frac{z + z^* \eta_e^2}{2} .
\end{equation}
Since $f_\delta$ on $V(G^\diamond_\delta)$ satisfies the relation~\eqref{eq: proj s-hol}, we call it 
\emph{spin-preholomorphic} (or
\emph{strongly preholomorphic}).
Spin-preholomorphic functions satisfy a discrete version
of the Cauchy--Riemann equations \cite{Smirnov:2010ie}.

\subsection{Martingale property  of the observable}

Let $\gamma = T_{a,d}$ where $T_{a,d}$ is the branch of the exploration tree constructed in Section~\ref{ssec: exploration tree definition}.
Let's parametrize $\gamma$ by lattice steps so that 
at integer values of the time, $\gamma$ is at the head of an oriented edge (in $\oddmdlattice$) between black and white octagon 
and at half-integer values of the time, it is at the tail of such an edge. 
Note that step between the times $t=k-1/2$ and $t=k$, $k =1,2,\ldots$, is deterministic given the information up to time $t = k-1/2$. 
Hence between two consecutive integer times, at most one bit of information is generated: 
namely, whether the curve turn left or right in an ``intersection''. Even this choice might be predetermined, if
we are visiting a boundary vertex or a vertex visited already before.

Denote the loop configuration by $\omega$, in the four marked point setting it consist of two arcs and a number of loops.
We note that the pair $(\gamma[0,t],\omega)$ can be sampled in two different ways
(which are basically describing the conditional distributions
of one given the other):
\begin{enumerate}
\item In the first option, we sample first $\omega$ and then $\gamma[0,t]$ is a deterministic function of $\omega$ as explained above.
\item In the second option, 
we sample first $\gamma[0,t]$ (or rather we keep the sample of $\gamma[0,t]$ of the previous construction)
and then we (re-) sample $\omega'$ 
(We can rename $\omega'$ in the end as $\omega$. The prime symbol is only used so that there is no confusion with $\omega$ defined above.)
in two steps: we sample $\omega''$ in the complement of $\gamma[0,t]$ 
using the boundary condition given by
$\gamma[0,t]$ and then 
for each visit of $\gamma[0,t]$ to the arc $bc \subset V_{\partial,1}$ 
we flip an independent coin $\zeta_v \in \{0,1\}$, $v \in bc$,
such that
\begin{equation}\label{eq: def zeta}
\P(\zeta_v=0) = \frac{1}{1 + \sqrt{2}} , \qquad \P(\zeta_v=1) = \frac{\sqrt{2}}{1 + \sqrt{2}}
\end{equation}
and then we open for each visited $v \in bc$ such that $\zeta_v=1$ the edge $e_v \in E(\sqlattice)$ in $\omega'' \cup \gamma[0,t]$
and call the resulting configuration $\omega'$. 
Notice that the numbers on the right-hand sides in
\eqref{eq: def zeta} are $1-p_c$ and $p_c$, repectively. The number $1-p_c$
is the weight of the configuration where the edge at $v$ is closed relative to the configuration
where the edge is open.
\end{enumerate}
See Table~\ref{tbl: martingale property} for details about the equivalence
of these two constructions of $(\gamma[0,t],\omega)$.
The configurations are grouped in Table~\ref{tbl: martingale property}
in pairs so that the configurations only differ by the state of the edge at $v$.
By construction,  
the right column occurs for $\gamma$ surely, since the left column leads to branching at $v$. 
The relative weights of the $\sqlattice$-edge at $v$ being open or closed 
are independent of 
are calculated in the table
giving \eqref{eq: def zeta}; namely, the relative weight of the pair of the configuration
is a constant.
Notice also that $bc$ mattered only here since when the branch hits $ab$ or $da$ (or even $cd$,
if we continue to process all the way up to $d$)
it continuous to follow the same ``loop''. Thus the loop-to-loop jump can only occur on $bc$.

\begin{table}[!ht]
\begin{tabular}{c|c|c}
 & 
\begin{minipage}[b]{5cm} 
\centering 
\def\svgwidth{5cm}
\begingroup%
  \makeatletter%
  \providecommand\color[2][]{%
    \errmessage{(Inkscape) Color is used for the text in Inkscape, but the package 'color.sty' is not loaded}%
    \renewcommand\color[2][]{}%
  }%
  \providecommand\transparent[1]{%
    \errmessage{(Inkscape) Transparency is used (non-zero) for the text in Inkscape, but the package 'transparent.sty' is not loaded}%
    \renewcommand\transparent[1]{}%
  }%
  \providecommand\rotatebox[2]{#2}%
  \ifx\svgwidth\undefined%
    \setlength{\unitlength}{431.94676208bp}%
    \ifx\svgscale\undefined%
      \relax%
    \else%
      \setlength{\unitlength}{\unitlength * \real{\svgscale}}%
    \fi%
  \else%
    \setlength{\unitlength}{\svgwidth}%
  \fi%
  \global\let\svgwidth\undefined%
  \global\let\svgscale\undefined%
  \makeatother%
  \begin{picture}(1,0.87695338)%
    \put(0,0){\includegraphics[width=\unitlength,page=1]{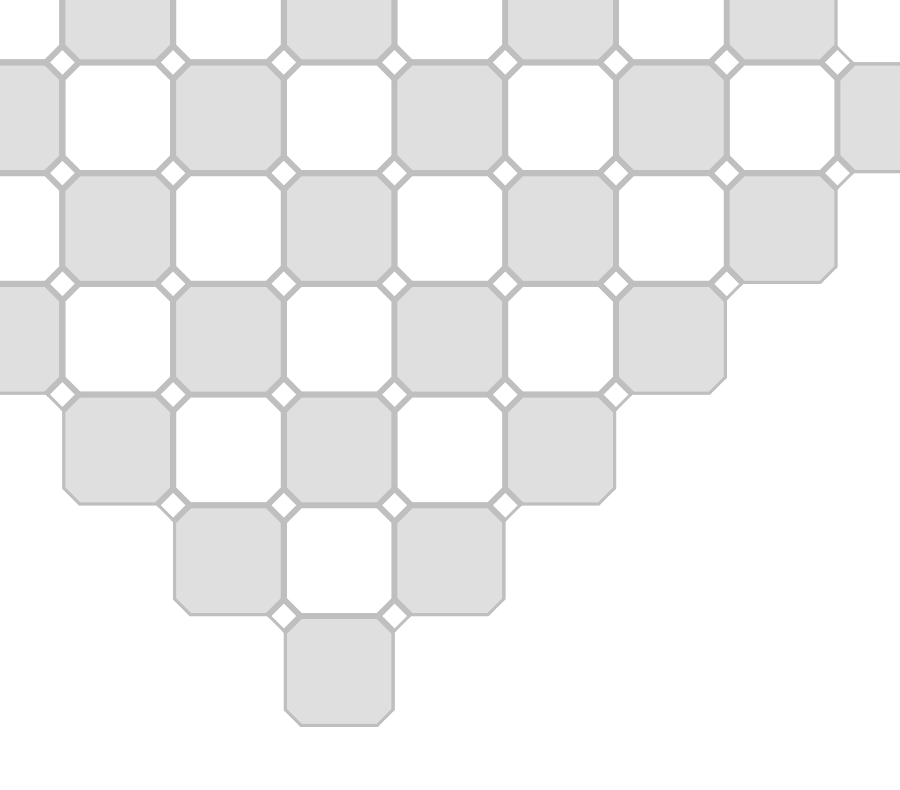}}%
    \put(0.575,0.39){\color[rgb]{0,0,0}\makebox(0,0)[lb]{\smash{$v$}}}%
    \put(0.57008427,0.49441716){\color[rgb]{0,0,0}\makebox(0,0)[lb]{\smash{$e_1$}}}%
    \put(0.485,0.37137052){\color[rgb]{0,0,0}\makebox(0,0)[lb]{\smash{$e_2$}}}%
    \put(0.61151213,0.69539457){\color[rgb]{0,0,0}\makebox(0,0)[lb]{\smash{$e$}}}%
    \put(0,0){\includegraphics[width=\unitlength,page=2]{boundary-involution-14-2a-1.pdf}}%
  \end{picture}%
\endgroup%

\end{minipage} & 
\begin{minipage}[b]{5cm} 
\centering 
\def\svgwidth{5cm}
\begingroup%
  \makeatletter%
  \providecommand\color[2][]{%
    \errmessage{(Inkscape) Color is used for the text in Inkscape, but the package 'color.sty' is not loaded}%
    \renewcommand\color[2][]{}%
  }%
  \providecommand\transparent[1]{%
    \errmessage{(Inkscape) Transparency is used (non-zero) for the text in Inkscape, but the package 'transparent.sty' is not loaded}%
    \renewcommand\transparent[1]{}%
  }%
  \providecommand\rotatebox[2]{#2}%
  \ifx\svgwidth\undefined%
    \setlength{\unitlength}{431.94676208bp}%
    \ifx\svgscale\undefined%
      \relax%
    \else%
      \setlength{\unitlength}{\unitlength * \real{\svgscale}}%
    \fi%
  \else%
    \setlength{\unitlength}{\svgwidth}%
  \fi%
  \global\let\svgwidth\undefined%
  \global\let\svgscale\undefined%
  \makeatother%
  \begin{picture}(1,0.87695338)%
    \put(0,0){\includegraphics[width=\unitlength,page=1]{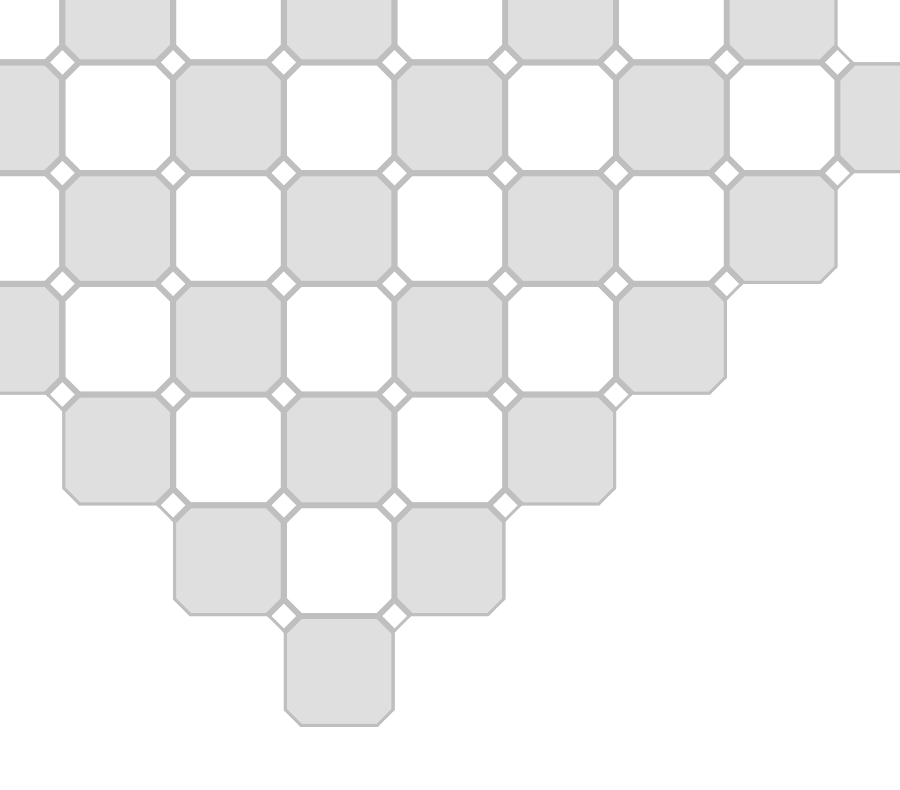}}%
    \put(0.575,0.39){\color[rgb]{0,0,0}\makebox(0,0)[lb]{\smash{$v$}}}%
    \put(0.57008427,0.49441716){\color[rgb]{0,0,0}\makebox(0,0)[lb]{\smash{$e_1$}}}%
    \put(0.485,0.37137052){\color[rgb]{0,0,0}\makebox(0,0)[lb]{\smash{$e_2$}}}%
    \put(0.61151213,0.69539457){\color[rgb]{0,0,0}\makebox(0,0)[lb]{\smash{$e$}}}%
    \put(0,0){\includegraphics[width=\unitlength,page=2]{boundary-involution-13-3b-1.pdf}}%
  \end{picture}%
\endgroup%

\end{minipage} \\
\hline
\begin{minipage}[c]{4.5cm} \smallskip \flushleft weight of configuration with $e_1 \in \hat{\gamma}$ \smallskip\end{minipage}
  & $X\sqrt{2}$ & $X$ \\
\hline
\begin{minipage}[c]{4.5cm} \smallskip \flushleft weight of configuration with $e_1 \notin \hat{\gamma}$ \smallskip\end{minipage}
  & $Y\sqrt{2}$ & $Y$ \\
\hline
\begin{minipage}[c]{4.5cm} \smallskip \flushleft winding at $e$ given $e \in \hat{\gamma}$ \smallskip\end{minipage}
  & $W$ & $W \; (\text{mod } 4\pi)$ 
\end{tabular}
\smallskip
\caption{The effect of a jump at a vertex $v$ to the weight of a loop configuration and the winding of $\hat{\gamma}$.
The involution 
where the state of the random cluster configuration is changed at $v$
and which is illustrated in this pair of figures,
preserves the winding factor in the observable and the relative weight of the configuration.
All the edges mentioned are in $E(G^\diamond_\delta)$, that is,
they are edges between two octagons. The edges $e_1$ and $e_2$ are edges starting at a boundary vertex $v$
(a vertex along the internal boundary).
The edge $e$ is any fixed edge on the graph, and it is the point at which we consider the value of
the observable.
}
\label{tbl: martingale property}
\end{table}

For a non-negative integer $t$,
let $G_t^\ddiamond$ be the 
slit graph
where we have removed $\gamma[0,t]$ and 
let $a_t$ be the edge $\gamma([t-1/2,t])$.

Next we decorate the boundary arc $bc \subset V_{\partial,1} (G^\diamond)$ with i.i.d.\ random variables
$(\xi_v)$, $\xi_v \in \{-1,1\}$, 
$v$ a vertex on $bc$,
such that
\begin{equation}
\P( \xi_v = -1 )=1 - \frac{1}{\sqrt{2}}, \qquad \P( \xi_v = 1 )= \frac{1}{\sqrt{2}}  .
\end{equation}
Let $(\F_t)_{t \in \R_+}$ be the filtration such that the $\sigma$-algebra
$\F_t$ is generated by $\gamma[0,t]$ and $\xi_v$ for any $v \in \gamma[0,t-1] \cap bc$.

Notice that we can couple $\zeta_v$ and $\xi_v$ in such a way that
\begin{equation}
\P( \xi_v = 1 \,|\, \zeta_v=0)=1, \qquad \P( \xi_v = 1 \,|\, \zeta_v=1)=\P( \xi_v = -1 \,|\, \zeta_v=1)=\frac{1}{2} .
\end{equation}
We interpret $\zeta_v$ so that $\zeta_v = 1$ if and only if we jump at $v$ 
from one boundary touching loop to the neighboring one,
and $\xi_v$ so that $\xi_v=1$ when there is no jump and $\xi_v$ is
a fair coin flip when there is a jump.

Set for any integer $t$
\begin{align}
M_t^+ &= \P^{G^\ddiamond_t,a_t,b_t,c,d}( \gamma_1 \subset \hat{\gamma}) \\
M_t   &= \left(\prod_{v \in \gamma[0,t-1] \cap bc} \xi_v \right) M_t^+
\end{align}
where we can either take $b_t=b$ or $b_t$ is 
the rightmost point visible from $cd$.
The rightmost visible point 
means that we move $b_t$ on all branching events to the place where we cut
the next loop open, and to the edge which points towards that vertex, to be more specific.
That is, $M_t^+$ is the probability that $a$ is connected to $d$ by the interface 
(internal arc) in the slit domain and by the proof of the next result, it 
is interpreted as a conditional probability in the original domain.
Lemmas~\ref{lem: M martingale} and \ref{lem: N martingale} below are very central for the
proof of the main theorem, namely, the law of the exploration process
is determined from the martingale property of these quantities.

\begin{lemma}\label{lem: M martingale}
$(M_t)$ is $(P,\F_t)$ martingale.
\end{lemma}

\begin{proof}
First notice that if $\gamma(t) \notin bc$, then $\E[ M_{t+1}^+ \,|\, \F_t] = M_t^+$ by the Markov property
of random cluster model and hence
\begin{equation}
\E[ M_{t+1} \,|\, \F_t] = \left(\prod_{v \in \gamma[0,t-1] \cap bc} \xi_v \right) \E[ M_{t+1}^+ \,|\, \F_t] = M_t .
\end{equation}
On the other hand, if $\gamma(t) \in bc$, then $M_{t+1}^+ = (1+\sqrt{2}) \, M_t^+$ 
by \eqref{eq: def zeta},
in other words, $M_{t+1}^+$ would hit $0$ if $\gamma$ continued to follow a loop which turns away from $d$,
and thus the other possible value of $M_{t+1}^+$ (when $\gamma$ turns towards $d$) has to be by
the Markov property of random cluster model and consequent martingale property, equal to
$(1+\sqrt{2}) \, M_t^+$ which is one over the probability of that event times the value
of $M_t^+$.
Hence
\begin{equation}
\E[ M_{t+1} \,|\, \F_t] = \left(\prod_{v \in \gamma[0,t-1] \cap bc} \xi_v \right) M_t^+ \,  
   \underbrace{(1+\sqrt{2}) \, \E[ \xi_{\gamma(t)} \,|\, \F_t]}_{=1} = M_t .
\end{equation}
Therefore $M_t$ is a martingale.
\end{proof}

Define 
for fixed edge $e$
\begin{equation}
N_t = \theta_\delta \, \E^{G^{\ddiamond}_t,a_t,b_t,c,d} (\ind_{e \in \hat{\gamma}} e^{-i \frac{1}{2} W(d,e)} )
\end{equation}
where $\theta_\delta$ is the constant with unit modulus chosen in Section~\ref{ssec: preholo}.
In words, the value of the process $(N_t)$ at time $t$ is the value at the given edge $e$ of the observable 
on the slit domain at time $t$. 
For the next property $\theta_\delta$ needs to be chosen consistently. This can be done for example by requiring that the sign of
the observable is always the same at $d$.

\begin{lemma}\label{lem: N martingale}
$(N_t)$ is $(P,\F_t)$ martingale.
\end{lemma}

\begin{proof}
If $\gamma(t) \notin bc$, then clearly $\E[ N_{t+1} | \F_t] = N_t$.

If $\gamma(t) \in bc$, then $N_{t+1} = N_t$,
since by the observations in Table~\ref{tbl: martingale property}
the random variables
$\ind_{e \in \hat{\gamma}}$ and $e^{-i \frac{1}{2} W(d,e)}$ are independent 
from the state of the edge at $\gamma(t)$ conditionally on the history up to time $t$.
Therefore $N_t$ is a martingale.
\end{proof}

\subsection{Convergence of the observables}

In this subsection we first work on the scaling limit of the \emph{4-point observable}. In that approach we keep the points $c_\delta$
and $d_\delta$ at a macroscopically positive distance. That could be considered as a motivation for the so called
\emph{3-point} or \emph{fused observable} where we take the same limit, but keeping $c_\delta$
and $d_\delta$ at a microscopically bounded distance.

\subsubsection{Convergence of the 4-point observable}\label{ssec: convergence observable 4-pt}

\begin{figure}[tbh]
\centering
\subfigure[The setup: $a$ and $c$ are inward pointing edges, $b$ and $d$ are outward pointing edges. The picture indicates also
the boundary conditions: $ab$ is free, $bc$ wired etc. 
We need the external connection from $b$ to $a$ to ensure
that the random cluster measure is consistent with the one introduced in Section~\ref{sec: setup}.] 
{
	\label{sfig: setup and observable a}
	\includegraphics[scale=.8]
{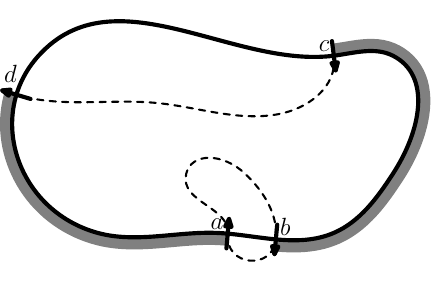}
} 
\hspace{0.2cm}
\subfigure[The discrete function $H_\delta$ is almost discrete harmonic: its restriction to black squares is 
discrete subharmonic
and to white squares discrete superharmonic. The boundary conditions it satisfies are given in the picture.
Notice that there is an unknown constant $0<\beta<1$ which has the interpretation that $\sqrt{\beta}$ 
is the probability of an internal connection
pattern.]
{
	\label{sfig: setup and observable b}
	\includegraphics[scale=.8]
{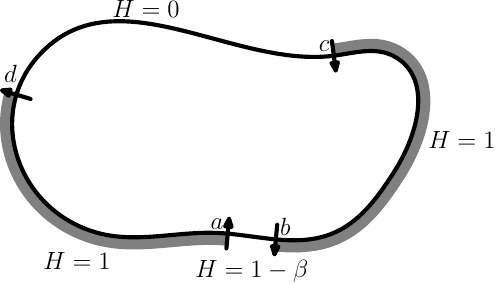}
}
\caption{The setup for the random cluster measure, for the curves and loops and the boundary conditions satisfied
by the (discrete) harmonic functions.}
\label{fig: setup and observable}
\end{figure}

It is straightforward to apply the reasoning of \cite{Smirnov:2010ie} to this case. We only summarize the method here without any proof.
See also the lecture notes \cite{DuminilCopin:2012wb}.

The observable $f_\delta$ is given on the medial lattice. It is preholomorphic on the vertices of the lattice and
it has well-defined projections to the complex lines $l(e)$ defined on the edges of the lattice.
Define a function $H_\delta$ on the square lattice by setting $H_\delta(B_0)=1$ where $B_0$ is the black square
next to $d$ and then extending to other squares by
\begin{equation}
H_\delta(B) - H_\delta(W) = |f_\delta(e)|^2
\end{equation}
where $B$ and $W$ are any pair of neighboring black and white square and $e$ is the edge of the medial lattice between
them. Now $H_\delta$ is well-defined since by the properties of $f_\delta$ the sum of differences of $H_\delta$ along
a closed loop is zero.
The boundary values $H_\delta$ are the following, see also Figure~\ref{fig: setup and observable}:
\begin{itemize}
\item $H_\delta$ is equal to $0$ on the arc $c d$.
\item $H_\delta$ is equal to $1$ on the arcs $d a$ and $b c$.
\item $H_\delta$ is equal to $1-\beta_\delta$ on the arc $a b$.
\end{itemize}

The relation between $f_\delta$ and $H_\delta$ becomes more clear after a small calculation. Namely,
by this calculation for neighboring black squares $B,B'$ with a vertex $v \in L_\diamond$
\begin{equation}
H_\delta(B') - H_\delta(B) = \frac{1}{2} \imag \left( f_\delta(v)^2 \,\frac{B'-B}{\delta} \right).
\end{equation}
Notice that the complex number $(B'-B)/\delta$ has modulus $\sqrt{2}$.
The natural interpretation is that $H_\delta(B)$ is the imaginary part of
the discrete integral (counted in lattice-step units) of $\frac{1}{2}f_\delta^2$ from $B_0$ to $B$ along
any connected path of black squares.
This means that
\begin{equation}
H_\delta(B') = H_\delta(B) +  \imag \int_B^{B'} \left( \frac{1}{\sqrt{2\delta}} f_\delta\right)^2
\end{equation}
where 
the integral sign denotes
the discrete integral is over any lattice path connecting $B$ to $B'$ 
and is defined as the sum over the edges of the path,
of the integrand evaluated at the mid point of the edge times the complex number which is the difference of the head and the tail
of the edge.

Denote the restriction of $H_\delta$ to black and white squares by $H_\delta^\bullet$ and $H_\delta^\circ$, respectively.
Again by the properties of $f_\delta$, $H_\delta^\bullet$ is subharmonic 
and $H_\delta^\circ$ is superharmonic, see 
\cite[Lemma~3.8]{Smirnov:2010ie}. 
Let $\tilde{H}_\delta^\bullet$ be the preharmonic function on the black squares with the same boundary values
as $H_\delta^\bullet$ and similarly $\tilde{H}_\delta^\circ$ be the preharmonic extension of the boundary values
of $H_\delta^\circ$. Also extend all these function to be continuous functions, say, by using the bilinear
extension which in takes the form
\begin{equation*}
h(x,y)= a_1 +a_2 x + a_3 y + a_4 x y 
\end{equation*}
in each square and matches with the values given at the corners of the square.
Then at each interior point
\begin{equation}\label{ie: harmonic major and minor}
\tilde{H}_\delta^\circ(z) \leq H_\delta^\circ(z) \leq H_\delta^\bullet(z) \leq \tilde{H}_\delta^\bullet(z) .
\end{equation}
Next apply standard difference estimates to show that the preharmonic functions $\tilde{H}_\delta$ and 
$\tilde{H}_\delta^\bullet$ have convergent
subsequences and also using crossing estimates show that their boundary values approach to each other.
Since $0 \leq \beta_\delta \leq 1$, by taking a subsequence 
we assume that
$\beta_\delta$ converges to a number $\beta$ and $H_\delta$ converges to a harmonic function on $\domain$
with the boundary values $H=0$ on $cd$, $H=1$ on $bc$ and $da$ and $H=1-\beta$ on $ab$.

As explained in 
\cite[Section~5]{Smirnov:2010ie} 
we can extend the convergence of $H_\delta$ to the convergence of $f_\delta$
and hence along the same subsequence as $H_\delta$ converges to $H$ also $\frac{1}{\sqrt{2 \delta}} f_\delta$ converges to
$f$ defined by
\begin{equation}
f(z) = \sqrt{\phi'(z)}
\end{equation}
where $\phi$ is any holomorphic function with $\imag \phi = H$.

In fact, the value of $\beta$ is determined uniquely and it depends only on the conformal type of the domain. 
We'll give the argument in Section~\ref{ssec: beta 4pt}
for completeness following the lines of Section~\ref{sssec: convergence 3pt}. Suppose for now that it is the case.
Then it follows that the whole sequence $\frac{1}{\sqrt{2 \delta}} f_\delta$ converges.

\subsubsection{A proof for the fused case}\label{sssec: convergence 3pt}

Consider now the same setup, but when $c$ and $d$ are close to each other, say, at most at bounded lattice distance from each other.
Denote by $z_0$ the point in the continuum that $c$ and $d$ are approximating. 
We will deal with the case of flat boundary near $z_0$.

We will change the definition of $H$ to that of $1-H$, it means that on the arc $ab$, $H \equiv \beta$, on the arcs $bc$ and $da$,
$H \equiv 0$ and on the arc $cd$, $H \equiv 1$. Then $H$ is superharmonic when restricted to $\sqlattice$ and subharmonic
when restricted to $\dsqlattice$ and the inequalities in \eqref{ie: harmonic major and minor} are reversed.
 
We expect that in the fused ($3$-point) limit, $\beta_\delta$ goes to zero, which we have to compensate
by renormalizing the observable. In effect, the value of $H$ on $cd$ will go to infinity and we expect to get a Poisson kernel type
singularity. Hence we say that there is a singularity at $z_0$ (or at $c$ or $d$).

We will make the following definitions:
\begin{itemize}
\item The discrete half-plane $\half_\delta^{(z_0, \theta)}$ is a discrete approximation of 
$\half^{(z_0, \theta)} \dd=z_0 + e^{i \theta} \half$
where $\half = \{ z \in \C \,:\, \imag[z]>0\}$. 
Suppose that the boundary lies between two parallel lines that are at
a distance which remains bounded (in lattice steps) as $\delta \to 0$. Assume also that the projection of a parametrization of 
the boundary on one of the lines is a monotone function, at least in sufficiently large neighborhood of $z_0$. 
This ensures that there are no long fjords near $z_0$.
\item 
$H^{+}_\delta$, $f^+_\delta$ are the half plane functions 
on $\half_\delta^{(z_0, \theta)}$ with the ``singularity'' at $z_0$. 
That is, $f^+_\delta$  is the unique, up to $\pm$ sign, 
bounded preholomorphic function on $\half_\delta^{(z_0, \theta)}$
whose boundary values are (i) $\pm 1$ times $l(c)$ (see Section~\ref{ssec: preholo})
on the edge at $c$, (ii) that value transported to $d$ by a lattice path connecting $c$ to $d$
in $\half_\delta^{(z_0, \theta)}$ 
on the edge at $d$
and (iii) otherwise satisfy the boundary condition that $f^+_\delta$ is parallel
to $\frac{1}{\sqrt{\tau(z)}}$ where $\tau(z)$ is the unit tangent at the boundary point $z$
(e.g. if the incoming edge is pointing to the direction of the unit complex number $e_1$
and the outgoing edge is pointing to the direction of the unit complex number $e_2$, then
the uni tangent is $(e_1 + e_2)/\sqrt{2}$). Such $f^+_\delta$ can be defined also
as the usual FK Ising observable in $\half_\delta^{(z_0, \theta)}$. The function
$H^{+}_\delta$ is defined similarly as above by summing $\pm |f^+_\delta(e)|^2$ along
lattice paths, and its boundary values are $1$ on $cd$ and $0$ elsewhere.
\item Suppose that $\Omega_\delta$ is a discrete domain and $z_0$ is its boundary point. 
Assume that the boundary in a $r$~neighborhood of $z_0$ is flat in the same uniform manner 
in $\delta$ as in the definition of the half plane $\half_\delta^{(z_0, \theta)}$.
\item $H_\delta$, $f_\delta$ are the functions on $\Omega_\delta$ with the ``singularity'' at $z_0$.
\end{itemize}

The next lemma gives the convergence of the observable in a half-plane. We will compare the other observables to this one.

\begin{lemma}\label{lm: obs convergence half-plane}
Let $w_\delta$ be the discrete approximation of the point $z_0 + i\,e^{i \theta}$. Let $L_\delta = H^+_\delta(w_\delta)$
and $\hat{H}^+_\delta = \frac{1}{L_\delta} H^+_\delta$. The following statements hold:
\begin{enumerate}\enustyii
\item \label{enui: obs convergence half-plane i}
As $\delta \to 0$ the sequence $\hat{H}^+_\delta$ converges uniformly on compact sets to $\imag ( -\frac{e^{i \theta}}{z-z_0})$.
\item \label{enui: obs convergence half-plane ii}
For each sequence $\delta_n \searrow 0$ as $n \to \infty$, there exists a subsequence $\delta_{n_k}$
and a constant $c^+ \in \R_{>0}$ such
that as $k \to \infty$ the sequence $L_{\delta_{n_k}}/\delta_{n_k}$ converges to $c^+$
and the sequence $\delta_{n_k}^{-1} H^+_{\delta_{n_k}}$ converges
uniformly on compact sets to $c^+ \imag ( -\frac{e^{i \theta}}{z-z_0})$.
\end{enumerate}
\end{lemma}


\begin{proof}
Let's first prove \ref{enui: obs convergence half-plane i} and then use that result to prove
\ref{enui: obs convergence half-plane ii}.

We need here that $\hat{H}^+_\delta(w_\delta)=1$ and $\hat{H}^+_\delta \geq 0$.
Harnack's inequality and Harnack boundary principle \cite{ChelkakSmirnov:2012} imply that for any $r>0$,
$\hat{H}^+_\delta \leq C$ in $\half_\delta^{(z_0, \theta)} \setminus B(z_0,r)$.

Fix a compact subset of $\half^{(z_0,\theta)}$ with non-empty interior. 
From boundedness of $\hat{H}^+_\delta$ in the whole domain it follows
that the corresponding $f$ function $\frac{1}{\sqrt{\delta}} \hat{f}^+_\delta$ remains bounded
on that compact set, see \cite[Section~5.1.3]{DuminilCopin:2012wb}.
The boundary the values of $\hat{H}^+_\delta$
on $\sqlattice$ and $\dsqlattice$ are close 
and hence the harmonic extensions
of the two functions to the interior are close. Hence (by a standard argument)
along a subsequence $\hat{H}^+_\delta$ converges
to a function on that compact set and the limit is harmonic in the interior points.
It follows that $\frac{1}{\sqrt{\delta}} \hat{f}^+_\delta$ converges uniformly on compact
subset along that subsequence similarly as in Section~\ref{ssec: convergence observable 4-pt}.

By taking an increasing sequence of compact sets we see that the convergence takes place in 
the whole half-plane for a subsequence. 
The limit has to be the Poisson kernel of the half-plane 
normalized to have value $1$
at the point $z_0 + i\,e^{i \theta}$,
because the limit is harmonic and
$\hat{H}^+_\delta$ is non-negative, has zero boundary values
away from $z_0$ and satisfies the normalization $w_\delta$.
This claim can be proven using integration with respect 
to the Poisson kernel of the upper half-plane
(shifted by small $\epsilon>0$ towards the interior of the domain).
Since the limit is the same for all subsequences, the whole sequence converges.

For \ref{enui: obs convergence half-plane ii} we use
the first claim, \ref{enui: obs convergence half-plane i}, and the fact that 
$H^+_\delta = L_\delta \, \hat{H}^+_\delta$.
Notice that the harmonic upper and lower bound for the restrictions of $H^+_\delta $
to $\dsqlattice$ and $\sqlattice$ can be bounded from above and from below, respectively,
by a quantity of the form $const. \delta $ just 
by considering the probabilities of simple random walks to exit the domain through $cd$
on these two lattices. 
Notice that we use here the fact that the boundary is flat around $z_0$
(though such bounds hold also true, when the boundary is smooth and
the approximating discrete boundary is well chosen, see also 
Remark~\ref{rem: obs smooth boundary}).
The best constants that we get might have a gap in between.
Nevertheless, there exist constants $C_1$ and $C_2$ such that 
$0 < C_1 < \delta^{-1} L_\delta < C_2 < \infty$ for small enough $\delta$.

Thus we can take a subsequence $\delta_n$ so that $\delta_n^{-1} L_{\delta_n}$ converges as $n \to \infty$.
Then the claim holds for $c^+ = \lim_{n \to \infty} \delta_n^{-1} L_{\delta_n}$.
\end{proof}


\begin{proposition}\label{prop: convergence of fused observable}
For any sequence $\delta_n \searrow 0$ as $n \to \infty$, 
any sequence $\half_{\delta}^{(z_0, \theta)}$ (where $z_0$, $\theta$ are fixed, for simplicity)
and any $r>0$, 
there exists a subsequence $\delta_{n_k}$
and constants $c \in \R_{>0}$ and $\beta \in \R_{\geq 0}$ such
that  
for any sequence of domains $\Omega_{\delta_n}$ that agrees with 
$\half_{\delta_n}^{(z_0, \theta)}$ in the $r$-neighborhood of
$z_0$ and that converges to a domain $\Omega$ in the Carath\'eodory sense, 
the sequence $\delta_{n_k}^{-1} H_{\delta_{n_k}}$ converges as $k \to \infty$
uniformly on compact sets to a harmonic function $h$ that satisfies the following boundary conditions
\begin{equation}\label{eq: h at boundary}
h=\begin{cases}
c \beta & \text{on } (a\,b) \\
\text{bounded}
& 
\text{near $a$ and $b$}
\\
0 & \text{on } (b\, z_0) \text{ and } (z_0\, a) \\
c \imag ( -\frac{e^{i \theta}}{z-z_0}) + \OO(1) & \text{as } z \to z_0
\end{cases}
\end{equation}
Furthermore, $\delta_{n_i}^{-1} H^+_{\delta_{n_i}}$ converges
uniformly on compact sets to $c^+ \imag ( -\frac{e^{i \theta}}{z-z_0})$ along the same subsequence
and $c=c^+$ and the convergence is uniform over all domains $\Omega$.
\end{proposition}

\begin{proof}
Use the same argument as in the proof of Lemma~\ref{lm: obs convergence half-plane}~\ref{enui: obs convergence half-plane ii}
to show that $\delta_{n_i}^{-1} H_{\delta_{n_i}}$ converges 
uniformly on compact sets to a harmonic function $h$ and it satisfies \eqref{eq: h at boundary}.
Suppose that $c \neq c_+$. Then, 
because $f$ can be recovered from $h$
by the formula $f = \sqrt{2 i \, \partial_z h}$, it is straightforward to show that
\begin{equation}
\tilde{h} = \imag \int \big(f-f^+\big)^2 
= \big(\sqrt{c} - \sqrt{c^+}\big)^2 \frac{\imag ((z-z_0) \, e^{-i\theta})}{|z-z_0|^2} 
  + \OO(\imag ((z-z_0) \, e^{-i\theta}))
\end{equation}
as 
$z \to z_0$.
There exists $r>0$ such that $\tilde{h}$ is positive in
$\half^{(z_0, \theta)} \cap B(z_0,r)$.

Next we notice that $\tilde{h}$ is not bounded in $\half^{(z_0, \theta)} \cap B(z_0,r)$.
On the other hand, if we consider the discrete version $\tilde{h}_\delta$ of $\tilde{h}$, it 
must remain  bounded on $\half_\delta^{(z_0, \theta)} \cap \partial B(z_0,r)$ 
uniformly in $\delta$
because of the convergence to $\tilde{h}$.
Thus it is bounded inside $\half_\delta^{(z_0, \theta)} \cap B(z_0,r)$ because of its $0$ boundary values on the straight part of the
boundary around $z_0$.
This leads to a contradiction.
\end{proof}

\begin{remark}
If we normalize by the value of $H^+$, the convergence holds for the whole sequence, not just along subsequences.
\end{remark}

\begin{remark}\label{rem: obs smooth boundary}
This proof can be generalized to any domain with $z_0$ lying on a smooth boundary segment, if the
domain is approximated in a nice way around $z_0$. This means that the boundary near $z_0$ 
lies between two copies of the same smooth curve shifted by a bounded number of lattice steps, for instance. 
The upper and lower bound for the harmonic function 
with a pole at a boundary vertex,
evaluated at fixed interior point 
take the form $const. \delta $ since we can bound it using hitting probabilities 
in regular regions such as squares or discs (or rather their discrete approximations).
The constant depends on the chosen interior point as well as the continuum domain,
but not on its discrete approximation. 
Namely, 
the local geometry of the boundary near the pole doesn't play a role as we can take
minimums and maximums over the finite number of possibilities (when the boundary is between two curves which are closer than a finite multiple of $\delta$).
All $H$ functions on domains that have a given $r$ neighborhood of $z_0$
must have the same singular part. Moreover, the value at a fixed point of any fixed $H$ function 
can be used for normalization for the other functions and they converge to a limit.

The generalization to non-smooth boundary would require a matching pair of upper and lower bounds for 
$H_{\delta}$, the former for the restriction to $\dsqlattice$ and the latter to $\sqlattice$. 
This is therefore equivalent of knowing (uniform bounds for) the leading term of the
asymptotics of the exit probabilities with different lattice ``boundary shapes'' as $\delta \to 0$.
The asymptotics is heavily influenced by the local geometry around $z_0$.
\end{remark}

In the final result of this subsection, we derive a characterizing property for the constant $\beta$.
Hence this constant is uniquely determined by the continuum setting and doesn't depend on the
discrete approximations we are using. In fact, it only depends on the conformal type of the domain 
$(\Omega,a,b,z_0)$ with a prescribed length scale (derivative) at $z_0$.

Suppose for simplicity that $\Omega$ is a Jordan domain. 
Let $h: \Omega \to \R$ be a harmonic function, for instance,
$h$ is the scaling limit from Proposition~\ref{prop: convergence of fused observable}.
Suppose that $z$ is a boundary point and $h= \beta$ near $z$ on the boundary.
We define a weak version of the sign of the normal derivative $\partial_n h(z)$ to the direction of the outer normal 
at a boundary point $z$ of a (possibly non-smooth) domain by
\begin{itemize}
\item $\partial_n h \geq 0$ at $w \in ab$, if
\begin{equation*}
h^{-1}( \,(-\infty,\beta]\, ) \cap B(w,\eps) \neq \emptyset
\end{equation*}
for all $\eps>0$.
\item $\partial_n h \leq 0$ at $w \in bc \cup ca$, if
\begin{equation*}
h^{-1}( \,[\beta,\infty)\, ) \cap B(w,\eps) \neq \emptyset
\end{equation*}
for all $\eps>0$.
\end{itemize}

\begin{proposition}\label{prop: value of beta}
The constant $\beta$ in \eqref{eq: h at boundary} is the unique constant such that $h$ has the normal derivatives
\begin{equation}\label{eq: h normal derivatives}
\begin{cases}
\partial_n h \geq 0 & \text{on $ab$} \\
\partial_n h \leq 0 & \text{on $bc$ and $ca$} 
\end{cases}
\end{equation}
and there exists a point $w$ on $ab$ such that $h^{-1}( \,(\beta,\infty)\, ) \cap B(w,\eps) \neq \emptyset$ for all $\eps>0$.
\end{proposition}

\begin{proof}
The normal derivatives \eqref{eq: h normal derivatives} follow by the same argument as in \cite{ChelkakSmirnov:2012},
Section~6.

Suppose that $\phi:U \to \half$ is conformal and onto and $\phi(c)=\infty$. Then $h = h^\half \circ \phi$ where
\begin{equation}
h^\half (z) = \imag \int (f^\half)^2
\end{equation}
and up to a universal multiplicative constant
\begin{equation}\label{eq: fused f in half}
f^\half(z)= \sqrt{1 + \beta \left( -\frac{1}{z-u} + \frac{1}{z-v}\right)},
\end{equation}
where $u=\phi(a)$ and 
$v=\phi(b)$.

Write
\begin{equation}\label{eq: f in terms of Q}
f^\half(z)= \sqrt{\frac{Q(z)}{(z-u)(z-v)}},
\end{equation}
where
\begin{equation}\label{eq: Q in terms of u v}
Q(z)=z^2 - (u+v) z + (uv + \beta(v-u)) .
\end{equation}
Since the coefficients of the quadratic polynomial $Q$ are real, it either has two real zeros or a pair of complex conjugate zeros
(with non-zero imaginary part). Since \eqref{eq: f in terms of Q} holds and 
$f^\half$ is 
single-valued in $\half$ (being a scaling limit of a discrete observable which is single-valued),
there can't be any zeros of multiplicity one in the upper half-plane. Thus the zeros of $Q$ are real.

Now 
\begin{equation}\label{eq: normal de h in terms of Q}
\partial_n h^\half (x) = - \partial_y h^\half (x) = -\real  [(f^\half(x))^2] = \frac{Q(x)}{(x-u)(v-x)}
\end{equation}
on $\R \setminus \{u,v\}$. Let the zeros of $Q$ be $w_1$ and $w_2$ with $w_1 \leq w_2$.
Then $w_1 + w_2 = u+v$ by \eqref{eq: Q in terms of u v}. Therefore we have to analyze the following four cases
\begin{enumerate}
\item $u < w_1=w_2 < v$
\item $u < w_1 < w_2 < v$
\item $u=w_1 < w_2=v$
\item $w_1 < u < v < w_2$.
\end{enumerate}
We notice using \eqref{eq: normal de h in terms of Q} 
that only the first case is consistent with \eqref{eq: h normal derivatives}.

Thus we have shown that $w_1=w_2 = (u+v)/2$ and hence $\beta = (v-u)/4$. 
The normal derivative $\partial_n h^\half (x)$ is positive when $x \in (u,v)$ and negative when $x \in \R \setminus [u,v]$.
The value of $\beta$ is uniquely determined and
it is the only real value such that $\partial_n h $ has the properties claimed.
\end{proof}

\subsection{Value of \texorpdfstring{$\beta$}{beta} for the 4-point observable}\label{ssec: beta 4pt}

We will determine the value of $\beta$ for the 4-point observable in the same way as in Proposition~\ref{prop: value of beta}.
Remember that $0 \leq \beta \leq 1$.

Consider $f^\half$ for the domain $\half$ with the four marked boundary points $u<v<w$ and $\infty$.
This means that $h$ can be written as $h = h^\half \circ \phi$ where
\begin{equation}
h^\half = \imag \int (f^\half)^2
\end{equation}
and $\phi(a)=u$, $\phi(b)=v$, $\phi(c)=w$ and $\phi(d)=\infty$.

For any $u<v<w$ 
\begin{equation}\label{eq: h half 4pt}
h^{\half,u,v,w} (z) = \frac{1}{\pi} \imag \left( -\log (z-w) + \beta \left( -\log(z-u) + \log(z-v) \right) \right) .
\end{equation}
Hence
\begin{align}
\sqrt{\pi} \, f^{\half,u,v,w} (z) &= \sqrt{ 
  -\frac{1}{z-w} + \beta \left( -\frac{1}{z-u} + \frac{1}{z-v} \right) } \nonumber \\ 
   &= \sqrt{ -\frac{ Q(z)}{(z-u)(z-v)(z-w)} } \label{eq: f in terms of Q 4pt}
\end{align}
where $Q(z)$ is a quadratic polynomial.

Let's simplify things by setting $u=0$ and $w=1$. Hence for $0<v<1$, $Q(z)$ can be written as
\begin{align}
Q(z) &= z\,(z-v) - \beta \,v\, (z-1) \nonumber \\
     &= z^2 - (\beta+1)\,v\,z + \beta\,v .
\end{align}
Since the coefficients of $Q$ are real, there are two options: \emph{either} we have that $\imag w_1 \neq 0$ or $\imag w_2 \neq 0$,
and then $w_1^* = w_2$, \emph{or} we have that $w_1$ and $w_2$ are real. Since \eqref{eq: f in terms of Q 4pt} holds
and $f^\half$ is single valued in $\half$, there can't be any zeros in the upper half-plane. Hence the zeros are real.

Let's write also in this case the normal derivative in the direction of the outer normal
\begin{equation}
\partial_n h^\half(x) = -\partial_y \imag  \int (f^\half)^2 = \real \frac{ Q(x)}{x(x-v)(x-1)}
\end{equation}
for all $x \in \R \setminus \{0,v,1\}$. By the same argument as in \cite{ChelkakSmirnov:2012}, Section~6,
this normal derivative is negative on $(-\infty,0) \cup (v,1)$ and positive on $(0,v)\cup(1,\infty)$.
This is only possible if the two roots of $Q$ are equal.
Therefore in addition to $0 < \beta < 1$, the constant
$\beta$ has to satisfy
\begin{equation}
(\beta+1)^2 v^2 = 4 \beta v
\end{equation}
and hence $\beta=\beta_-$ or $\beta=\beta_+$ where 
\begin{equation*}
\beta_\pm = \frac{ -v + 2 \pm 2 \sqrt{1-v}}{v}.
\end{equation*}
Let's write $\beta_\pm -1 = \frac{2}{v} \,\sqrt{1-v}\,(\sqrt{1-v} \pm 1)$, which
is positive for $\beta_+$ and negative for $\beta_-$ for all $v \in (0,1)$.
Thus $\beta = \beta_-$
and we find that
\begin{equation} \label{eq: b solution for bdry observable}
\beta= \frac{ -v + 2 - 2 \sqrt{1-v}}{v} = \left( \frac{1-\sqrt{1-v}}{\sqrt{v}} \right)^2 = \left( \tan\left( \frac{x}{2} \right) \right)^2.
\end{equation}
where $0<x<\pi/2$ is such that $v=\sin^2 x$. Notice that the double root of $Q$ is $\sqrt{\beta v}$ and it lies in the interval $(0,v)$.

To conclude, we state that the value of $\beta$ is characterized uniquely by the formula~\eqref{eq: h half 4pt}, the fact that
$\beta \in (0,1)$ and that the normal derivative of $h^\half$ is 
negative on $(-\infty,u) \cup (v,w)$ and positive on $(u,v)\cup(w,\infty)$.

\subsection{A remark on crossing probabilities}

As a side remark, let's derive the
probability $P$ of the internal arc pattern
\begin{equation} 
\intarcp{0}{\infty}{1}{v}
\end{equation}
under the random cluster measure where there isn't any exterior connections, that is, the arcs $\gamma_1$ and $\gamma_2$
(as defined in Section~\ref{ssec: setup}) are not counted as loops in the weight of a configuration.
Since $\sqrt{\beta}$ and $1-\sqrt{\beta}$ are proportional to $P$ and $\sqrt{2}(1-P)$, respectively, $P$ satisfies
\begin{equation}
\frac{P}{P+\sqrt{2}(1-P)} = \sqrt{\beta}
\end{equation}
and hence
\begin{align}
P &= \frac{\sqrt{2} \sqrt{\beta}}{(\sqrt{2}-1) \sqrt{\beta} + 1 } 
  = \frac{ \sin\left( \frac{x}{2} \right) }{\sin\left( \frac{x}{2} \right) + \sin\left( \frac{\pi}{4} - \frac{x}{2} \right)} .
\end{align}
By using the relations
\begin{equation}
\sin\left( \frac{x}{2} \right) = \frac{1}{\sqrt{2}} \sqrt{1 - \sqrt{1-v}}, \qquad
  \sin\left( \frac{\pi}{4} - \frac{x}{2} \right) = \frac{1}{\sqrt{2}} \sqrt{1 - \sqrt{v}}
\end{equation}
we can write this into the form
\begin{equation}
P = \frac{ \sqrt{1 - \sqrt{1-v}}}{ \sqrt{1 - \sqrt{1-v}} +  \sqrt{1 - \sqrt{v}}} .
\end{equation}
This is consistent with the result in \cite{ChelkakSmirnov:2012}.


\section{Characterization of the scaling limit}\label{sec: characterization}

\subsection{Martingales and uniform convergence with respect to the domain}

Consider the scaling limit of a single branch from $a$ to $z_0$ in the domain $\Omega$.
To simplify the setting, 
map the discrete random curves to a reference domain $\disc$ using conformal maps
so that the resulting probability law $\P_\delta$ 
is the law of a random curve in $\disc$ from $-1$ to $1$ 
and $\P^*$ is any subsequent scaling limit of $\P_\delta$.
Consider some scheme of parametrizing the curves that works for all $\delta$. 
We will mostly use the half-plane capacity parametrization.
Let $X$ be the space of continuous function from the time interval used for parametrization
to $\overline{\disc}$. We consider $X$ as a metric space with the metric defined by
the sup norm.
Denote by $\F_t$ the filtration generated by the curve up to time $t$.

After we start exploring the branch we will move automatically from the setting of two points to a setting of three points.
Hence we will also consider the setups of $(\Omega,a,b,z_0)$, where $z_0$ is the fused arc $(cd)$, and $(\Omega,a,b,c,d)$.

Remember that the two martingales were
\begin{align}
M_t^{(\delta)} &= \pm\sqrt{\frac{1}{\delta} \beta_t^{(\delta)}} = \pm  \sqrt{\frac{1}{\delta} \beta^{(\delta)}(\Omega_t,a_t,b_t,z_0)} \\
N_t^{(\delta)} &= \frac{1}{\delta} f_t^{(\delta)}(w_0) =  \frac{1}{\delta} f_t^{(\delta),\Omega_t,a_t,b_t,z_0}(w_0) 
\end{align}
Notice that we have included the scaling by a power of $\delta$ that makes these quantities converge in the limit $\delta \to 0$,
at least for subsequences.

Consider one of the processes above,  for instance, $(M_t^{(\delta)})_{t \geq 0}$ the martingale property could be formulated
so that if $0 \leq s < t$ and if $\psi: X \to \R$ is bounded, uniformly continuous and 
$\F_s$-measurable,
then 
\begin{equation}
\E^{(\delta), \Omega, a,b,z_0} [ \psi \,M_t^{(\delta)} ] 
  = \E^{(\delta), \Omega, a,b,z_0} [ \psi \,M_s^{(\delta)} ]
\end{equation}
Now due to the uniform convergence of $\beta^{(\delta)}(\Omega_t,a_t,b_t,z_0)$ over the domains,
Proposition~\ref{prop: convergence of fused observable}, the expected values on both sides will converge and we get
\begin{equation}
\E^{*, \Omega, a,b,z_0} [ \psi \,\tilde M_t ] = \E^{*, \Omega, a,b,z_0} [ \psi \,\tilde M_s ]
\end{equation}
where $\tilde{M}_t = \lim_{n \to \infty} M_t^{(\delta_n)} = \pm \sqrt{\beta(\Omega_t,a_t,b_t,z_0)}$. Thus $(\tilde{M}_t)_{t \geq 0}$ 
is a martingale.

By a similar argument, $\tilde{N}_t = \lim_{n \to \infty} N_t^{(\delta_n)}= f_t^{\Omega_t,a_t,b_t,z_0}(w_0)$ defines a martingale.

\subsection{Simple martingales and a martingale problem}

We wrote $f$ in the upper half-plane already in \eqref{eq: fused f in half}. Let us now analyze what happens
for a growing curve which we interpret as a random Loewner chain. For that we use Theorem~\ref{thm: ks}.
Next we notice that for all domains (and their approximating sequences) that agree near $z_0$, we had
a singularity in $H$ with the same constant in front, see Proposition~\ref{prop: convergence of fused observable}.
Fix some domain $(\Omega,a,b,z_0)$ and map it to the upper half-plane conformally.
Suppose that $w$ is the image of $z_0$. Then the singularity is of the form $c \imag( -1/(z-w) )$.
If we have a slit domain $(\Omega \setminus \gamma[0,t],\gamma(t),b,z_0)$ and we apply further the Loewner
map $g_t$ in the upper half-plane, then the singularity has to be
$c \imag( -g_t'(w)/(g_t(z)-g_t(w)) ) = c \imag( -1/((z-w) ) + \oo(1)$, as $z \to w$. This shows that
the functions $H$ transform as 
\begin{equation}
H^{\half \setminus K_t ,U_t,V_t,w} (z) = 
g_t'(w)
H^{\half,g_t(U_t),g_t(V_t),g_t(w)} (g_t(z)) .
\end{equation}
Since $f= \sqrt{ 2i \Phi'}$ where $\Phi$ is holomorphic and $\imag \Phi = H$,
\begin{equation}
f^{\half \setminus K_t ,U_t,V_t,w} (z) = \sqrt{g_t'(w) \, g_t'(z)} f^{\half,g_t(U_t),g_t(V_t),g_t(w)} (g_t(z)) .
\end{equation}

Now if we choose to send $w$ to $\infty$, then the observable is of the form~\eqref{eq: fused f in half}.
For Loewner chains $g_t'(\infty)=1$ when appropriately interpreted, and hence
\begin{equation}\label{eq: f in half t dependent}
f^{\half \setminus K_t ,U_t,V_t,\infty} (z) =
      \sqrt{g_t'(z)} \, \sqrt{ 
        1 + \beta_t \left( -\frac{1}{g_t(z)-U_t} + \frac{1}{g_t(z)-V_t} \right) }
\end{equation}
As we saw in the proof of Proposition~\ref{prop: value of beta}, the value of $\beta_t$ is
\begin{equation*}
\beta_t = \frac{1}{4} (V_t - U_t) .
\end{equation*}
We define 
\begin{align}
M_t & =  \pm \sqrt{4\beta_t} = \pm \sqrt{V_t - U_t}  \label{eq: def of mart m}   \\
N_t & = 4\big(\beta_t (V_t-U_t) - 2t\big) = M_t^4 - 8t .   \label{eq: def of mart n} 
\end{align}
The former quantity is proportional to $\tilde{M}_t$ and the latter one to the first non-trivial coefficient in 
the expansion of \eqref{eq: f in half t dependent} around $z=\infty$. 
Here $\pm$ signs are  
constant on each excursion of $V_t - U_t$
and distributed as independent fair (symmetric) coin flips
for each excursion. Here we interpret that 
an excursion starts at $0$, ends at $0$ and
is positive in between.

By the martingale properties in Section~\ref{ssec: setup} and the convergence results of the observables
we have the following result.

\begin{proposition}
Let $\P^*$ be a subsequent limit of the sequence of laws of FK Ising branch in discrete approximations of $(\Omega,a,b,z_0)$.
Let $\phi: \Omega \to \half$ be a conformal, onto map such that $\phi(z_0)=\infty$. Let $\gamma$ be the random curve
distributed according to $\P^*$ in the capacity parametrization, 
$U_t = \phi(\gamma(t))$ and $V_t$ is the ``right-most point'' in the hull of $\phi(\gamma(t))$.
Let the signs in \eqref{eq: def of mart m} be i.i.d. fair coin flips independent of $\gamma$.
Then processes $(M_t)_{t \geq 0}$ and $(N_t)_{t \geq 0}$ are martingales.
\end{proposition}

In the rest of this section we consider the following \emph{martingale problem}:
\begin{quote}
Let $(U_t,V_t)_{t \geq 0}$, $(M_t)_{t \geq 0}$ and $(N_t)_{t \geq 0}$ as above,
that is, satisfying that $(M_t)_{t \geq 0}$ and $(N_t)_{t \geq 0}$ are martingales
and satisfy relation~\eqref{eq: def of mart n}.
What is their law given that  $(M_t)_{t \geq 0}$ and $(N_t)_{t \geq 0}$ are martingales?
\end{quote}
We claim that the required properties (with the above functional dependency of 
the processes) uniquely determine the joint law of the processes. 
We call the verification of this claim and 
the explicit formulation of the law as the \emph{solution of the martingale problem}.

The solution is divided into two part.
In Section~\ref{ssec: char of V-U} we will show that $(|M_t|^{\frac{1}{\alpha}})_{t \geq 0}$ for some $\alpha>0$
is a Bessel process. 
In Section~\ref{ssec: char of pair U V}, we will show that $(V_t)_{t \geq 0}$ follows 
an evolution such that $V_t$ is a sum of a term from Loewner equation and a term
whose value is changing only in the random Cantor set $\{ t \,:\, U_t=V_t\}$, and 
then we show that the latter ``singular'' term is in fact identically $0$.

\subsection{Characterization of \texorpdfstring{$V_t - U_t$}{Vt-Ut}}\label{ssec: char of V-U}

In this section, we show how the ``martingale problem'' characterizes the law of
$(V_t- U_t)$. 


More concretely, we work towards the following theorem. Its proof is given in 
Section~\ref{sssec: x is bessel proof}.

\begin{theorem}\label{thm: x is bessel}
Let $X_t = V_t - U_t$ where $U_t$ and $V_t$ are the processes followed
by the marked points for the subsequent scaling limit of the FK Ising exploration process. 
Then $(X_t)_{t \geq 0}$ is a Bessel process of dimension $\delta = 3/2$
scaled by a constant $\sqrt{16/3}$.
\end{theorem}

\begin{remark}
In other words, $(X_t)_{t \geq 0}$ satisfies
\begin{equation}\label{eq: bessel sde}
\de X_t = \frac{\kappa(\delta-1)}{2 X_t } \de t + \sqrt{\kappa} \de B_t
\end{equation}
where $\kappa = \sqrt{16/3}$ and $\delta = 3/2$.
\end{remark}


\subsubsection{Relation to L\'evy's and Stroock--Varadhan martingale characterizations}

The argument which we will present can be compared to
Paul L\'evy's characterization of Brownian motion. The law of Brownian motion $(B_t)_{t \geq 0}$
is characterized by the fact that $B_t$ and $B_t^2 -t$ are martingales.
In the setting of general diffusions, 
the classical Stroock-Varadhan martingale problem approach describes weak solutions $X_t$ 
to stochastic differential equations of the form $\de X_t = \sqrt{a} \,\de B_t + b \,\de t$
(with coefficients $a$ and $b$ satisfying suitable measurability conditions) as
exactly those that all the quantities of the form 
$f(X_t) - \int_0^t ( \frac{1}{2} a_s (X) f''(X_s) + b_s (X) f'(X_s) ) \de s$ 
are martingales for a class of test functions
$f$, see Sections~V.19 and V.20 in \cite{Rogers:2000jy}. 
However, similarly to 
L\'evy's theorem,
there are stronger results, stating that two (well-chosen) 
martingales are enough to characterize the law of a diffusion, see \cite{Arbib:1965bx} and \cite{Voit:1998ue}. 
We take this path, using two martingales to show that the diffusion in question is the Bessel process.

\subsubsection{Lemmas}\label{ssec: mart prob lemmas}

We need the next two lemmas, which we write in greater generality suitable for the $4$-point case.

\begin{lemma}
Suppose that $T>0$ is a stopping time and suppose that $\psi \in C^2$ satisfies 
$\psi(0)=0$ and $\psi''(0)=0$.
If $A_t$ and $C_t$ are continuous, predictable processes which satisfy
\begin{enumerate}
\item $A_t$ is of bounded total variation
\item $C_t$ is non-decreasing, differentiable and satisfies $\dot{C}_t>0$ almost surely on $[0,T)$.
\end{enumerate}
then any continuous martingale $(M_t)_{t \in [0,T]}$ 
with the property that the process $(N_t)_{t \in [0,T]}$ defined by
\begin{equation}
N_t = A_t \psi(M_t) - C_t
\end{equation}
is a martingale, satisfies 
\begin{equation}
\P\left[ \int_0^T \ind_{M_t=0} \, \de t =0 \right]  = 1.
\end{equation}
\end{lemma}

\begin{proof}
Since $(M_t)_{t \in [0,T]}$ is a continuous martingale, it has quadratic variation process $\langle M \rangle_t$.
By It\^o's formula
\begin{equation}
\de N_t = \psi(M_t) \,\de A_t - \dot{C}_t \,\de t + \frac{1}{2} A_t \psi''(M_t) \,\de \langle M \rangle_t
 + A_t \psi'(M_t) \,\de M_t .
\end{equation}
Since $\psi(M_t)$ and $\psi''(M_t)$ vanish when $M_t=0$,
it follows that
\begin{equation}\label{eq: mart char lemma visit to zero}
\int_0^t \ind_{M_s=0} \,\de N_s - \int_0^t \ind_{M_s=0} A_s \psi'(M_s) \,\de M_s 
  = - \int_0^t \ind_{M_s=0} \dot{C}_s \,\de s .
\end{equation}
The process $\ind_{M_t=0}$ is predictable, 
because it is a pointwise limit of adapted continuous processes
(for instance, $\ind_{M_t=0}=\lim_{n \to \infty} \max\{0, 1 - n |M_t|\}$), 
and hence 
the left-hand side of \eqref{eq: mart char lemma visit to zero} is a local martingale.
On the other hand the right-hand side of \eqref{eq: mart char lemma visit to zero} is bounded variation
process. Hence $\int_0^t \ind_{M_s=0} \dot{C}_s \,\de s =0$ almost surely.
Since $\dot{C}_s>0$,  the claim follows.
\end{proof}

\begin{lemma}
If $(M_t)_{t \in [0,T]}$ is a continuous martingale with respect to $(\F_t)_{t \in [0,T]}$ such that
\begin{enumerate}
\item almost surely $\int_0^T \ind_{M_t=0} \de t  = 0$
\item $\de \langle M \rangle_t = \sigma_t^2 \de t$
and $\sigma_t >0$ for any $t$ such that $M_t \neq 0$,
\end{enumerate}
then $(B_t)_{t \in [0,T]}$ defined as
\begin{equation}
B_t = \int_0^t \ind_{ M_s \neq 0 } \, \sigma_s^{-1} \, \de M_s
\end{equation}
is well-defined and continuous and it is 
a one-dimensional standard $(\F_t)$-Brownian motion.
\end{lemma}

\begin{proof}
Let $F= [0,T] \setminus \{ t \,:\, M_t=0\}$ which is relatively open in $[0,T]$.
For any $(t_{k-1},t_k] \subset F$,
it holds that $\int \ind_{(t_{k-1},t_k]} \,\sigma_s^{-2} \,\de \langle M \rangle_s = t_k-t_{k-1}$
by the assumptions.
We can approximate the set $F$ from below by finite unions of this type of intervals
and use monotone convergence theorem to show that
\begin{equation}
\int_0^t \ind_F \,\sigma_s^{-2} \,\de \langle M \rangle_s = \int_0^t \ind_F  \,\de s
\end{equation}
where the left-hand side is a Lebesgue--Stieltjes integral and the right-hand side is
a Lebesgue integral defined pointwise in the randomness almost surely.

This shows that
$\ind_F \sigma_s^{-1}$ is $\de M_s$ integrable (belongs to the square integrable processes with respect to
the variation process of $M$) and $(B_t)$ is well-defined and continuous in $t$. 
Therefore clearly, $(B_t)$ is a local martingale and satisfies $\langle B \rangle_t = t$.
Hence $(B_t)$ is a standard Brownian motion by L\'evy's characterization theorem.
\end{proof}

\subsubsection{\texorpdfstring{$X_t=V_t - U_t$}{Xt=Vt-Ut} is a Bessel process}
\label{sssec: x is bessel proof}

%
\begin{proof}[Proof of Theorem~\ref{thm: x is bessel}]
The claim follows from the next theorem with $M_t  = \frac{1}{2} \sqrt{X_t}$, or written in the
other way $X_t  = 4 M_t^2$, and $\psi(x)= x^4$. Notice that if $C M_t^4$, where $C>0$ is a constant,
is a squared Bessel
process of dimension $\delta$, then $(\sqrt{C}/4) X_t$ is a Bessel process of dimension $\delta$.
Here $\delta = 3/2$ and $C=2 \delta = 3$, which implies that
$X_t$ is a Bessel process scaled by the constant $4/\sqrt{C} = \sqrt{16/3}$.
\end{proof}

\begin{theorem}
Let $(a,b) \subset \R$.
Suppose that $\psi: (a,b) \to \R$ is twice continuously differentiable function 
which is convex, i.e. $\psi'' \geq 0$, and such that $\psi'$ is strictly increasing.
Let $M=(M_t)_{t \in \R_+}$ be a continuous stochastic process adapted to a filtration $(\F_t)_{t \in \R_+}$ and
$N=(N_t)_{t \in \R_+}$ a process defined by $N_t = 2 \psi(M_t) - t$.
Suppose that $M$ and $N$ are martingales.
Then the following claims hold.
\begin{enumerate}
\item If the process $W=(W_t)_{t \in \R_+}$ is defined as
\begin{equation*}
W_t = \int_0^t \sqrt{\psi''(M_s)} \de M_s ,
\end{equation*}
then it is a standard Brownian motion.
\item If for some $\eps>0$, $\psi ( x ) = |x|^{2 + \eps}$, then there exists constants $C>0$ and $1<\delta<2$ 
such that $Z_t = C \, \psi(M_t)$ is a squared Bessel process of dimension $\delta$.
More specifically, the constants are  $\delta = 2 \frac{1+\eps}{2+\eps}$ and $C=2 \delta$.
\item 
Suppose there exists continuous functions $F,\tilde{F}$ such that $2\tilde{F}(x)=F(2 \psi(x))$
and $\psi'(x) = \sgn(\psi'(x)) \tilde{F}(x) \sqrt{\psi''(x)}$ for all $x$
--- in particular $\psi''$ is positive except possibly at the point (there exists at most one such point) where $\psi'$ is zero.
Then $Z_t = 2\psi(M_t)$ is a solution to the stochastic differential equation 
\begin{equation}
\de Z_t = F(Z_t) \de \tilde{W}_t + \de t
\end{equation}
for some standard Brownian motion $(\tilde{W}_t)_{t \in \R_+}$.
\end{enumerate}
\end{theorem}

\begin{proof}
(1) Similarly as in the previous section, we notice that we can do stochastic analysis with $M$, because
$M$ is a continuous martingale, see Chapter~2 of \cite{Durrett:1996wh}. 
The same argument,
using that $(M_t)_{t \in \R_+}$ and $(N_t)_{t \in \R_+}$ are martingales and  
that 
$N_t$ is given in terms of $M_t$ as $N_t = 2 \psi(M_t) - t$,
as above tells us
the process defined as
\begin{equation*}
W_t = \int_0^t \sqrt{\psi''(M_s)} \,\de M_s ,
\end{equation*}
is a continuous martingale with a variation process $\langle W \rangle_t = t$. 
Namely, by It\^o's lemma 
we have $\de N_t = (\psi''(M_t)\de \langle M \rangle_t - \de t) + 
2\psi'(M_t) \de M_t$ and thus by martingale property of $N_t$ the quantity
inside the first brackets has to vanish identically, and we can
apply the previous lemmas for the claim.
Hence 
by L\'evy's characterization theorem, it is a standard Brownian motion.

(2) When $\psi ( x ) = |x|^{2 + \eps}$, there is a constant $D>0$ such that
\begin{equation*}
\psi'(x) = D \sgn (x) \sqrt{ \psi''(x) \psi(x) }
\end{equation*}
Therefore $Z_t = 2 \tilde C \psi(M_t)$ satisfies
\begin{align*}
\de Z_t &= 2 C \psi'(M_t) \,\de M_t + \tilde C \psi''(M_t) \de \langle M \rangle_t \\
  &= 2 \tilde C\, D \sgn (M_t) \sqrt{ \psi(M_t) } \,\de W_t + \tilde C \de t \\
  &= 2\, \sqrt{\tilde C/2}\, D \sgn (M_t) \sqrt{ Z_t } \,\de W_t + \tilde C \de t .
\end{align*}
Hence if we choose $\tilde C= 2 \, D^{-2}$, $Z_t$ is 
squared Bessel process with the parameter $\delta = 2\, D^{-2}$.
Here we used the fact that $\int_0^t \sgn (M_s) \,\de W_s$ is  a standard Brownian motion
by L\'evy's characterization theorem. The claim follows for $C=2 \tilde C$.

A direct calculation shows that $D= \sqrt{(2+\eps)/(1+\eps)} \in  (1,\sqrt{2})$.

(3) Similarly as above we can write
\begin{align*}
\de Z_t 
  &= 2 \psi'(M_t)\,\de M_t + \psi''(M_t) \,\de \langle M \rangle_t \\
  &= 2  \sgn(\psi'(M_t)) \tilde{F}(M_t) \, \de W_t + \de t \\
  &= \sgn(\psi'(M_t)) F(Z_t) \, \de W_t + \de t 
\end{align*}
from which the claim follows.
\end{proof}

\begin{remark}
Note that when $1/\alpha=2+\eps$ and $\alpha = 1 - \delta/2$ and $D= \sqrt{(2+\eps)/(1+\eps)}$, 
then $2\, D^{-2} = \delta$.
\end{remark}

\subsection{Characterization of \texorpdfstring{$(U_t,V_t)$}{(Ut,Vt)}}\label{ssec: char of pair U V}

Next result together with Theorem~\ref{thm: x is bessel} gives the
distribution of the pair of processes $(U_t,V_t)$. 

\begin{theorem}\label{thm: u v loewner 4pt}
Let $X_t,V_t,U_t$ be as in Theorem~\ref{thm: x is bessel}.
Then $U_t$ and $V_t$ satisfy
\begin{align}
V_t       &= V_0 +2 \int_0^t \frac{\de s}{X_s}   \\
U_t       &= U_0 +2 \int_0^t \frac{\de s}{X_s} - X_t + X_0
\end{align}
\end{theorem}

\begin{remark}
This equation for $V_t$ is the Loewner equation. Notice that 
$\int_0^t \frac{\de s}{X_s} $ is finite since $(X_t)$ is Bessel process
of dimension $\delta = 3/2$.
The equation for $U_t$ is obtained from the one of $V_t$ and the definition of $X_t$.
\end{remark}

\begin{proof}[Proof of Theorem~\ref{thm: u v loewner 4pt}]
Since $\int_0^t X_s^{-1}\de s$ is finite, we have shown so far that
\begin{align}
V_t       &= V_0 +2 \int_0^t \frac{\de s}{X_s} + \Lambda_t \\
U_t       &= U_0 +2 \int_0^t \frac{\de s}{X_s} - X_t +X_0 + \Lambda_t
\end{align}
where $\Lambda_t$ is some non-decreasing process which is constant on any subinterval of $\{t \,:\, X_t>0\}$.
See Proposition~\ref{prop: rightmost point and increasing process} in Appendix~\ref{sec: appendix loewner}
for the generalized Loewner equation.
The claim follows when we show that
$\Lambda_t\equiv 0$.
This is shown in the next proposition.
\end{proof}

\begin{proposition}\label{prop: lambda vanishes}
$\Lambda_t \equiv 0$.
\end{proposition}

\begin{proof}
Let $\Sigma = \{t \,:\, X_t=0\}$.
The \emph{index} of a Bessel process is defined as
\begin{equation}
\nu = \frac{\delta}{2}-1 .
\end{equation}
Notice that $\nu=-1/4$ in our case when the Bessel process has dimension $3/2$.
By \cite{Revuz:1999wo} Exercise~XI.1.25 and \cite{Bertoin:1996tv} Theorem III.15, the Hausdorff dimension of the support
of the local time of a Bessel process with index $\nu \in [-1,0]$ is $-\nu$.
Hence the Hausdorff dimension $\hdim ( \Sigma ) = 1/4$.

Notice next that 
$(X_t)_{t \in [0,T]}$ is $1/2 - \eps$ H\"older continuous for any $\eps>0$,
since it is a Bessel process scaled by constant (this can be derived in several ways; for instance
starting from the density of the transition semigroup \cite[p.~446]{Revuz:1999wo}
and checking the asumptions of the Kolmogorov continuity theorem) 
and $(\int_0^t X_s^{-1} \de s)_{t \in [0,T]}$ 
is $1/2 - \eps$ H\"older continuous for any $\eps>0$,
since it can be written as 
\begin{equation*}
\int_0^t X_s^{-1} \de s = \frac{3}{4} \left(X_t - X_0 - \frac{4}{\sqrt{3}} B_t \right)
\end{equation*}
from \eqref{eq: bessel sde}.

The claim follows from Lemma~\ref{lm: holder and devils staircase} and from the fact that $(U_t)_{t \in [0,T]}$
is $1/2 - \eps$ H\"older continuous for any $\eps>0$ by Theorem~\ref{thm: ks}.
\end{proof}

\begin{lemma}\label{lm: holder and devils staircase}
Let $I \subset [0,S]$ be a closed set with Hausdorff dimension $\alpha_0 \in [0,1]$.
Suppose that $f: [0,S] \to \R$ is continuous and constant on any subinterval of $[0,S] \setminus I$.
If $f$ is not a constant function and if $f$ 
is H\"older continuous with exponent $\alpha$, then $\alpha \leq \alpha_0$.
\end{lemma}

The proof is standard using the definition of Hausdorff measure and
\begin{equation*}
|f(x) -f(y)| \leq C |x-y|^\alpha ,
\end{equation*}
and we leave it to the industrious reader.

\subsection{The martingale characterization in the 4-point case}\label{ssec: martingale char for 4-pt}

\subsubsection{The \textnormal{SLE}$[\kappa,Z]$ process}

The drift of a SLE process can be given in general 
in terms of a partition function $Z(u,v,\ldots)$
where $u,v,\ldots$ are the marked points of the process
(for instance, the chordal SLE has marked points $u, \infty$ and 
$\infty$ doesn't appear explicitly in $Z$, and the \kr{} has marked points $u,v,\infty$).
We call such process SLE$[\kappa,Z]$ and the driving process $U_t$ is given by
\begin{equation}\label{eq: partition fn general}
\de U_t = \sqrt{\kappa}\de B_t + \kappa \partial_u \log Z(U_t,V_t,\ldots) \de t .
\end{equation}
The other points follow the Loewner equation.

The partition function is not unique, 
but if we require M\"obius covariance and
finite limit as $w \to \infty$, 
then  SLE$[\kappa,Z]$ with $\kappa = 16/3$
\begin{equation} \label{eq: partition fn of FK}
Z(u,v,w)=
\frac{1}{y^{1/8}\,(m^2+1)^{1/4}\,m^{1/4}}
\end{equation}
where $y=w-v$, $m=\sqrt{1+y/x} - \sqrt{y/x}$ and $x=v-u$,
describes the scaling limit of FK Ising exploration in $(\Omega,a,b,c,d)$ from $a$ to $d$
which is reflected on $bc$ towards $d$. 
Notice that the process with the partition function~\eqref{eq: partition fn of FK}
and $\kappa = 16/3$ isn't a \kr{} process.

\begin{theorem}
In the $4$-point setting, the scaling limit of FK Ising exploration in $(\Omega,a,b,c,d)$ from $a$ to $d$
which is reflected on $bc$ towards $d$ is SLE$[\kappa,Z]$ with $\kappa = 16/3$
and the partition function given by \eqref{eq: partition fn of FK}.
\end{theorem}

This result follows from the estimates on H\"older regularity of the random curves
and the characterization result Theorem~\ref{thm: x,y is gen of bessel} below.
It is possible to use this result to show that 
the interface when conditioned on an internal arc pattern $\intarcp{a}{d}{b}{c}$
converges towards so called hypergeometric SLE
\cite{Kemppainen:CEhYg0we}.

\begin{theorem}\label{thm: x,y is gen of bessel}
Let $X_t = V_t - U_t$, $Y_t = W_t - V_t$ where $U_t$, $V_t$ and $W_t$ are the processes followed
by the marked points for the subsequent scaling limit of the FK Ising exploration process
in the $4$-point setting. 
Then for some Brownian motion $(B_t)_{t \geq 0}$ 
the pair of processes $(X_t, Y_t)_{t \geq 0}$ 
satisfies
\begin{align}
X_t &= X_0 + 
\frac{4}{\sqrt{3}} \, B_t 
  + \frac{1}{3} \, \int_0^t \frac{(3M_s^4+2M_s^2+1)(1-M_s^2)^2}{Y_s M_s^2 (M_s^2+1)^2} \de s 
  \label{eq: sde of X 4pt integral} \\
Y_t &= Y_0 -\int_0^t \frac{1}{2 Y_s} 
        \frac{\left(1-M_s^2\right)^4  }{  M_s^2  (1+M_s^2)^2  }  \de s 
        \label{eq: de of Y 4pt integral}
\end{align}
where
\begin{equation}
M_t^2 = \left( \sqrt{1 + \frac{Y_t}{X_t}} - \sqrt{\frac{Y_t}{X_t}} \right)^2 .
\end{equation}
Furthermore, the driving process $(U_t)_{t \geq 0}$ is recovered from 
$(Y_t)_{t \geq 0}$ by
\begin{equation}\label{eq: sde of U 4pt integral}
U_t  = U_0 +2 \int_0^t \frac{\de s}{X_s} - X_t + X_0
\end{equation}
\end{theorem}

\begin{remark}
Notice that \eqref{eq: sde of X 4pt integral} is a combination of Loewner equations
for $V_t$ and $W_t$ and \eqref{eq: sde of U 4pt integral} follows from
definition of $X_t$ and the Loewner equation for $V_t$.
Notice also that the given form of $Z$ in SLE$[\kappa,Z]$ follows from 
a direct calculation using the equation~\eqref{eq: partition fn general}.
\end{remark}

We work towards the proof of this result in next subsections.

\subsubsection{Simple martingales from the observable}

Let $U_t$ be the driving process, $V_t$ the point corresponding to $b_t$, and $W_t$ the point corresponding to $c$
then set
\begin{align}
X_t &= V_t - U_t \\
Y_t &= W_t - V_t .
\end{align}
Then given $\gamma[0,t]$ in $\half$, the map
\begin{equation}
z \mapsto \frac{g_t(z) - U_t}{W_t - U_t} 
\end{equation}
maps $H_t = \half \setminus \gamma[0,t]$ onto $\half$ and the marked points to $0$, $X_t/(X_t + Y_t)$ and $1$.
Therefore
\begin{equation}
\sqrt{-\pi}
  \, f^{\half,U_t,V_t,W_t} (z) = \sqrt{g_t'(z)} \, \sqrt{ 
  \frac{1}{g_t(z)-W_t} + \beta_t \left( \frac{1}{g_t(z)-U_t} - \frac{1}{g_t(z)-V_t} \right) } 
\end{equation}
where
\begin{equation}\label{eq: 4 point beta}
\beta_t = \beta\left(\frac{X_t}{X_t + Y_t}\right) = \left( \sqrt{1 + \frac{Y_t}{X_t}} - \sqrt{\frac{Y_t}{X_t}} \right)^2
\end{equation}

The quantity
\begin{equation}\label{eq: martingale 1}
M_t = \pm \sqrt{ \beta_t}
\end{equation}
is a conditional probability  
on $\F_t$ of the
event $\gamma_1 \subset \hat{\gamma}$ 
hence
$(M_t)_{t \geq 0}$ is a martingale. 
Here $\pm$ sign is needed to extend the martingale property beyond the hitting of $0$ by $M_t$.
We can solve $X_t$ in terms of $Y_t$ and $M_t$ as
\begin{equation}
X_t = Y_t \, \frac{4 M_t^2}{(1-M_t^2)^2}
\end{equation}

From the expansion of $f$ as $z \to \infty$, we get that
\begin{equation}\label{eq: martingale 2}
N_t = X_t \, M_t^2 - W_t
\end{equation}
is a martingale.

Write \eqref{eq: martingale 2} as
\begin{equation}\label{eq: martingale 2 var}
N_t = Y_t \psi(M_t) - W_t .
\end{equation}
where $\psi(m) = 4 [ m^2/(1-m^2 ) ]^2$. 
Note that $W_t$ is always differentiable and $Y_t$ is differentiable outside the set of times 
$\Sigma \dd= \{ t \,:\, X_t=0\} = \{ t \,:\, M_t=0\}$.
Now we have to solve the following \emph{martingale problem}:
\begin{quote}
for given filtered probability space $(\P^*,(\F_t))$,
determine the law of $(M_t,Y_t,W_t)_{t \in [0,T]}$ which satisfies
\begin{enumerate}
\item $(W_t)$ is strictly increasing and $C^1$ for all $t \in [0,T]$ 
\item $(Y_t)$ is strictly 
decreasing
and $C^1$ for all $t \in [0,T] \setminus \Sigma$
\item $\dot{W}_t$ is given by
\begin{equation}
  \dot{W}_t = \frac{2}{Y_t} \left(  \frac{1-M_t^2 }{1+M_t^2}   \right)^2  
\end{equation}
and for a non-decreasing process $\Lambda_t$ given 
by Proposition~\ref{prop: rightmost point and increasing process} in Appendix~\ref{sec: appendix loewner},
it holds that
\begin{equation}
    Y_t = Y_0 -\int_0^t \frac{1}{2 Y_s} 
        \frac{\left(1-M_s^2\right)^4  }{  M_s^2  (1+M_s^2)^2  }  \de s 
       - \Lambda_t . 
\end{equation}
\item $(M_t)$ and $(N_t)$ (defined by \eqref{eq: martingale 2 var}) are martingales.
\end{enumerate}
\end{quote}
It is understood that solving the martingale problem means, 
as before, that we claim that these properties uniquely 
characterize the law of $(M_t)$ and hence also the law of $(X_t)$ and also that
the law can be explicitly descibred.

\subsubsection{Solving the martingale problem}\label{sssec: solving mart prob 4pt}

By the lemmas of Section~\ref{ssec: mart prob lemmas}, we can construct out of $(M_t)$, which is a continuous martingale, 
a Brownian motion $(B_t)$ so that there exists a process $(\sigma_t)$, 
$\sigma_t \geq 0$,
both adapted to $(\F_t)$, such that
\begin{equation}
\langle M \rangle_t = \int_0^t \sigma_s^2 \,\de s, \qquad M_t = M_0 + \int_0^t \sigma_s \,\de B_s .
\end{equation}
Hence we know that there exists a $(\P^*,(\F_t))$ Brownian motion $(B_t)$ and
when $M_t \neq 0$, $\de M_t = \sigma_t \de B_t$.
Write the drift of the process $(N_t)$ using It\^o's lemma as
\begin{equation}
\sigma_t^2 \cdot \frac{24 \, Y_t \, M_t^2 (M_t^2+1)}{(M_t^2-1)^4}  - \frac{2(M_t^2-1)^2}{Y_t(M_t^2+1)} .
\end{equation}
Therefore we have to have
\begin{equation}
\sigma_t = \frac{1}{2 \sqrt{3}} \frac{(1-M_t^2)^3}{Y_t\,|M_t| (M_t^2+1)} .
\end{equation}
Hence the solution of the martingale problem is that
\begin{align}
M_t &= M_0 + \int_0^t \frac{1}{2 \sqrt{3}} \frac{(1-M_s^2)^3}{Y_s\,|M_s| (M_s^2+1)} \de B_s.
\end{align}
If we plug this into the expression of 
$X_t = (4 M_t^2)/(Y_t (M_t^2-1)^2)$ which can be solved from \eqref{eq: 4 point beta} 
and \eqref{eq: martingale 1},
we get
\begin{equation} \label{eq: sde of X}
\de X_t = \frac{4}{\sqrt{3}} \de B_t + \frac{1}{3} \frac{(3M_t^4+2M_t^2+1)(M_t^2-1)^2}{Y_t M_t^2 (M_t^2+1)^2} \de t 
\end{equation}
which is equivalent to \eqref{eq: sde of X 4pt integral}.

Notice that when $X_t/Y_t \to 0$ then
\begin{equation}\label{eq: M approx}
M_t \approx \pm \frac{1}{2} \sqrt{ \frac{X_t}{Y_t} }
\end{equation}
and \eqref{eq: sde of X} becomes
\begin{equation}\label{eq: X sde approx}
\de X_t \approx \frac{4}{\sqrt{3}} \de B_t + \frac{4}{3} \frac{\de t}{X_t} .
\end{equation}
which corresponds to the Bessel process with dimension $\delta=\frac{3}{2}$, as it should.

In the same way,
when $X_t/Y_t \to 0$
\begin{equation}
\sigma_t \approx \frac{1}{\sqrt{3}} \frac{1}{\sqrt{X_t Y_t}}
\end{equation}
Hence $\sigma_t^2$ is integrable with respect to $\de t$ in the same sense as $X_t^{-1}$.

%
\begin{lemma}\label{lm: X and int 1 over X holder}
Processes $(X_t)_{t \in [0,\infty)}$ and $(\int_0^t X_s^{-1} \de s)_{t \in [0,\infty)}$
are Hölder continuous with any exponent less than $\frac{1}{2}$.
\end{lemma}

\begin{proof}
This claim follows from comparison to Bessel process~\eqref{eq: M approx}
and \eqref{eq: X sde approx} and a similar argument as we used
for a similar claim in the proof of 
Proposition~\ref{prop: lambda vanishes}.
\end{proof}

\subsubsection{Characterization of $(U_t,V_t,W_t)_{t \in [0,T]}$}

\begin{proof}[Proof of Theorem~\ref{thm: x,y is gen of bessel}]
By comparing to a Bessel process as we did in Section~\ref{sssec: solving mart prob 4pt},
we see that $\int_0^t X_s^{-1}\de s$ is finite. 
Therefore we have shown so far that
\begin{align}
V_t       &= V_0 +2 \int_0^t \frac{\de s}{X_s} + \Lambda_t \\
U_t       &= V_0 +2 \int_0^t \frac{\de s}{X_s} - X_t + \Lambda_t
\end{align}
where $\Lambda_t$ is some non-decreasing process which is constant on any subinterval of $\{t \,:\, X_t>0\}$.
It remains to show that $\Lambda_t \equiv 0$.
This implies then that \eqref{eq: de of Y 4pt integral} and \eqref{eq: sde of U 4pt integral} hold
and finalizes the proof of Theorem~\ref{thm: x,y is gen of bessel}.

Recall from the proof of Proposition~\ref{prop: lambda vanishes} 
that $\Sigma = \{t \,:\, X_t=0\}$ and that
the index of a Bessel process is defined as
\begin{equation}
\nu = \frac{\delta}{2}-1 .
\end{equation}
Notice that $\nu=-1/4$ in our case. 

Now $\Lambda_t \equiv 0$ follows from Lemma~\ref{lm: holder and devils staircase},
Lemma~\ref{lem: hdim sigma}  below,
which shows that the Hausdorff dimension of $\Sigma $ is $1/4$,
and from the fact that $(U_t)_{t \in [0,T]}$
is $1/2 - \eps$ H\"older continuous for any $\eps>0$ by Lemma~\ref{lm: X and int 1 over X holder}.
\end{proof}

\begin{lemma}\label{lem: hdim sigma}
$\hdim (\Sigma) = 1/4$
\end{lemma}

\begin{proof}
For fixed $\eps>0$ choose $\eps_1>0$ such that
\begin{gather}
\frac{1-\eps}{3} \frac{1}{m^2} \leq \frac{1}{3} \frac{(3 m^4+2 m^2+1)(m^2-1)^2}{m^2 (m^2+1)^2} \leq \frac{1}{3} \frac{1}{m^2} \\
\frac{1}{2 (1+\eps)^{1/2} } \sqrt{x}  \leq \sqrt{1 + \frac{1}{x}} - \sqrt{\frac{1}{x}} \leq \frac{1}{2} \sqrt{x}
\end{gather}
for all $0< m \leq \sqrt{\eps_1}/2$ and $0< x \leq \eps_1$.
Then whenever $t$ is such that $X_t/Y_t \leq \eps_1$, it is possible to estimate
\begin{equation}
\frac{4(1-\eps)}{3} \frac{1}{X_t} \leq (\text{drift of } X_t) \leq \frac{4(1+\eps)}{3} \frac{1}{X_t}
\end{equation}
and hence we can couple $(X_t)_{t \in [\tau,\tilde{\tau}]}$, where
$\tau$ is any stopping time such that $X_\tau \leq \eps_1/2$ and $\tilde{\tau}=\inf\{t \geq \tau \,:\, X_t=\eps_1\}$,
to (scaled) Bessel processes $X_t'$ and $X_t''$ which satisfy
\begin{align}
\de X_t'  &= \frac{4}{\sqrt{3}} \de B_t + \frac{4(1-\eps)}{3} \frac{\de t}{X_t'} \\
\de X_t'' &= \frac{4}{\sqrt{3}} \de B_t + \frac{4(1+\eps)}{3} \frac{\de t}{X_t''}
\end{align}
and $X_\tau =X_\tau'=X_\tau''$. Under this coupling $X_t' \leq X_t \leq X_t''$ for any $t \in [\tau,\tilde{\tau}]$. 

Clearly
\begin{equation}
\{t \,:\, X_t''=0\} \subset \{t \,:\, X_t=0\} \subset \{t \,:\, X_t'=0\} .
\end{equation}
By \cite{Revuz:1999wo} Exercise~XI.1.25 and \cite{Bertoin:1996tv} Theorem III.15, the Hausdorff dimension of the support
of the local time of a Bessel process with index $\nu \in [-1,0]$ is $-\nu$. By Markov property of the Bessel
process, the support of the local time is the entire set of times when the process is at the origin.
Hence the set $\Sigma$ is sandwiched between two sets which have dimension arbitrarily close to $1/4$.  
\end{proof}


\section{Proof of Theorem~\ref{thm: main theorem}}\label{sec: proof main theorem}

As a conclusion to this entire article we outline below the proof of Theorem~\ref{thm: main theorem}.

As above, we consider a sequence of domains $\Omega_{\delta_n}$ converging to
a domain $\Omega$ with respect to a fixed interior point $w_0$.
Let $\phi_n$ be the  conformal map from $\Omega_{\delta_n}$ to $\disc$
normalized using $w_0$.
Suppose also that we have a sequence of boundary points
$v_{\delta_n} \in \partial \Omega_{\delta_n}$ 
that converge in the sense $\phi_n(v_{\delta_n})$ converges to a point on
$\partial \disc$. We use $v_{\delta_n} $ as the root point in the
construction of the exploration tree.

The basic crossing estimates were established for the FK Ising exploration tree and its branches in Section~\ref{ssec: tree crossing},
Theorem~\ref{thm: FK Ising and condition}.
Based on those estimates, the precompactness of probability laws of a single branch or a finite subtree 
(i.e., a subtree with a fixed number of target points) 
was shown in Section~\ref{ssec: a priori single branch}. By those results we can choose convergent subsequences.
The structure of the tree is characterized by the target independence, the independence of the branches after disconnection and
the martingale characterization of a single branch in Section~\ref{sec: characterization}.
Here it is needed that the branches converge in the strong sense as capacity parameterized curves.

By these results, every sequence of finite subtrees of the approximating domains converges in distribution to a finite subtree of 
the SLE$(\kappa, \kappa -6)$ exploration tree with $\kappa=16/3$. The convergence takes place, for instance, in the
unit disc $\disc$ after a conformal transformation 
and under the metric defined in Section~\ref{sssec: topo trees}. In fact, it is possible to extend this
convergence to the original domain (without the conformal transformation). See Corollary~1.8 in \cite{Kemppainen:2012vm}
for such a result.

The precompactness of the probability laws of the full tree and of the loop collection were established in
Theorems~\ref{thm: aizenman--burchard bounds} and \ref{thm: precompactness loop collections}, respectively.
Therefore we can choose subsequences such that both the full tree and the loop collection converge.
The above together with the finite-tree approximation in Theorem~\ref{thm: uniform tree approximation} implies that
the full tree has a unique limit which is the SLE$(\kappa, \kappa -6)$ exploration tree with $\kappa=16/3$.
Similarly using the finite-tree approximation in Theorem~\ref{thm: uniform loop approximation}
we show that the loop collection has a unique limit 
which is characterized 
by the one-to-one correspondence to
the exploration tree under the maps introduced in Sections~\ref{sssec: exploration tree} and
\ref{ssec: exploration tree definition}, see Theorem~\ref{thm:  tree to loops in the limit}.
This ends the outline of the proof. 
Notice that the convergence takes place in a topology where
both the tree and the loop ensemble converge simultaneously and convergence
occurs for the full objects, not just the finite tree approximations
(with only a finite fixed number of target points).

Since the orientation of the image of the exploration tree when mapped to $\disc$
depends on $\lim_n \phi_n(v_{\delta_n})$, while the scaling limit
of the loop collection is the same for any sequence $v_{\delta_n}$. 
This gives the rotational invariance of the loop collection and thus implies
the complete conformal invariance of the scaling limit.

\subsection*{Acknowledgements}

A.K. was financially supported by the Academy of Finland. 
Both authors were supported by the Swiss NSF, NCCR SwissMAP and ERC AG COMPASP.
A part of this research was done during A.K.'s participation in programs 
``Conformal Geometry'' at the Simons Center for Geometry and Physics, Stony Brook University,
and ``Random Geometry'' at the Isaac Newton Institute for Mathematical Sciences, University of Cambridge.
We thank the anonymous referees for their valuable comments.

\appendix

\section{Distortion of annuli under conformal maps}\label{sec: appendix annulus}

Let $\domain$ be a simply connected domain and $\phi: \disc \to \domain$ is a conformal map.
The next two lemmas establish bounds for distortion of annuli in terms of their
conformal modulus which is a constant times $\log(R/r)$ for an annulus $A=A(z_0,r,R)$.
The proofs can be found from Section~6.3.3 of \cite{Kemppainen:tb}.
The results are needed in the proof of Theorem~\ref{thm: aizenman--burchard bounds}
above.

\begin{lemma}[Distortion of annuli contained in $\disc$]\label{lem: dist ann bulk}
For any $\rho'>1$ sufficiently large, there exists $\rho>1$ 
independent of $\Omega$ and $\phi$
such that the following holds. Suppose 
that the annulus $A=A(z_0,r,R) \subset \disc$ is such that $R/r > \rho$ and $R\leq \frac{1}{2} (1-|z_0|)$. 
Then there exists an annulus $A'=A(z_0',r',R')$  with $R'/r' > \rho'$
such that $A' \subset \phi(A)  $ and $A' $ separates the components of $\phi(A)$ in $\C$.
Furthermore the dependency of $\rho$ and $\rho'$ can be made linear.
\end{lemma}

\begin{lemma}[Distortion of annuli not fully contained in $\disc$]\label{lem: dist ann boundary}
For any $\rho'>1$ sufficiently large, there exists $\rho>1$ 
independent of $\Omega$ and $\phi$
such that the following holds.
Suppose that the annulus $A=A(z_0,r,R)$ is such that  $R/r > \rho$,
$1-|z_0|<r$ and $R<1$ (that is, $\partial \disc$ crosses $A$). 
Then there exists an annulus $A'=A(z_0',r',R')$  with $R'/r' > \rho'$
such that there exists a connected component $O$ of $\domain \cap A'$
such that $O \subset \phi(A \cap \disc)$ and $O$ separates 
the components of $\phi((\partial A) \cap \disc)$ in $\domain$.
\end{lemma}

\section{Auxiliary results on a priori bounds}\label{sec: appendix a priori}

Denote conformal images of the FK Ising boundary loop exploration tree and the random cluster configuration
under $\phi^{-1}:\domain \to \disc$ by $\todisc{\therndttr}$ and $\todisc{\omega}$, respectively.
The following auxiliary results on the relation of the number of crossings by disjoint segments of tree branches
and the number of disjoint open or dual-open crossings could be written in the original domain or any other
``reference'' domain. The unit disc is nice in the sense that 
the annulus has two sides on the boundary and thus the number $n-1$ appears 
in the formulation of the second lemma instead of a smaller number needed when the boundary makes more crossings
of the annulus.
The results are needed in the proof of Theorem~\ref{thm: aizenman--burchard bounds}
above.

\begin{lemma}\label{lem: crossings and arms bulk}
Suppose that $A=A(z_0,r,R) \subset \disc$.
If there are at least $2n$ disjoint segments of the tree $\todisc{\therndttr}$ crossing $A$,
then there are 
at least
$n$ dual-open arms in the random cluster configuration $\todisc{\omega}$.
\end{lemma}

\begin{proof}
Any crossing of $A$ by the tree $\todisc{\therndttr}$, when $A$ is contained in the domain,
is a subcurve of a single loop.
Thus its left-hand side is entirely open and the right-hand side entirely dual-open.
Due to topological reasons, any pair of adjacent minimal crossings (by the tree) 
have common type of boundary in the region between them.
Thus in every second pair of adjacent minimal crossings
have to have a dual-open crossing between them.
\end{proof}

\begin{lemma}\label{lem: crossings and arms boundary}
Suppose that $A=A(z_0,r,R)$ such that $B(z_0,r) \cap \partial\disc \neq \emptyset$.
If there are at least $2n$ disjoint segments of the tree $\todisc{\therndttr}$ crossing $A$,
then there are 
at least
$n-1$ dual-open arms in the random cluster configuration $\todisc{\omega}$.
\end{lemma}

\begin{proof}
Let $k$ be the number of minimal crossings of $A$
which touch the boundary
so that their dual-open right-hand side touches free boundary (i.e. the vertex sets of the dual-open path
and the dual boundary intersect) and that along each them there is at least one branching point.
Then due to topological reasons $k$ is at most $2$.
Let $m$ be the total number crossings.
Then $m-k$ crossings are
disjoint from the boundary or touch the boundary by their open left-hand sides.
In particular, for these types of crossings, the crossings consists only of a part following a single loop.
Thus the entire right-hand side is dual-open.
Similarly as in the previous proof, we see that there are at least
$\lfloor(m-k)/2\rfloor$ dual-open crossings of $A$.
\end{proof}

\section{Auxiliary results on Loewner evolutions}\label{sec: appendix loewner}

The results of this section are needed in Section~\ref{ssec: char of pair U V}
above.

For a Loewner chain $(K_t)_{t \in [0,T]}$ in $\half$ driven by $(U_t)_{t \in [0,T]}$,
let $(V_t)_{t \in [0,T]}$ be the function defined by
\begin{equation*}
V_t = g_t( \sup( K_t \cap \R ) ) .
\end{equation*}
The point $V_t$ is thus the image of the rightmost point on the hull under the conformal map $g_t$.
The following lemma can be extracted from the proof of Proposition~3.12 in \cite{Sheffield:2009}.

\begin{lemma}\label{lm: leb u equals v}
For any $(U_t)_{t \in [0,T]}$ and $(V_t)_{t \in [0,T]}$ as above,
it holds that $\int_0^T \ind_{V_t=U_t} \de t = 0$. 
\end{lemma}

\begin{proposition}\label{prop: rightmost point and increasing process}
Suppose that $\int_0^T \frac{\de t}{V_t - U_t} < \infty$.
There exists a non-decreasing function $(\Lambda_t)_{t \in [0,T]}$ with $\Lambda_0=0$ such that
\begin{equation}\label{eq: prop rightmost point and increasing process}
V_t = V_0 + \int_0^t \frac{2 \de s}{V_s - U_s} + \Lambda_t .
\end{equation}
\end{proposition}

\begin{proof}
Let $\eps>0$,
$J_t^\eps = \min \{ k \eps > \sup( K_t \cap \R ) \,:\, k \in \Z \}$ and $\tilde V_t^\eps = g_t(J_t^\eps)$.
Then it follows from monotonicity of $g_t$ that $V_t \leq V_t^\eps \leq V_t + \eps$.
By the Loewner equation we can write
\begin{equation}\label{eq: prop rightmost point and increasing process auxiliary process}
\tilde V_t^\eps = \tilde V_0^\eps + \int_0^t \frac{2 \de s}{\tilde V_s^\eps - U_s} + \sum_{s \leq t} \xi^\eps_s
\end{equation}
where $\xi^\eps_s \in [0,\eps]$ and $\xi^\eps_s \neq 0$ only when $\tilde V_u^\eps - U_u$ hits zero as $u \nearrow s$.
Thus the sum on the right-hand side is a finite sum. By Lemma~\ref{lm: leb u equals v}  
and Lebesgue's dominated convergence theorem
the middle term in \eqref{eq: prop rightmost point and increasing process auxiliary process}
converges to the integral in \eqref{eq: prop rightmost point and increasing process}.
Since also $\tilde V_t^\eps$ converges uniformly to $V_t$, it holds that
$\sum_{s \leq t} \xi^\eps_s$ converges uniformly to some continuous function $\Lambda_t$ as $\eps \to 0$.
It follows that $\Lambda_t$ is non-decreasing and $\Lambda_0=0$. The claim follows.
\end{proof}

\begin{remark}
Notice that in fact, $\int_0^T \frac{\de t}{V_t - U_t} < \infty$ always. Namely,
\begin{align*}
\int_0^T \frac{\de t}{V_t - U_t + \eps} &\leq \int_0^T \frac{\de t}{\tilde V_t^\eps - U_t} \leq
\frac{1}{2} \left(\tilde V_T^\eps - \tilde V_0^\eps  \right) 
 \leq \frac{1}{2} \left( V_T -  V_0 +\eps  \right) .
\end{align*}
Thus $\int_0^T \frac{\de t}{V_t - U_t} < \infty$ follows from Lebesgue's monotone convergence theorem.
\end{remark}

\nocite{*}

\bibliographystyle{habbrv}
\bibliography{bndry_touching_loops.bib}

\end{document}